%% file: 00-main.tex
\documentclass[11pt]{article}

\usepackage{comment}

\usepackage[margin=1in]{geometry}

\usepackage{rotating}
\usepackage{makecell}
\usepackage{natbib}
\usepackage{booktabs}
\usepackage{multirow}
\usepackage{pgfplots}

\usepackage{tcolorbox}

\newcolumntype{C}[1]{>{\centering\arraybackslash}p{#1}}

\usepackage[multiple]{footmisc}

\usepackage{float}
\usepackage{comment}
\usepackage{soul}
\usepackage{wrapfig}

\usepackage{amsmath,amsthm,amssymb,mathtools}

\usepackage{graphicx}
\usepackage{dsfont}
\usepackage{tikz}
\usetikzlibrary{calc}
\usetikzlibrary{shapes.geometric}
\usepackage{paralist}

\usepackage{wrapfig}

\usepackage{float}
\usepackage{comment}
\usepackage{algorithm}
\usepackage[noend]{algpseudocode}
\usepackage{algpseudocodex}

\usepackage{thmtools}
\usepackage{thm-restate}

\usepackage{caption}
\usepackage[labelformat=simple]{subcaption}

\usepackage{aliascnt}

\usepackage{xcolor}

\usepackage[backref=page]{hyperref}
\hypersetup{colorlinks=true,citecolor=blue,linkcolor=blue}
\hypersetup{colorlinks=true}
\hypersetup{linkcolor=[rgb]{.7,0,0}}
\hypersetup{citecolor=[rgb]{0,.7,0}}
\hypersetup{urlcolor=[rgb]{.7,0,.7}}

\makeatletter
\newcommand{\skipitems}[1]{%
  \addtocounter{\@enumctr}{#1}%
}
\makeatother

\declaretheorem[numberwithin=section]{theorem}
\declaretheorem[sibling=theorem, style=definition]{definition}
\declaretheorem[sibling=theorem]{lemma}
\declaretheorem[sibling=theorem]{claim}

\declaretheorem[sibling=theorem]{corollary}
\declaretheorem[sibling=theorem, style=definition]{remark}
\declaretheorem[sibling=theorem]{proposition}
\declaretheorem[sibling=theorem, style=definition]{example}

\makeatletter
\patchcmd{\ALG@step}{\addtocounter{ALG@line}{1}}{\refstepcounter{ALG@line}}{}{}
\newcommand{\ALG@lineautorefname}{Line}
\makeatother

\newcommand{\R}{\mathbb{R}}

\newcommand{\N}{\mathbb{N}}

\DeclareMathOperator*{\argmax}{arg\,max}

\newcommand{\Pl}[1]{\mathds{P}\left[{#1}\right]}

\newcommand{\1}[1]{\mathds{1}\left[{#1}\right]}

\makeatletter
\newcommand\Ps@textstyle[2]{\mathbb{P}_{#1}\left[{#2}\right]}
\newcommand\Es@textstyle[2]{\mathbb{E}_{#1}\left[{#2}\right]}
\newcommand\Ps[2]{%
  \mathchoice %
  {\underset{{#1}}{\mathbb{P}}\left[{#2}\right]}
  {\Ps@textstyle{#1}{#2}}
  {\Ps@textstyle{#1}{#2}}
  {\Ps@textstyle{#1}{#2}}
}
\newcommand\Es[2]{%
  \mathchoice %
  {\underset{{#1}}{\mathbb{E}}\left[{#2}\right]}
  {\Es@textstyle{#1}{#2}}{\Es@textstyle{#1}{#2}}{\Es@textstyle{#1}{#2}}
}
\makeatother

\algdef{SE}[DOWHILE]{Do}{doWhile}{\algorithmicdo}[1]{\algorithmicwhile\ #1}%

\usepackage{comment}

\usepackage{sidecap}

\usepackage{csquotes}
\usepackage{colortbl}

\usepackage{tikz}
\usetikzlibrary{automata, positioning, calc, fit, shapes.misc}
\usetikzlibrary{matrix}
\usetikzlibrary{positioning}
\usetikzlibrary{arrows, decorations.pathmorphing}
\usetikzlibrary{decorations.pathreplacing,calligraphy}

\newcommand{\D}{\mathcal{D}}

\newcommand{\T}{\mathcal{T}}

\newcommand{\Alloc}{\mathcal{X}}

\usepackage{stmaryrd}

\newcommand{\coal}{\mathsf{coal}}

\newcommand{\adv}{a_{\mathsf{mi}}}

\newcommand{\bBuild}{\mathcal{B}}
\newcommand{\onCG}{\mathcal{C}}
\newcommand{\offCG}{\mathcal{D}}

\newcommand{\mi}{{\mathsf{mi}}}

\newcommand{\usr}{{\mathsf{usr}}}

\usepackage{pifont}%
\newcommand{\cmark}{\ding{51}}%
\newcommand{\xmark}{\ding{55}}%

\newcommand{\Distr}{\mathcal{T}}
\newcommand{\RegDist}{\mathsf{T}}

\newcommand{\supp}{\mathsf{supp}}
\newcommand{\vv}{\varphi}
\newcommand{\Threshold}{\Phi}
\newcommand{\MargThreshold}{\Delta \Threshold}

\newcommand{\AllocRule}{X}
\newcommand{\Pay}{P}
\newcommand{\Rev}{\mathsf{Rev}}
\newcommand{\Burn}{\mathsf{Burn}}
\newcommand{\BurnB}{B}

\makeatletter
\let\oldtheequation\theequation
\renewcommand\tagform@[1]{\maketag@@@{\ignorespaces#1\unskip\@@italiccorr}}
\renewcommand\theequation{(\oldtheequation)}
\makeatother

\makeatletter
\let\save@mathaccent\mathaccent
\newcommand*\if@single[3]{%
  \setbox0\hbox{${\mathaccent"0362{#1}}^H$}%
  \setbox2\hbox{${\mathaccent"0362{\kern0pt#1}}^H$}%
  \ifdim\ht0=\ht2 #3\else #2\fi
  }
\newcommand*\rel@kern[1]{\kern#1\dimexpr\macc@kerna}
\newcommand*\widebar[1]{\@ifnextchar^{{\wide@bar{#1}{0}}}{\wide@bar{#1}{1}}}
\newcommand*\wide@bar[2]{\if@single{#1}{\wide@bar@{#1}{#2}{1}}{\wide@bar@{#1}{#2}{2}}}
\newcommand*\wide@bar@[3]{%
  \begingroup
  \def\mathaccent##1##2{%
    \let\mathaccent\save@mathaccent
    \if#32 \let\macc@nucleus\first@char \fi
    \setbox\z@\hbox{$\macc@style{\macc@nucleus}_{}$}%
    \setbox\tw@\hbox{$\macc@style{\macc@nucleus}{}_{}$}%
    \dimen@\wd\tw@
    \advance\dimen@-\wd\z@
    \divide\dimen@ 3
    \@tempdima\wd\tw@
    \advance\@tempdima-\scriptspace
    \divide\@tempdima 10
    \advance\dimen@-\@tempdima
    \ifdim\dimen@>\z@ \dimen@0pt\fi
    \rel@kern{0.6}\kern-\dimen@
    \if#31
      \overline{\rel@kern{-0.6}\kern\dimen@\macc@nucleus\rel@kern{0.4}\kern\dimen@}%
      \advance\dimen@0.4\dimexpr\macc@kerna
      \let\final@kern#2%
      \ifdim\dimen@<\z@ \let\final@kern1\fi
      \if\final@kern1 \kern-\dimen@\fi
    \else
      \overline{\rel@kern{-0.6}\kern\dimen@#1}%
    \fi
  }%
  \macc@depth\@ne
  \let\math@bgroup\@empty \let\math@egroup\macc@set@skewchar
  \mathsurround\z@ \frozen@everymath{\mathgroup\macc@group\relax}%
  \macc@set@skewchar\relax
  \let\mathaccentV\macc@nested@a
  \if#31
    \macc@nested@a\relax111{#1}%
  \else
    \def\gobble@till@marker##1\endmarker{}%
    \futurelet\first@char\gobble@till@marker#1\endmarker
    \ifcat\noexpand\first@char A\else
      \def\first@char{}%
    \fi
    \macc@nested@a\relax111{\first@char}%
  \fi
  \endgroup
}
\makeatother

\usepackage{etoolbox}
\makeatletter
\patchcmd{\hyper@makecurrent}{%
    \ifx\Hy@param\Hy@chapterstring
        \let\Hy@param\Hy@chapapp
    \fi
}{%
    \iftoggle{inappendix}{%
        \@checkappendixparam{chapter}%
        \@checkappendixparam{section}%
        \@checkappendixparam{subsection}%
        \@checkappendixparam{subsubsection}%
        \@checkappendixparam{paragraph}%
        \@checkappendixparam{subparagraph}%
    }{}%
}{}{\errmessage{failed to patch}}

\newcommand*{\@checkappendixparam}[1]{%
    \def\@checkappendixparamtmp{#1}%
    \ifx\Hy@param\@checkappendixparamtmp
        \let\Hy@param\Hy@appendixstring
    \fi
}
\makeatletter

\newtoggle{inappendix}
\togglefalse{inappendix}

\apptocmd{\appendix}{\toggletrue{inappendix}}{}{\errmessage{failed to patch}}

\newcommand{\ICSet}{\Gamma}

\newcommand{\OnCAlloc}{\overline{X}}
\newcommand{\OnCPay}{\overline{P}}

\newcommand{\OnCBurn}{\overline{B}}
\newcommand{\OnCUtil}{\overline{U}}
\newcommand{\DirReveal}{(\OnCAlloc, \OnCPay, \OnCBurn)}

\newcommand{\QPay}{Q}

\newcommand{\seq}[2]{\big(#1\big)_{#2}}

\newcommand{\Moff}{\mathcal{M}_{\mathsf{off}}}
\newcommand{\Feasibility}{\mathcal{F}_{\bBuild}}

\newcommand{\FeasAlloc}{\mathcal{X}}
\newcommand{\Outcomes}{\mathcal{Y}}
\newcommand{\TypeDistr}{\mathcal{T}}

\newcommand{\vectr}[1]{\vec{{#1}}}

\begin{document}

\title{Characterizing Off-Chain Influence Proof Transaction Fee Mechanisms}

\author{
  Aadityan Ganesh\thanks{Princeton University | \emph{E-mail}: \href{mailto:}{aadityanganesh@princeton.edu}. Supported by an Ethereum Foundation Academic Grant.} 
  \and
  Clayton Thomas\thanks{Yale University | \emph{E-mail}: \href{mailto:}{thomas.clay95@gmail.com}.}
  \and
  S. Matthew Weinberg\thanks{Princeton University | \emph{E-mail}: \href{mailto:}{smweinberg@princeton.edu}. Supported by an Ethereum Foundation Academic Grant and NSF CAREER Award CCF-1942497.}
}
\date{}

\begin{titlepage}
\maketitle
\begin{abstract}

\input{01-abstract.tex}

\end{abstract}
\thispagestyle{empty}
\end{titlepage}

\maketitle

\tableofcontents
\thispagestyle{empty}

\clearpage
\pagenumbering{arabic}

\input{10-intro.tex}

\input{20-prelims.tex}

\input{30-postedprice}

\input{40-reduction}

\input{50-priorindependentimpossibilitiesandresults}

\input{60-priordependentimpossibilitiesandresults}

\input{61-priordependentconverse}

\input{62-deterministicmechanisms}

\input{63-positionauctionwithburns}
\input{64-generalizedpositionauctions}

\input{70-priordependentpostedprice}

 \bibliographystyle{apalike-three-authors}
\bibliography{Bib}

\appendix

\addtocontents{toc}{\protect\setcounter{tocdepth}{2}}%

\input{91-appendix-extra-content}

\input{92-appendix-omitted-proofs}

\end{document}

%% file: 01-abstract.tex
\citet{Roughgarden20} initiates the study of Transaction Fee Mechanisms (TFMs), and posits that the on-chain game of a ``good'' TFM should be on-chain simple (OnC-S), i.e., incentive compatible for both the users and the miner. 
Recent work of \citet{GaneshTW24} posit that they should additionally be Off-Chain Influence-Proof (OffC-IP), which means that the miner cannot achieve any additional revenue by separately conducting an off-chain auction to determine on-chain inclusion.
They observe that a cryptographic second-price auction satisfies both properties, but leave open the question of whether other mechanisms (such as those not dependent on cryptography) satisfy these properties.

In this paper, we characterize OffC-IP TFMs: They are those satisfying a \emph{burn identity} relating the burn rule to the allocation rule. 
In particular, we show that auction is OffC-IP if and only if its (induced direct-revelation) allocation rule $\widebar{X}(\cdot)$ and burn rule $\widebar{B}(\cdot)$ (both of which take as input users' values $v_1, \dots, v_n$) are truthful when viewing $\bigl(\widebar{X}(\cdot), \widebar{B}(\cdot)\bigr)$ as the allocation and pricing rule of a multi-item auction for a single additive buyer with values $\bigl(\varphi(v_1),\ldots, \varphi(v_n)\bigr)$ equal to the users' virtual values.

Building on this burn identity, we characterize OffC-IP and OnC-S TFMs that are deterministic and do not use cryptography: They are posted-price mechanisms with specially-tuned burns.
As a corollary, we show that such TFMs can only exist with infinite supply and prior-dependence.
However, we show that for \emph{randomized} TFMs, there are additional OnC-S and OffC-IP auctions that do not use cryptography (even when there is {finite} supply, under prior-dependence with a bounded prior distribution).
Holistically, our results show that while OffC-IP is a fairly stringent requirement, families of OffC-IP mechanisms can be found for a variety of settings.

%% file: 10-intro.tex
\section{Introduction}

In increasingly-important and high-demand blockchain applications, miners use Transaction Fee Mechanisms (TFMs) to allocate block space based on user-submitted bids.
TFMs are subject to a host of novel auction design constraints in this distributed trustless setting.
Recent work of \citet{GaneshTW24} highlights in particular the need for TFMs to be resistant to miners conducting an off-chain auctions in order to determine behavior on-chain, and call such a TFM off-chain influence-proof (OffC-IP).
They show that EIP-1559 (the existing state-of-the-art TFM) is not OffC-IP, but that a cryptographic variant of a second-price auction satisfies OffC-IP along with other desirable simplicity properties.
However, they leave open the question of precisely which auctions satisfy OffC-IP; for instance, whether there are any practical OffC-IP TFMs that do not use cryptography.

In this paper, we characterize OffC-IP TFMs, and delineate precisely how this property is compatible with other desirable properties considered by prior work. 
Our findings specify how the burn rule of the TFM is uniquely determined by its allocation rule, and highlight posted-price mechanisms as the \emph{unique} deterministic TFMs satisfying our desirable desiderata while avoiding cryptography.

\paragraph{TFM Design Thus-Far.}

Originally, blockchains such as bitcoin used a first-price auction as their TFM, requiring strategic bidding from the users.
Interest grew in simplifying the TFM for users, leading to a new proposed TFM for the Ethereum blockchain known as EIP-1559.
Motivated by these concerns, \citet{Roughgarden20, Roughgarden21} studies the theory of TFM design, and proposes three desiderata for a desirable TFM.
First, the on-chain TFM should be ``simple for users'' (OnC-US). 
That is, \emph{assuming the miner honestly implements the proposed TFM}, it should be a dominant strategy for users to simply bid their value. 
Second, the on-chain TFM should be ``simple for miners'' (OnC-MS).
That is, \emph{assuming users simply bid their value},  the miner should have no incentive to deviate from the proposed TFM (e.g., censor users or include shill bids), even after seeing all users' bids.
Third, the TFM should be ``robust to collusion between the miner and users'' (SCP). That is, the miner and users together cannot jointly profit by submitting anything other than their true bids.\footnote{
    Multiple variants of ``robust to collusions'' have been proposed.
    See \citet{ChungRS24,GafniY24} and \citet{GaneshTW24} for discussions.}
\citet{Roughgarden21} shows that EIP-1559 satisfies all three desirable properties.

EIP-1559 is essentially a `posted-price and burn' mechanism. The protocol exogenously sets a `base fee' $p$. Any user who wishes to include their transaction can pay $p$ to do so, and their payment is \emph{burned} (that is, the miner receives no revenue).\footnote{
    This describes EIP-1559 in the `ideal' case where the base fee is high enough to reduce demand below the maximum block size. In case the base fee is too low, EIP-1559 devolves into a first-price auction.}
Prior analysis suggests EIP-1559 is a very strategically-simple mechanism.
For instance, since there is no amount of censoring or fabricating bids that can leave the miner with a positive revenue, EIP-1559 is ``simple for miners''.

Despite EIP-1559's theoretical appeal and practical success, recent work of~\citet{GaneshTW24} observes the following challenge.
Suppose the miner publicly threatens to censor all users who do not pay her an entry fee of \$5 off-chain.\footnote{
  In fact, in the real deployment of EIP-1559 in Ethereum, users are able to submit an optional on-chain tip to the miner, allowing miners to conduct this attack entirely on-chain.
} 
So long as \emph{any} user capitulates and pays this extra \$5 to be included, the miner increases her revenue by making this threat.
They therefore propose that a TFM should be ``off-chain influence-proof'' (OffC-IP). 
That is, even if the miner were willing to act like a Bayesian monopolist off-chain, she should optimizes her expected revenue by simply following the prescribed TFM on-chain with no off-chain behavior. 
They formally observe that EIP-1559 is not OffC-IP.

Furthermore, \citet{GaneshTW24} establish that no TFM can satisfy all four desirable properties (OnC-US, OnC-MS, OffC-IP, and SCP), even with cryptography.
On the other hand, a cryptographic second-price auction\footnote{
    In a cryptographic second-price auction, all users submit encrypted bids. Then,  without decrypting the bids, the miner chooses which bids to censor, whether to insert fake bids, and how to set a reserve. Then, all bids are decrypted.}
is OnC-US, OnC-MS, and OffC-IP;
they argue that this strongly suggests that users can simply bid their values in this TFM, without worry.
However, the full possibilities and limitations imposed by the OffC-IP constraint remain highly unclear.
This is exactly the gap our paper fills.\\

\noindent\textbf{Main Contribution I: A Burn Identity for OffC-IP TFMs.} %
Our first main result characterizes OffC-IP TFMs by establishing a reduction to a canonical auction problem with a monopsonist and many items.
To explain this result concretely, consider a TFM as a function that takes as input a single bid from each of $n$ bidders $\vec{v} = (v_1,\ldots, v_n)$, and outputs an allocation vector $\AllocRule_1(\vec{v}),\ldots, \AllocRule_n(\vec{v})$, a payment vector $\Pay_1(\vec{v}),\ldots, \Pay_n(\vec{v})$ and a burn amount $\Burn(\vec{v})$. 
On input $\vec{v}$, each bidder $i$ receives the item with probability $\AllocRule_i(\vec{v})$, pays $\Pay_i(\vec{v})$, and the miner receives a revenue $\sum_i \Pay_i(\vec{v}) - \Burn(\vec{v})$.

Now, given that user values are drawn from a prior with virtual value function $\varphi$,\footnote{
  See \autoref{sec:MechDesign} for an overview of virtual values.}  
define  $\AllocRule^{\varphi}_i(\vec{\beta}):=\AllocRule_i(\varphi^{-1}(\beta_1),\ldots, \varphi^{-1}(\beta_n))$
to be the allocation rule that takes as input $\vec{\beta}$ and allocated to bidder $i$ according to $\AllocRule_i$ \emph{as if the submitted bids were $\vec{v} = \varphi^{-1}(\vec{\beta})$ instead of $\vec{\beta}$}.
Now view $\AllocRule^{\varphi}(\cdot)$ as a function that takes as input $n$ values $\vv(\vec{v}) = (\vv(v_1), \dots, \vv(v_n))$ from a monopsonist with an additive valuation, awards that monopsonist the item $i$ with probability $\AllocRule_i^{\varphi}(\vv(\vec{v})) = \AllocRule_i(\vec{v})$ for all $i$, and charges her a payment $\Burn^{\vv}(\vv(\vec{v})) = \Burn(\vec{v})$.
The monoposonist's utility is given by $\sum_{i} \vv(v_i) \AllocRule^\varphi_i(\vv(v_i)) - \Burn^\varphi(\vv(v_i))$, which in expectation (by Myerson's lemma), equals the net revenue received by the miner in equilibrium.
~\autoref{thm:UtilityVersionMainReduction} and \autoref{thm:VirtualPartialConverse} establish that $\bigl(\AllocRule(\cdot), \Pay(\cdot), \Burn(\cdot)\bigr)$ is OffC-IP as a TFM for $n$ single-dimensional bidders if and only if $\bigl(\AllocRule^{\varphi}(\cdot), \Burn(\cdot)\bigr)$ is DSIC as an auction for the multi-dimensional monopsonist, thus
\emph{fully} characterizing the allocation and burn rules of OffC-IP TFMs.

\vspace{0.5em}
\noindent\textbf{Main Contribution II: Posted Price Mechanisms are (nearly) the only plaintext OnC-S and OffC-IP TFMs.}
After establishing our burn identity for OffC-IP mechanisms, we apply this result to OnC-S (i.e., OnC-US and OnC-MS) mechanisms to investigate the limits of these desirable properties.
In particular, we ask whether any such TFMs are \emph{plaintext}, i.e., allowing users to submit un-encrypted bids without resorting to the strong use of cryptography employed by the second-price auction considered by \citet{GaneshTW24}.

On this front, we first show that any nontrivial plaintext OnC-S and OffC-IP mechanism must be \emph{prior-dependent}. That is, the TFM must somehow `know' the same prior as the miner and set prices and burns based on this prior.\footnote{
  This may be an entirely realistic assumption in case the miner forms their prior exclusively from on-chain data, and an entirely unrealistic assumption in case the miner forms their prior primarily based on off-chain data (such as a tweet announcing an NFT drop). We leave as an important direction for future work to understand the fraction of blocks for which on-chain data suffices to form a reasonably accurate prior, and how the magnitude of inaccuracy impacts OffC-IP.}
At a high level, this is because our burn identity shows that any OffC-IP mechanism is dependent on the virtual value $\vv$, which in turn, depends on the prior.
If the miner herself is the source of how the blockchain learns the prior, this input is provided by the miner after seeing the users' plain-text bids, and we show that this cannot be OnC-MS. Therefore, the protocol's knowledge of the prior must be independent of the miner.

Given the above, we ask: are there any desirable prior-dependent mechanisms satisfying these properties?
Under the additional assumption that there is infinite supply (i.e., any number of users can be included), there is a natural deterministic OnC-S and OffC-IP mechanism: The revenue-optimal posted-price mechanism that charges each included user the Myerson reserve $p$ and transfers all charged prices to the miner (i.e., burning nothing).
This mechanism is indeed OnC-S and OffC-IP. 
In fact, we show that the \emph{only} deterministic plaintext TFMs that are OnC-S and OffC-IP are posted price auctions; specifically, they are generalizations of the revenue-optimal mechanism in which a specially-tuned fraction of each price $p$ is burnt (depending on $p$ and on the prior distribution).
Since posted-price mechanisms are not feasible for a finite supply, this implies that there exists \emph{no} deterministic plaintext TFM with finite supply that is OnC-S and OffC-IP.

Given the above, it is natural to ask whether there exist any \emph{randomized} plain-text, OnC-S, and OffC-IP TFMs that are meaningfully distinct from posted-price mechanisms. 
Surprisingly, we show this is indeed the case. 
In particular, we construct a plaintext, OnC-S, and OffC-CIP ``position auction'', where the $k$\textsuperscript{th} highest bidder is awarded an item with some probability $x^{(k)}$,
for any bounded regular distribution, even for a finite block capacity.
While our constructed auction is fairly complex and likely impractical, it shows that the space of OffC-IP and OnC-S mechanisms has an intricate boundary.
Additionally, our techniques in this construction extend to solve a related open problem in TFM design posted by \citet{ChungRS24,GafniY24Discrete}.\footnote{
  Specifically, 
  \citet{ChungRS24} and \citet{GafniY24Discrete} study the design of Global Strong Collusion Proof (GSCP) mechanisms: mechanisms that are resilient to collusion by the global coalition containing the miner and all users.
  They pose the question of designing OnC-S and GSCP mechanisms that also yield a positive revenue to the miner.
  We design such (randomized) mechanisms in \autoref{sec:MIRC}, which are also feasible for a finite block when the users' values are drawn from some bounded support.
}

\vspace{0.5em}
\noindent\textbf{Additional Discussions.} 
Our results highlight the importance of posted-price mechanisms in TFM design. 
There are several variants of posted-price mechanisms.\footnote{
    For example: is the payment burned or given to the miner? Is cryptography used? Can the protocol set prior-dependent reserves?} 
Due to~\citet{GaneshTW24}, none of these can be OnC-US, OnC-MS, OffC-IP, and SCP simultaneously. 
However, several variants satisfy three of the four desiderata; 
hence, we initiate a discussion of the key tradeoffs involved. 
\begin{itemize}
    \item \textbf{SCP or OffC-IP?} 
    In \autoref{sec:posted-zoo}, we address the tradeoff between OffC-IP and SCP: together with OnC-US and OnC-MS, it is possible to have one of them, but not both.
    OffC-IP is robust to miners acting as Bayesian monopolists, or colluding with users they do not trust.
    But, OffC-IP mechanisms are not robust to collusion between trusted parties who know each others' private values. SCP on the other hand is robust even to collusion between trusted parties who know each others' private values (i.e.~even Coinbase cannot profit in EIP-1559 by having its stakers collude with its users). But, SCP is not robust to miners acting as Bayesian monopolists. Which property is more desirable in practice, and whether either admit meaningful approximate variants, remain intriguing open directions.
    \item \textbf{Prior-dependence, cryptography, or communication?} 
    In \autoref{sec:discussion}, we give additional discussion on the merits and drawbacks of different implementations of a posted-price variant that are OnC-US, OnC-MS, and OffC-IP. 
    We discuss three such variants.
    \begin{itemize}
        \item  First, the protocol can set the price.
    This may be possible using on-chain data, but under what conditions and for what fraction of the blocks?
    \item Second, users can submit encrypted bids that will surely decrypt after the miner has finalized the input. 
    Striking modern advances in cryptography make this more realistic every day, but will intensive cryptography ever be able to match the throughput required?
    \item Third, and perhaps most-speculatively, miners can `commit' to a price in advance by, for example, posting it on a previous block.
    Can this cover realistic variations in demand over time, or do miners' prior change so quickly from block-to-block that a rigorous analysis of this sequential pricing game is necessary to understand its viability?
    \end{itemize}
\end{itemize}

Finally, we remark on the role of our paper within the literature on TFM design. 
One way to frame the main question of this literature is: How will a block-building protocol devolve into the actual mechanism faced by users? 
OnC-MS stipulates that miners may manipulate the protocol by myopically best-responding; thus, for example, a second-price auction will devolve into a first-price auction (since the miner will insert fake bids just below the winner's). 
On the other hand, OffC-IP stipulates that the miner will act as a monopolist running whatever mechanism she wants to determine the inputs to the protocol; thus, for example, EIP-1559 will devolve into a mechanism where users must tip the miner to be included.
More generally, under the lens of OffC-IP, the \emph{only} binding constraint on the miner is the set of all possible allocations that are supported by the TFM and the burns corresponding to these allocations, which act as a production cost.
Hence, we expect the miner to implement (by any means necessary) the revenue-optimal mechanism given these production costs, and the block-building process does not have the ability to determine the payment rule faced by users in the end.
However, when the mechanism is OffC-IP, users can (at least) reliably expect the entire mechanism to be conducted on-chain.
This motivates the search for a deep understanding on the burn rules supported by OffC-IP TFMs, so that other properties that enable better user experience, like OnC-MS, can be guaranteed in conjunction with OffC-IP.
This is precisely what our paper provides.

\subsection{Related Work}

Our work contributes to the rapidly-expanding literature on transaction fee mechanism design and more broadly, on game-theoretic applications to decentralized platforms.
While \citet{LaviSZ19,BasuEOS19} and \citet{Yao18} were amongst the first to view inclusion of transactions in a block as a mechanism design problem, a research agenda around TFMs started building after \citet{Roughgarden20, Roughgarden21} established basic desiderata that a ``good'' TFM must satisfy, and analyzed EIP-1559, the TFM adopted by Ethereum, through the lens of these desiderata.
However, \citet{ChungS23} argue that simultaneously satisfying the three properties proposed by \citet{Roughgarden20} --- truthfulness for users and the miner and collusion-resistance --- is impossible for any non-trivial mechanism that ever allocates some user with a positive probability.
Various notions of collusion-resistance have been considered \citep{ChungS23, FerreiraGR24} while \citet{GafniY24, ChungRS24} investigate the relationship between these collusion-resistance desiderata.
Further work has also considered relaxing the concept of truthfulness for users from DSIC to BIC \citep{GafniY22,ChenSZZ24}.
\citet{ShiCW23, WuSC24} consider expanding the space of TFMs by arming the mechanism with cryptography.
While all the aforementioned works treat TFMs as single-shot mechanisms, \citet{FMDPS21, BabaioffN24} and \citet{ADM24} investigate the long-term dynamics of the repeated game that arises from running the TFM in every block.  \citet{BahraniGR24} are the first to revisit the desiderata for TFMs in blockchains with MEV.

Beyond TFMs, a wide body of literature has been growing around economic design in blockchains, starting with \citet{BabaioffDOZ12} and \citet{EyalS14}.
More recent works range from designing credible auctions via cryptography \citep{FerreiraW20, EssaidiFW22, ChitraFK23, GZ25, ChungWS25} and persuasion mediated via blockchains \citep{DrakopoulosLM23}.
There has also been growing literature on cryptographic and consensus protocols in the presence of rational agents \citep{garay2013rational,halpern2004rational,asharov2011game,groce2012rationalBA, gong2025collusion, KGPBW25}.

Our paper is a follow-up to \citet{GaneshTW24}, which define off-chain influence-proofness (OffC-IP).
They construct a cryptographic OffC-IP and on-chain simple (OnC-S) TFM based on a second price auction, and prove that OffC-IP and OnC-S are incompatible with strong notions of collusion-proofness considered in previous papers.
Besides fully characterizing the allocation and burn rules of any OffC-IP mechanism, the present paper advances beyond the results of \citet{GaneshTW24} by (1) uncovering \emph{plaintext} OffC-IP and OnC-S TFMs, thus removing the strong use of cryptography, and (2) providing formal impossibility results regarding OffC-IP and OnC-S TFMs (regardless of whether the TFM is collusion-proof).

\subsection{Road-map}
\label{sec:roadmap}

After giving preliminaries, our paper proceeds along the lines of our two main contributions as follows:
\begin{enumerate}
  \item First, we characterize the allocation and burn rule of any OffC-IP TFM: they are exactly those satisfying the \emph{burn identity} we define in \autoref{thm:UtilityVersionMainReduction}. 
  \begin{itemize}
      \item \autoref{sec:PostedPrice} first provides a much-simpler warm-up by characterizing the burn rule of OffC-IP \emph{posted-price} TFMs.
      \item \autoref{sec:MonopsonistLens} presents the burn identity for general mechanisms.
  \end{itemize}
  \item Second, we examine settings under which OffC-IP and OnC-S TFMs exist, and show that posted-price mechanisms are the only deterministic plaintext OffC-IP and OnC-S TFMs.
  \begin{itemize}
    \item In \autoref{sec:PriorIndependent}, we consider prior-independent TFMs, and show that in order to be OffC-IP and OnC-S, such TFMs require \emph{both} cryptography and miner-advice (as leveraged by \citealp{GaneshTW24} to construct such a TFM).
    
    \item Motivated by the above, in \autoref{sec:ImpossibilitesandPossibilities} we consider prior-dependent mechanisms. We prove our second main result in \autoref{thm:NoDeterministicMechanisms}: that deterministic OffC-IP and OnC-S TFMs must \emph{necessarily} be the posted-price mechanisms of \autoref{sec:PostedPrice}. Then we examine various implications and limitations of this result.
    \begin{itemize}
        \item In particular, \autoref{thm:impossibility-deterministic-finite} observes that deterministic OffC-IP and OnC-S TFMs thus \emph{require infinite supply}, i.e., they do not exist with finite capacity.
        \item In contrast, \autoref{sec:rand-auction} shows that there exist \emph{randomized} OffC-IP and OnC-S TFMs with finite supply, so long as the prior distribution of values is bounded.
    \end{itemize}
  \end{itemize}
\end{enumerate}

%% file: 20-prelims.tex
\section{Preliminaries: Model of Transaction Fee Mechanisms} \label{sec:MOdel}

We begin with review of transaction fee mechanisms.
For preliminaries on Bayesian mechanism design generally, see \autoref{sec:MechDesign}.

We follow \citet{GaneshTW24}, and consider a model of TFMs which allows the miner to perform both on-chain manipulations (i.e., dropping bids or submitting fake bids) and off-chain manipulations (i.e., conducting an arbitrary off-chain mechanism to determine behavior on-chain).
In this model, the TFM is specified by a \emph{block-building process} $\bBuild$, which is an algorithm for constructing a block based on bids collected from all the users and an input from the miner termed the ``miner advice'' (for example, a reserve price).

Strategic play proceeds in either the \emph{on-chain game} $\onCG$ or the \emph{off-chain game} $\offCG$.
In $\onCG$, the miner submits inputs to $\bBuild$ by---as a function of the submitted bids---deciding on an advice, censoring any set of bids, and submitting her own fabricated bids.
In $\offCG$, the miner decides on an arbitrary \emph{off-chain mechanism} $\Moff$ which solicits inputs from users and dictates the on-chain strategies of all agents in $\onCG$.
In addition to the payments made in the on-chain game $\onCG$,
the miner can also direct off-chain transfers through $\Moff$ in $\offCG$.
We omit the full details of this model of TFMs, and refer the reader to \citet{GaneshTW24}.

We adopt the notation $\bigl(\AllocRule(\cdot), \Pay(\cdot), \Burn(\cdot)\bigr)$ for specifying the block-building process $\bBuild$,
where $\AllocRule_i(a_\mi,b_1,\ldots,b_n)$ and $\Pay_i(a_\mi,b_1,\ldots,b_n)$ specify user $i$'s allocation and payment, and $\Burn(a_\mi,b_1,\ldots,b_n)$ specifies the burn, for the miner's advice $a_{\mi}$ and bids $b_1, \dots, b_n$.
The on-chain revenue to the miner equals $\sum_{i=1}^n \Pay_i(a_\mi,b_1,\allowbreak \ldots, \allowbreak b_n) \allowbreak - \Burn(a_\mi,b_1,\ldots,b_n)$.
We summarize the main concepts and notation in \autoref{tab:main-notation}.

Most of our paper focuses on the \emph{plaintext} model in which bids are un-encrypted and visible to the miner, and hence her strategy can condition on the values of the bids.
\citet{GaneshTW24} also consider a \emph{cryptographic} model of TFMs, which differs only in that users submit encrypted bids, and hence the miner's action must be independent of the values of the bid.

We restrict attention to TFMs in which the block-building process $\bBuild$ is \emph{anonymous}.%
\footnote{
  Anonymity was also implicitly assumed in the impossibility result in \citet{GaneshTW24} (e.g., this assumption is needed to say that the outcome of the block-building process when a miner fabricates a bid of $0$ is the same as when a user submits a bit of $0$).  
} 
As in standard mechanism design, this means that each user is treated identically by the TFM, i.e., that if the users' $\{1,\ldots,n\}$ are relabeled with any permutation $\pi$, then the outcome is identical except that it is relabeled by $\pi$.

\newcommand{\mc}[1]{\multicolumn{2}{c}{{#1}}}
\newcommand{\mcf}[1]{\multicolumn{2}{c}{\makecell{\footnotesize{}({#1})}}}

\begin{table}[h]
\centering

\renewcommand{\arraystretch}{2.5}
\newcommand{\snug}{-0.8em}

\begin{tabular}{lll}
{Concept} 
& {Notation}
& {Brief Description}
\\
\hline
\hline
Block-building process 
& $\bBuild( a_{\mi}, \vectr{b} )$
&{\footnotesize{}\makecell[l]{
    Algorithm which determines outcome based on 
    \\miner advice $a_\mi$ and users' bids $\vectr{b}$} 
    }
\\
\makecell[l]{
Allocation, price, 
\\ \& burn rule of $\bBuild$
}
& $(\AllocRule, \Pay, \Burn)$
&{\footnotesize{}\makecell[l]{
    $\AllocRule_i(a_{\mi},\vectr{b})$ denotes $i$'s
    allocation, 
    \\ and $P_i(a_{\mi},\vectr{b})$ denotes $i$'s price.
    \\ The miner earns $\sum_i \Pay_i(a_\mi,\vectr{b}) - \Burn(a_\mi,\vectr{b})$. }
    }
\\
On-chain game 
& $\onCG\big( s^\onCG_\mi, s^\onCG_{\usr}(\vectr{v}) \big)$
&{\footnotesize{}\makecell[l]{
    Game in which each user $i$ bids $s^\onCG_{\usr,i}(v_i)$, the 
    \\miner sees their bids, then forwards some
    \\advice, fabricated bids, and subset of user bids to $\bBuild$} 
    }
\\
User BNE in $\onCG$
& $\sigma^\onCG = \big(s^{\onCG}_{\usr,i}(v_i)\big)_{i=1}^n$
&{\footnotesize{}\makecell[l]{
    Strategy profile in $\onCG$ such that, given 
    \\ that the miner plays some fixed $s^\onCG_\mi$, the users 
    \\are playing in a Bayes-Nash equilibrium}
    }
\\
Off-chain mechanism
& $\Moff(\vectr{u})$
&{\footnotesize{}\makecell[l]{
    An arbitrary mechanism, decided by the miner,
    \\that determines the users' and miner's 
    \\strategy in the on-chain game}
    }
\\ 
Off-chain game
& $\offCG\big(\Moff; s^{\D, \Moff}_{\usr}(\vectr{v})\big)$
&{\footnotesize{}\makecell[l]{
    Game in which the miner commits to $\Moff$, and
    \\users submit messages and payments to $\Moff$ to
    \\determine users' and the miner's strategies in $\onCG$}
    }
\\ 
User BNE in $\offCG$
& $\sigma^\offCG = \big(s^{\offCG,\Moff}_{\usr,i}(v_i)\big)_{i=1}^n$
&{\footnotesize{}\makecell[l]{
    Strategy profile in $\offCG$ such that, given
    \\ that the miner commits to some $\Moff$, the users 
    \\are playing in a Bayes-Nash equilibrium}
    }
\\ 
Revenue in $\offCG$
& $\Rev^\offCG(\Moff, \sigma^\offCG(\vectr{v}))$
&{\footnotesize{}\makecell[l]{
    The miner's total reward in $\offCG$, i.e., the
    \\sum of the on- and off-chain payments.}
    }
\end{tabular}
\caption{Main elements of the model}
\label{tab:main-notation}
\end{table}

Before proceeding, we highlight in concrete terms the main ``actions'' that a miner can take in the off-chain game, which we need to exploit in our proofs.
A miner can (1) censor bids, (2) fabricate and submit their own bids, and (3) ask bidders to submit some alternative bid.
Since the miner conducts an entirely separate off-chain mechanism to decide how users behave on-chain, each of actions (1)-(3) can depend on the entire profile of users' messages in the off-chain game.
We always assume the users play in a BNE in the off-chain game.
Hence, by applying the revelation principle to the off-chain mechanism chosen by the miner, we might as well assume that the off-chain equilibrium is truthful, and the miner learns the values of all the users.
It is notationally convenient to describe the total payments made by the users in the off-chain game (i.e, in the composite mechanism including both, the on-chain and off-chain components) as the off-chain payments.\footnote{For example, if the TFM is a posted-price mechanism where the users are expected to pay $p$ on-chain and nothing off-chain, we say that $\Moff$ charges an off-chain payment $p$ from its users.}
Thus, the off-chain payments made by users must satisfy the payment identity (\autoref{item:payment-identity} in \autoref{thm:myerson}).

\citet{GaneshTW24} consider three core desiderata for TFM design: On-chain user simplicity, on-chain miner simplicity, and off-chain influence proofness,
which each capture different ways in which agents (the users or the miner) will want to deviate from a specified strategy profile.
We now briefly recall these definitions, and again refer the reader to \citet{GaneshTW24} for full details.
\begin{itemize}
  \item An on-chain miner strategy and user BNE $\sigma^\onCG =  (s_\mi^\onCG, s_{\usr, 1}^\onCG,\allowbreak \dots, \allowbreak s_{\usr, n}^\onCG)$ is \emph{on-chain user simple} if all users bid their value in $\sigma^\onCG$, and moreover, the equilibrium is DSIC (i.e., each user best-responds by bidding his value for \emph{any} bids of the other users).

  \item An on-chain miner strategy and user BNE $\sigma^\onCG =  (s_\mi^\onCG, s_{\usr, 1}^\onCG,\allowbreak \dots, \allowbreak s_{\usr, n}^\onCG)$ is \emph{on-chain miner simple} if the miner is best-responding to the strategies of the users, and moreover, $s^\onCG_{\mi}$ never drops or fabricates any bids and always submits the same constant miner advice.
  In such a case, we say that the miner's strategy is \emph{compliant}.

  \item An on-chain miner strategy and user BNE $\sigma^\onCG =  (s_\mi^\onCG, s_{\usr, 1}^\onCG,\allowbreak \dots, \allowbreak s_{\usr, n}^\onCG)$ is \emph{off-chain influence proof} if the miner achieves her maximum revenue in $\sigma^\onCG$ over \emph{all off-chain mechanisms} $\Moff$ and all possible user BNE $\big( \widetilde{s}_{\usr, i}^{\offCG, {\mathcal{M}}_{\mathsf{off}}} \big)_{i=1}^n$ under $\Moff$.
\end{itemize}
We say that $\sigma^{\onCG}$ is \emph{on-chain simple} if it satisfies both on-chain user and miner simplicity.

\citet{GaneshTW24} also provide supplementary desiderata on collusion resistance, mirroring the definitions of \citet{ChungS23} and \citet{ChungRS24}:
\begin{itemize}
    \item An on-chain miner strategy and user BNE $\sigma^{\onCG} = (s_\mi^\onCG, s_{\usr, 1}^\onCG,\allowbreak \dots, \allowbreak s_{\usr, n}^\onCG)$ is \emph{$1$-$1$-strong collusion proof} if any coalition containing the miner and one other user cannot deviate to increase the coalition's joint utility.
    \item An on-chain miner strategy and user BNE $\sigma^{\onCG} = (s_\mi^\onCG, s_{\usr, 1}^\onCG,\allowbreak \dots, \allowbreak s_{\usr, n}^\onCG)$ is \emph{global strong collusion proof} if the global coalition containing the miner and all users cannot deviate to increase their joint utility.
\end{itemize}
The difference between coalitions and the off-chain mechanisms is the following.
While forming a coalition, the agents in the coalition entirely stop being strategic with each other: the users truthfully reveal their values to the miner, and the coalition is purely interested in maximizing their joint utility as if they were a single entity.
In contrast, in the off-cain mechanism  the users and the miner continue to be strategic with each other (i.e., they must be playing in an equilibrium).

%% file: 30-postedprice.tex
\section{Warmup: Characterizing Off-Chain Influence Proof Posted-Price Mechanisms} \label{sec:PostedPrice}
In this section, we give exposition into our framework and our results by considering a simple sub-class of TFMs: posted-price mechanisms (with infinite supply, where the price and the burn are prior-dependent and determined by the blockchain).
We give a short proof characterizing the burn rule such that a posted-price mechanism is off-chain influence proof.
Our later sections use more involved arguments to generalize this characterization of the burn rule to all off-chain influence proof TFMs.

Formally, our goal is as follows.
Consider a posted-price mechanism with infinite supply.
The block-building process $\bBuild = (\AllocRule, \Pay, \Burn)$ does not receive any advice from the miner, and is determined by two fixed parameters: a price $\QPay$ and a burn $\BurnB$.
The mechanism includes all users with values larger than $\QPay$, charges each allocated user a payment $\QPay$ and burns $\BurnB$ per allocated user.
Consider the on-chain equilibrium $\sigma^{\onCG}_{\mathsf{honest}}$ where the users bid their values truthfully and the miner does not fabricate or censor bids.
We want to compute $\QPay$ and $\BurnB$ such that the equilibrium induced by $\sigma^{\onCG}_{\mathsf{honest}}$ in the off-chain game is off-chain influence proof.

\begin{proposition} \label{thm:PostedPrice}
    For a distribution $\Distr$ of user values with a continuous virtual value function $\vv$, suppose that the block-building process posts a price $\QPay$ and burns $\BurnB$ per allocated user.
    Then, the equilibrium $\sigma^{\onCG}_{\mathsf{honest}}$ is off-chain influence proof if and only if $\BurnB = \vv(\QPay)$.
\end{proposition}
\begin{proof}
Suppose that the miner runs an off-chain mechanism $\Moff$ and induces a BNE $\sigma^{\offCG} = (s^\offCG_\mi, s^\offCG_{\usr, 1}, \dots, s^\offCG_{\usr, n})$.
When users have a value profile $\Vec{v}$, let their on-chain allocation be $\OnCAlloc(\vec{v}) = \AllocRule(s^\offCG_\mi(\vec{v}), s^\offCG_{\usr, 1}(v_1), \dots, s^\offCG_{\usr, n}(v_n))$.
Let $\OnCPay(\vec{v})$ denote the equilibrium payments as dictated by the payment identity (\autoref{item:payment-identity} in \autoref{thm:myerson}) for the allocation rule $\OnCAlloc$.

Then, a Bayesian monopolist miner maximizes her net expected revenue equal to
\begin{align}
    \notag
    \Es{\vectr{v} \sim \Distr^n}{ \vphantom{\Big|} \Rev^\offCG(\Moff, \sigma^\offCG(\vectr{v}))} 
    &= \Es{\vec{v}\sim \Distr^n}{\sum_{i = 1}^n \OnCPay_i(\vec{v}) - \BurnB \cdot \OnCAlloc_i(\vec{v})} 
    \\ & \notag \qquad \text{(By Myerson's lemma, \autoref{item:revenue-equals-virtual-welfare} in \autoref{thm:myerson})}
    \\ \notag &= \Es{\vec{v}\sim \Distr^n}{\sum_{i = 1}^n \vv(v_i) \, \OnCAlloc_i(\vec{v}) - \BurnB \cdot \OnCAlloc_i(\vec{v})}
    \\ &= \Es{\vec{v}\sim \Distr^n}{\sum_{i = 1}^n \Big(\vv(v_i) - \BurnB \Big) \cdot \OnCAlloc_i(\vec{v})} \label{eqn:PostedPriceVV}
\end{align}

The miner optimizes her expected revenue precisely by allocating all users with a virtual value at least $\BurnB$, which corresponds to setting a price $\QPay$ such that $\BurnB = \vv(\QPay)$.
\end{proof}

\autoref{thm:PostedPrice} shows that, with a carefully-tuned burn rule, any posted price auction yields an off-chain influence proof TFM.
Additionally, it is not hard to see that these TFMs are on-chain simple too.
This highlights the sharp difference between the TFMs we consider and those of prior work; for instance, \citet{ChungS23} show that there is \emph{no} TFM with positive miner revenue that satisfies on-chain simplicity along with $1$-$1$-strong collusion proofness.

\autoref{eqn:PostedPriceVV} gives intuition towards our burn identity in \autoref{sec:MonopsonistLens}, 
which generalizes this result to apply to \emph{all} TFMs (and not only posted-price mechanisms).
Our reduction in \autoref{sec:MonopsonistLens} views the miner as a (single) bidder in a multi-item auction, where each item $i$ corresponds to including bidder $i$, and relates the total burn in the TFM to the price charged in the multi-item auction.
From here, we argue that the miner's objective is to maximize the expectation of $\sum_{i = 1}^n \vv(v_i) \, \AllocRule_i$ minus the total burn, generalizing \autoref{eqn:PostedPriceVV}.

Furthermore, \autoref{eqn:PostedPriceVV} can  be used to find the revenue-optimal equilibrium of a posted price TFM when the mechanism is not off-chain influence proof.
Specifically:
\begin{itemize}
\item 
When the burn $B$ satisfies $\vv(\QPay) > \BurnB$, the miner will want to include users with virtual values in the range $[\BurnB, \vv(\QPay)]$ which otherwise are not included in the on-chain equilibrium $\sigma^{\onCG}_{\mathsf{honest}}$.
The miner can ask the users with virtual values in the said interval to amplify their bids to $\QPay$ and get included in the block.
Since all users with a value at least $\vv^{-1}(\BurnB)$ are now included, the miner has to charge each user a payment $\vv^{-1}(\BurnB)$ in equilibrium.
Thus, the miner issues a rebate $\QPay - \vv^{-1}(\BurnB)$ to all users since they would have paid $\QPay$ on-chain in order to get included.

\item
On the other hand, when $\vv(\QPay) < \BurnB$, the miner would not want to include the users with virtual values $[\vv(\QPay), \BurnB]$.
Thus, the miner sets up an off-chain mechanism where it charges an entry fee $\vv^{-1}(\BurnB) - \QPay$ on top of the $\QPay$ charged on chain.
\end{itemize}

Before proceeding, we note that there are other variants of posted-price mechanisms beyond the variant considered here. In \autoref{sec:posted-zoo}, we discuss in detail the block-building processes of various posted-price mechanisms (for example, EIP-1559, prior-dependent and prior-independent posted-price mechanisms, etc) and summarize the properties satisfied by their respective equilibria.

%% file: 40-reduction.tex
\section{Myerson-in-Range Mechanisms and a Multi-Item Monopsonist Lens} \label{sec:MonopsonistLens}

In this section, we give a reduction relating off-chain influence proof mechanisms to mechanisms in the classical multi-item single-buyer environment.
This allows us to derive what we call the \emph{burn identity}, which relates the burn to the allocation rule (in pretty much the same way that the \emph{payment identity} relates the payments to the allocation rule). 

For an arbitrary block-building process $\bBuild$, we would like to check whether a given equilibrium $\sigma^{\offCG}$ in the off-chain game is off-chain influence proof, which entails showing that $\sigma^{\offCG}$ maximizes revenue over the set of all off-chain mechanisms $\Moff$ and all possible BNE supported by the mechanism $\Moff$.
We prove three easy-to-verify sufficient conditions for the equilibrium $\sigma^{\offCG}$ to be off-chain influence proof which will, in turn, be useful in deriving the burn identity.

\subsection{Myerson-in-Range Mechanisms and Off-Chain Influence Proofness}

As a first step towards deriving the burn identity, we re-interpret the revenue optimization problem faced by the miner in the off-chain game in terms of the equilibrium burn and the virtual values of the allocated users.

For a revenue optimal equilibrium $\sigma^{\offCG}$ and a corresponding mechanism $\Moff$ in the off-chain game, we can apply the revelation principle to $\sigma^{\offCG}$ to unpack the underlying truth-telling mechanism.
We say that the truth-telling mechanism $\DirReveal$ corresponding to a revenue optimal equilibrium $\sigma^{\offCG}$ is \emph{Myerson-in-Range}.

\begin{definition}[Direct-revelation mechanism] \label{def:DirectRevMechanism}
    For a block-building process $\bBuild = (\AllocRule, \Pay, \Burn)$, an off-chain mechanism $\Moff$, and a user BNE $\sigma^{\offCG} = (s^\offCG_\mi, s^\offCG_{\usr, 1}, \dots, s^\offCG_{\usr, n})$ with respect to $\Moff$, the \emph{direct-revelation mechanism} $(\OnCAlloc, \OnCPay, \OnCBurn)$ corresponding to $\Moff$ and $\sigma^\offCG$ is obtained by applying the revelation principle; i.e., 
    $$(\OnCAlloc(\vec{v}), \OnCBurn(\vec{v})) = \big(\AllocRule(\sigma^{\offCG}(\vec{v})), \Burn(\sigma^{\offCG}(\vec{v}))\big)$$
    and $\OnCPay$ from the payment identity (\autoref{item:payment-identity} in \autoref{thm:myerson}) for the allocation rule $\OnCAlloc$.
    We will say that $\OnCAlloc, \OnCPay$ and $\OnCBurn$ are the value allocation, payment and burn rules respectively.
\end{definition}

\begin{definition}[Myerson-in-Range]
    Fix a block-building process $\bBuild$ and a distribution of user values $\Distr$.
    Let $\Moff$ and $\sigma^{\offCG} = (s^\offCG_\mi, s^\offCG_{\usr, 1}, \dots, s^\offCG_{\usr, n})$ respectively be an off-chain mechanism and a user BNE with respect to $\Moff$ that maximizes the miner's expected net revenue $\Es{\vectr{v} \sim \Distr^n}{ \Rev^\offCG(\Moff, \sigma^\offCG(\vectr{v}))}$ for any number of users $n$.
    We say that $(\OnCAlloc, \OnCPay, \OnCBurn)$ is \emph{Myerson-in-Range} for $\bBuild$ and the distribution $\Distr$ if $(\OnCAlloc, \OnCPay, \OnCBurn)$ is the direct-revelation mechanism corresponding to such an $\Moff$ and $\sigma^\offCG$.
\end{definition}

It immediately follows from the definition that an off-chain equilibrium $\sigma^{\offCG}$ with a trivial off-chain component is off-chain influence proof if and only if the corresponding direct-revelation mechanism is Myerson-in-Range. 

\begin{lemma} \label{thm:MyersonEnvelopeIffOffCIP}
    A BNE $(s^\offCG_\mi, s^\offCG_{\usr, 1}, \dots, s^\offCG_{\usr, n})$ with a trivial off-chain component is off-chain influence proof for the block-building process $\bBuild = (\AllocRule, \Pay, \Burn)$ if and only if its direct-revelation mechanism $(\OnCAlloc, \OnCPay, \OnCBurn)$ is Myerson-in-Range.
\end{lemma}

It is notationally cumbersome to keep mentioning that the mechanism and the corresponding equilibrium satisfies some property (for eg. off-chain influence proofness).
We abuse notation to instead say that the corresponding direct-revelation mechanism $(\OnCAlloc, \OnCPay, \OnCBurn)$ satisfies the same property.

In the remainder of this section, we characterize all Myerson-in-Range mechanisms for a given block-building process $\bBuild$.
We argue that the direct-revelation mechanism $\DirReveal$ is Myerson-in-Range if and only if it optimizes the miner's expected \emph{virtual utility} --- the difference between her expected \emph{virtual surplus} $\Es{\vec{v} \sim \Distr^n}{\sum_{i = 1}^n \vv(v_i) \, \OnCAlloc(\vec{v})}$ and the burn $\Es{\vec{v} \sim \Distr^n}{\OnCBurn(\vec{v})}$ --- for all $n \in \N$.
In detail, for an equilibrium $\sigma^{\offCG}$ in the off-chain game with a corresponding direct-revelation mechanism $\DirReveal$, the miner's net expected revenue equals
\( 
\Es{\vec{v} \sim \Distr^n}{\sum_{i = 1}^n \OnCPay_i(\vec{v}) - \OnCBurn(\vec{v})}.
\)
Since $\OnCPay$ is the value payment rule of a BNE, we can apply Myerson's lemma (\autoref{item:revenue-equals-virtual-welfare} in \autoref{thm:myerson}) to replace the expected payments by the expected virtual surplus in the expected net revenue.
In other words, we can pretend that the miner receives a surplus equal to $\sum_{i = 1}^n \vv(v_i) \, \OnCAlloc_i(\vec{v})$ upon allocating users with values $\vec{v}$, but is charged $\OnCBurn(\vec{v})$ via burns.
Therefore, the miner optimizes her expected virtual utility $\Es{\vec{v} \sim \Distr^n}{\sum_{i = 1}^n \vv(v_i) \, \OnCAlloc(\vec{v}) - \OnCBurn(\vec{v})}$ over all possible BNE $\sigma^{\offCG}$ in the off-chain game.

In fact, for a regular distribution $\Distr$, the miner can pointwise optimize her virtual utility for $\vec{v} \sim \Distr^n$.
The value allocation and burn rules $(\OnCAlloc, \OnCBurn)$ can take any value from the set of all \emph{feasible outcomes} $\Feasibility$ of $\bBuild = (\AllocRule, \Pay, \Burn)$, i.e,
\[ \Feasibility
= \big\{ \bBuild(a_\mi; \vectr{b}) \big\}_{a_\mi, \vectr{b}} . \]
For any allocation rule $\OnCAlloc$ supported on $\Feasibility$, by the payment identity (\autoref{item:payment-identity} in \autoref{thm:myerson}), $\OnCAlloc$ is the value allocation rule of some BNE $\sigma^{\offCG}$ if and only if $\OnCAlloc_i$ is monotone non-decreasing in $v_i$ for all $1 \leq i \leq n$ ($\OnCAlloc_i$ is the allocation rule for user $i$).
For a regular distribution $\Distr$, it is fairly straightforward to see that there exists
\[\OnCAlloc(\vec{v}) = \arg \max_{\AllocRule : (\AllocRule, \BurnB) \in \Feasibility} \sum_{i = 1}^n \vv(v_i) \AllocRule_i - \BurnB\]
that is monotone non-decreasing in user values.
For that matter, any variant of $(\OnCAlloc, \OnCBurn)$ that is monotone, and pointwise optimizes for the virtual utility except with probability zero constitutes a Myerson-in-Range mechanism $\DirReveal$ for a suitably chosen value payment rule $\OnCPay$.
The miner's expected virtual utility is maximized even if there exists a set of value profiles with a probability measure zero for which the virtual utility optimal outcome is not implemented on-chain.

We summarize our discussion in the theorem below.
\begin{theorem} \label{thm:MyerImpliesVirtualUtilityOpt}
    Let $\Distr$ be a regular distribution.
    Then, the direct-revelation mechanism $(\OnCAlloc, \OnCPay, \OnCBurn)$ is Myerson-in-Range for the block-building process $\bBuild$ if and only if (a) $\OnCAlloc$ is monotone non-decreasing for all $\vec{v} \in \supp(\Distr^n)$ and $n \in \N$, (b) $\OnCPay$ satisfies the payment identity (\autoref{item:payment-identity} in \autoref{thm:myerson}) and (c) except with probability zero, $$\big(\OnCAlloc(\vec{v}), \OnCBurn(\vec{v})\big) = \arg \max_{(\AllocRule, \BurnB) \in \Feasibility} \sum_{i = 1}^n \vv(v_i) \AllocRule_i - \BurnB$$
    for $\vec{v} \sim \Distr^n$.\footnote{The term ``Myerson-in-Range'' is derived from ``Maximal-in-Range'' which originated in the combinatorial auctions literature. A mechanism is said to be maximal over some range $\Feasibility$ of feasible allocations if it implements the surplus optimal allocation over $\Feasibility$ (typically, $\Feasibility$ is a strict subset of the set of all feasible allocations). Similarly, a Myerson-in-Range mechanism maximizes the miner's virtual utility while restricting allocations to those permitted by the set $\Feasibility$, which could possibly pack much less users than the capacity of the block.}
\end{theorem}

\subsection{Three Canonical Properties of Myerson-in-Range Mechanisms}

In our second step towards obtaining the burn identity, we will identify three intuitive properties satisfied by Myerson-in-Range mechanisms.

The first property (labeled \emph{optimal for $n$ users}) is an immediate special case of our characterization of Myerson-in-Range mechanisms.
While \autoref{thm:MyerImpliesVirtualUtilityOpt} suggests that $\OnCAlloc(\vec{v})$ pointwise optimizes $\sum_{i = 1}^n \vv(v_i) \AllocRule - \BurnB$ almost everywhere for $(\AllocRule, \BurnB) \in \Feasibility$, the first condition specifically observes that $\vv(\vec{v}) \, \OnCAlloc(\vec{v}) - \OnCBurn(\vec{v}) \geq \vv(\vec{v}) \, \OnCAlloc(\vec{w}) - \OnCBurn(\vec{w})$ with probability $1$ for $\vec{v} \sim \Distr^n$ and all $\vec{w} \in \supp(\Distr^n)$.

The second property (labeled \emph{Negative $\vv$'s are suboptimal}) observes that including users with value below the monopoly reserve is suboptimal for a virtual utility maximizing miner.
Suppose an allocation rule includes users with a negative virtual value with probability greater than zero.
Then, the miner can increase her virtual utility by first censoring such a user's bid and fabricating an identical bid of her own.
Since the same set of bids are submitted to the anonymous TFM, except for un-allocating the user with the negative virtual value, all other users are allocated the same and the burn remains unchanged too, thereby increasing the miner's expected virtual utility.

Finally, the third property (labeled \emph{No censoring or fabricating}) connects the allocation rule for $n$ users with the allocation rule for $\hat{n}$ users.
It suggests that the miner's virtual utility should remain invariant to censoring or fabricating bids with a non-positive virtual value.

\begin{theorem} \label{thm:MainReduction}
    Let $\Distr$ be a regular distribution with a continuous virtual value function $\vv$ and a probability mass function such that $\Pr_{v \sim \Distr}[\vv(v) \leq 0] > 0$.
    Further, let $\Distr_{\vv \leq 0}$ be the distribution $\Distr$ conditioned on the virtual value of a random draw having a non-positive virtual value.
    Then, a direct-revelation mechanism $(\OnCAlloc, \OnCPay, \OnCBurn)$ is Myerson-in-Range only if for any number $n$ of users:
    \begin{enumerate}[(A)]
        \item \label{Bul:1} (\emph{Optimal for $n$ users.}) For a value profile $\vec{v} \sim \Distr^n$ and $\vec{w} \in \supp(\Distr^n)$,
        $$\sum_{i = 1}^n \vv(v_i) \, \OnCAlloc_i(\vec{v}) - \OnCBurn(\vec{v}) \geq \sum_{i = 1}^n \vv(v_i) \, \OnCAlloc_i(\vec{w}) - \OnCBurn(\vec{w})$$
        with probability one.
        \item \label{Bul:2} (\emph{Negative $\vv$'s are suboptimal.}) For $\vec{v} \sim \Distr^n$, $\OnCAlloc_i(\vec{v}) = 0$ almost surely whenever $\vv(v_i) < 0$.
        \item \label{Bul:3} (\emph{No censoring or fabricating.}) For $\vec{v} = (v_1, \dots, v_n, v_{n+1}, \dots, v_{n+t}) \sim \Distr^{n} \times \Distr_{\vv \leq 0}^t$ and $\vec{w} = (v_1, \dots, v_{n})$, 
        $$\sum_{i = 1}^{n} \vv(v_i) \OnCAlloc_i(\vec{v}) \allowbreak - \OnCBurn(\vec{v}) \allowbreak = \sum_{i = 1}^{n} \vv(v_i) \OnCAlloc_i(\vec{w}) \allowbreak - \OnCBurn(\vec{w})$$
        almost surely.
    \end{enumerate}    
\end{theorem}

We defer the proof of \autoref{thm:MainReduction} to \autoref{sec:ProofMainReduction}.
For the remainder of our discussions, we solely focus on the class of regular distributions $\Distr$ meeting the requirements of \autoref{thm:MainReduction}, i.e., having a continuous virtual value function and a positive probability of a random draw having a non-positive virtual value.
Most natural distributions such as the uniform distribution, the normal distribution and the exponential distribution satisfy this condition.
We say that a distribution is \emph{smooth} if it satisfies the requirements from \autoref{thm:MainReduction}.

\subsection{The Burn Identity for Myerson-in-Range Mechanisms}

In this section, we derive a closed-form expression for the burn rule of a Myerson-in-Range mechanism.
At a high level, we will reduce calculating the burn rule of a Myerson-in-Range mechanism to computing equilibrium payments in multi-item single-buyer environments.

We show that the virtual utility maximization problem faced by the miner is similar to the utility maximization problem faced by a monopsonist in multi-item settings.
In \autoref{thm:MyerImpliesVirtualUtilityOpt}, we argued that a Bayesian miner chooses an outcome from the menu of user allocation probabilities and burns specified by $\Feasibility$ in the off-chain game to optimize her virtual utility.
In other words, the miner ``buys'' an outcome supported by $\Feasibility$ by making a ``payment'' equal to the burn corresponding to the outcome.
In return, she derives a ``surplus'' equal to the virtual surplus of the allocation.
Concretely, the items in the multi-item single-buyer environment are analogous to users in the TFM environment and the monopsonist's values for the items correspond to the virtual values of the users and are drawn iid from $\Distr^{\vv}$, the distribution of virtual values when the value is drawn from $\Distr$.

From the reduction discussed above, computing the equilibrium allocation and burn of a Myerson-in-Range mechanism corresponds to computing the equilibrium allocation and payments in the analogous multi-item single-buyer environment.
Since we apply the revelation principle in the TFM environment, the value allocation and burn rules $(\OnCAlloc, \OnCBurn)$ will correspond to DSIC mechanisms for the monopsonist.
However, there are some subtle differences between the two settings.
In the multi-item single-buyer environment, the monopsonist should be disincentivized to deviate from truth-telling for all value profiles over the items.
However, in the TFM environment, an off-chain equilibrium and the corresponding direct-revelation mechanism is Myerson-in-Range even if the miner would like to deviate from the equilibrium behaviour with probability zero.
We will address the discrepancy in the definitions between the two environments by arguing that the allocation and the burn rules of a Myerson-in-Range mechanism can be \emph{smoothened} over some measure zero collection of value profiles so that the miner pointwise optimizes her virtual utility for all value profiles $\vec{v}$ in the smoothened mechanism.\footnote{Note that the smoothened mechanism might not necessarily be feasible, i.e, be supported over $\Feasibility$. However, we will use the smoothened mechanism to derive a burn identity which holds almost surely for the original mechanism.}

\begin{theorem}
\label{thm:SmootheningVirtualUtilityMaximization}
    Let the users' values be drawn from a smooth regular distribution $\Distr$.
    Suppose that the direct-revelation mechanism $(\OnCAlloc, \OnCPay, \OnCBurn)$ is Myerson-in-Range for some block-building process $\bBuild$.    
    Then, there exists an allocation rule $\widetilde{\AllocRule}$ and a burn rule $\widetilde{\BurnB}$ such that $(\widetilde{\AllocRule}, \widetilde{\BurnB})$ is identical to $(\OnCAlloc, \OnCBurn)$ almost surely for any number $n$ of users and for any profile of values $\vec{v}, \vec{w} \in \supp(\Distr^n)$,
        $$\sum_{i = 1}^n \vv(v_i) \, \widetilde{\AllocRule}_i(\vec{v}) - \widetilde{\BurnB}(\vec{v}) \geq \sum_{i = 1}^n \vv(v_i) \, \widetilde{\AllocRule}_i(\vec{w}) - \widetilde{\BurnB}(\vec{w}).$$
\end{theorem}

We will say that such a pair $(\widetilde{\AllocRule}, \widetilde{\BurnB})$ is a \emph{monopsonist smoothening} of $(\OnCAlloc, \OnCBurn)$.

\begin{definition}[Monopsonist smoothening]
    An allocation and burn rule $(\widetilde{\AllocRule}, \widetilde{\BurnB})$ is a \emph{monopsonist smoothening} of a Myerson-in-Range mechanism $(\OnCAlloc, \OnCBurn)$ for a distribution $\Distr$ and a block-building process $\bBuild$ if (a) $(\widetilde{\AllocRule}, \widetilde{\BurnB})$ is identical to $(\OnCAlloc, \OnCBurn)$ except on a set of probability measure zero for any number $n$ of users and (b) $\sum_{i = 1}^n \vv(v_i) \, \widetilde{\AllocRule}_i(\vec{v}) - \widetilde{\BurnB}(\vec{v}) \geq \sum_{i = 1}^n \vv(v_i) \, \widetilde{\AllocRule}_i(\vec{w}) - \widetilde{\BurnB}(\vec{w})$ for all $\vec v, \vec{w} \in \supp(\Distr^n)$.
\end{definition}

We defer the proof of \autoref{thm:SmootheningVirtualUtilityMaximization} to \autoref{sec:ProofofSmoothening}.

The monopsonist smoothening enables the application of results from existing literature on multi-item auctions to transaction fee mechanism design.
For instance, for a given utility function $U(\vec{v})$ mapping the buyer's value profile to her utility\footnote{Such a utility function can be derived even for non-truth-telling equilibria via the revelation principle (see \citealp{Rochet85}, for example).}, \citet{Rochet85} argues that the allocation rule and payment rule of the DSIC mechanism that induces the utility function $U$ is essentially unique.

\begin{theorem}[\citealp{Rochet85}] \label{thm:Rochet}
    A function $U : \R^n  \xrightarrow{} \R$ is the utility function of some DSIC mechanism in an $n$-item single-buyer environment if and only if $U$ is convex and non-decreasing.
    The allocation rule of such a DSIC mechanism is given by $\OnCAlloc(\vec{v}) = \nabla U(\vec{v})$ and the payment charged to the buyer equals $\Pay(\vec{v}) = \vec{v}^{\,\intercal} \, \nabla U(\vec{v}) - U(\vec{v})$.
\end{theorem}

Note that when $U$ is convex, $\nabla U$ is defined almost everywhere, and the allocation and payment rules $\OnCAlloc$ and $\OnCPay$ are uniquely defined at all such points with a well-defined gradient.
At values where $\nabla U$ is not well-defined, setting $\OnCAlloc$ to be any sub-derivative of $U$ satisfies dominant-strategy incentive compatibility.

Applying \autoref{thm:SmootheningVirtualUtilityMaximization} in conjunction with \autoref{thm:Rochet} for a virtual utility optimizing miner yields the following.

\begin{corollary} \label{thm:VirtualCondition1}
    Consider a Myerson-in-Range mechanism $(\OnCAlloc, \OnCBurn)$ for a block-building process $\bBuild$ and a smooth regular distribution $\Distr$.
    If, $$\OnCUtil(\vv(\vec{v})) = \sum_{i = 1}^n \vv(v_i) \, \OnCAlloc_i(\vec{v}) - \OnCBurn(\vec{v}) \geq \sum_{i = 1}^n \vv(v_i) \, \OnCAlloc_i(\vec{w}) - \OnCBurn(\vec{w})$$
    almost surely for all $\vec{v} \sim \Distr^n$ and $\vec{w} \sim \supp(\Distr^n)$, then, there exists a function $U$ convex and non-decreasing in virtual values $\vv(\vec{v})$ such that 
    \begin{align*}
      U(\vv(\vec{v})) = \overline{U}(\vv(\vec{v})), \hspace{1cm} \OnCAlloc(\vec{v}) = \nabla^{\vv} U(\vv(\vec{v})) \hspace{0.5cm}\text{ and } \hspace{0.5cm} \OnCBurn(\vec{v}) = \sum_{i = 1}^n \vv(v_i) \, \nabla^{\vv}_i U(\vv(\vec{v})) - U(\vv(\vec{v}))
    \end{align*}
    with probability $1$ for $\vec{v} \in \supp(\Distr^n)$, where $\nabla^{\vv}$ is the gradient with respect to the virtual values $\vv(\vec{v})$.
\end{corollary}
We say that the function $U$ is the \emph{smoothened virtual utility function} of the mechanism $\DirReveal$.

\autoref{thm:VirtualCondition1} expresses condition \ref{Bul:1} from \autoref{thm:MainReduction} in terms of the smoothened virtual utility.
We can also similarly rewrite conditions \ref{Bul:2} and \ref{Bul:3}. 
Condition \ref{Bul:2} states that the mechanism should almost never allocate a user with a virtual value smaller than zero.
Thus, the gradient $\nabla_i^{\vv} U(\vv(\vec{v})) = \OnCAlloc_i(\vec{v}) = 0$ almost surely whenever $\vv(v_i) < 0$.
Since $U$ is convex in $\vv(\vec{v})$, $\nabla^{\vv}_i U(\vv(\vec{v}))$ is increasing in $\vv(\vec{v})$ and thus, $\nabla_i^{\vv} U(\vv(\vec{v})) = 0$ for all $v_i$ with a strictly negative virtual value.

Condition \ref{Bul:3} suggests that fabricating or censoring a bid with a non-positive virtual value should almost never change the miner's virtual utility.
Thus, for value profiles $\vec{v} = (v_1, \dots, v_n, v_{n+1}, \dots, v_{n+t})$ and $\vec{w} = (v_1, \dots, v_n)$ such that $\vv(v_{n+i}) \leq 0$ for all $1 \leq i \leq t$, the smoothened virtual utility function at $\vec{v}$ and $\vec{w}$ must satisfy $U(\vec{v}) = U(\vec{w})$.

We summarize the discussions above to state the burn identity.

\begin{theorem}[Burn identity] \label{thm:UtilityVersionMainReduction}
    Let $(\OnCAlloc, \OnCBurn)$ be the value allocation and burn rule of a Myerson-in-Range mechanism for a smooth regular distribution $\Distr$.
    Then, there exists a family of smoothened virtual utility functions $\seq{U^n}{n \in \N}$, $U^n: \R^n \xrightarrow{} \R$, such that for all $n \in \N$,
    \begin{enumerate}
        \item $U^n$ is convex and non-decreasing as a function of the virtual values of the bids,
        \item $\nabla_i^{\vv} U^n (\vv(\vec{v})) = 0$ whenever $\vv(v_i) < 0$,
        \item $U^n(\vv(\vec{v})) = U^{n+1}(\vv(\vec{w}))$ for all $\vec{v}, \vec{w}$ such that
        $\vec{v} = (v_1, \dots, v_n, v_{n+1}, \dots, v_{n+t}) \in \supp(\Distr^{n+t})$ and $\vec{w} = (v_1, \dots, v_n)$ for $\vv(v_{n+i}) \leq 0$ for all $1 \leq i \leq t$, and,
        \item With probability $1$ for $\vec{v} \sim \Distr^n$,
        $$\OnCAlloc(\vec{v}) = \nabla^{\vv} U^n(\vv(\vec{v})) \text{ and } \OnCBurn(\vec{v}) = \sum_{i = 1}^n \vv(v_i) \, \nabla^{\vv}_i U(\vv(\vec{v})) - U(\vv(\vec{v})).$$

    \end{enumerate}
\end{theorem}
Finally, we will combine the burn identity with the payment identity (\autoref{item:payment-identity} in \autoref{thm:myerson}) to derive a DSIC payment rule in terms of the smoothened virtual utility function.
From the payment identity,
$$\OnCPay_i(\vec{v}) = v_i \OnCAlloc_i(\vec{v}) - \int_0^{v_i} \OnCAlloc_i(\upsilon, \vec{v}_{-i}) \,d \upsilon.$$
The smoothened virtual utility function determines the allocation rule everywhere except at a measure zero set of points.
Further, since the distribution $\Distr$ is regular, all points between the infimum and supremum of $\supp(\Distr)$ have a positive probability density and as a consequence, even when $\OnCAlloc$ differs from $\nabla^{\vv} U$ at a measure zero set of points,
$\int_0^{v_i} \OnCAlloc_i(\upsilon, \vec{v}_{-i}) \,d \upsilon = \int_0^{v_i} \nabla^{\vv}_i U(\vv(\upsilon), \vv(\vec{v}_{-i})) \,d \upsilon$,
almost everywhere.
Thus, the payment rule of a Myerson-in-Range mechanism with a smoothened utility function $U$ equals
$$\OnCPay_i(\vec{v}) = v_i \nabla^{\vv}_i U(\vv(\vec{v})) - \int_0^{v_i} \nabla^{\vv}_i U(\vv(\upsilon), \vv(\vec{v}_{-i})) \,d \upsilon$$
with probability $1$.

To conclude our discussion on Myerson-in-Range mechanisms, by \autoref{thm:MyersonEnvelopeIffOffCIP}, observe that a direct-revelation mechanism is off-chain influence proof only if it has a trivial off-chain component and its burn rule satisfies the burn identity.

The reduction from TFMs to multi-item auctions characterizing off-chain influence proof mechanisms can be extended to also characterize global strong collusion proof mechanisms.
In \autoref{sec:MIRC}, we use our characterization to design a mechanism that is on-chain simple and global strong collusion proof which yields a positive revenue to the miner.\footnote{
  \cite{ChungRS24, GafniY24} show that no such deterministic mechanism exists that yields positive miner revenue.
  We construct a randomized mechanism with positive miner revenue.
  Our mechanism applies for both infinite and bounded block size, although for a finite block, our mechanism requires the prior distribution to have a known upper bound (and it remains an open question whether such mechanisms exist without a known upper bound).
}

%% file: 50-priorindependentimpossibilitiesandresults.tex
\section[Prior-Independent OnC-MS and OffCIP TFMs Require Both Cryptography and Advice]{Prior-Independent On-Chain Miner Simple and Off-Chain Influence Proof Mechanisms Require Both Cryptography and Miner Advice} \label{sec:PriorIndependent}

We now use the burn identity derived in \autoref{sec:MonopsonistLens} to reason about desirable TFMs in various settings.
Going forward, we focus on TFMs that satisfy all the desirable properties highlighted by \citet{GaneshTW24}, and call these mechanisms \emph{simple to participate}.

\begin{definition}
    A direct-revelation mechanism $\DirReveal$ is \emph{simple to participate}
    if it is on-chain user simple, on-chain miner simple, and off-chain influence proof.
\end{definition}

In this section, we consider prior independent mechanisms, i.e., ones where the block-building process $\bBuild$ has no knowledge of the distribution of values $\Distr$. 
\citet{GaneshTW24} show that there exists a cryptographic variant of the second-price auction that is prior-independent and simple to participate for all regular distributions $\Distr$.
However, their auction requires two features: heavy-duty cryptography, and allowing the miner to inform the mechanism through an advice (namely, allowing the miner to set a reserve).
We show in this section that \emph{both} of these features are necessary to get a nontrivial prior independent simple-to-participate TFM.

First, we show that the only (possibly cryptographic) prior-independent mechanism that does not rely on any advice from the miner and is simple to participate for a sufficiently rich class of distributions is the trivial mechanism that allocates users with zero probability irrespective of the distribution.
In order to satisfy condition \ref{Bul:2} of \autoref{thm:MainReduction}, an off-chain influence proof mechanism should be able to discriminate between users with values smaller than the monopoly reserve and the rest only using information available to the block-building process $\bBuild$, so as to allocate the former with zero probability.
However, $\bBuild$ is prior-independent and is therefore agnostic to the monopoly reserve of $\Distr$.\footnote{The reason we require on-chain user simplicity for \autoref{thm:NoAdvice} is subtle. Suppose the mechanism is not on-chain user simple and the users adopt different strategies for different distributions. The mechanism can potentially infer the value distribution based on the bids submitted by the users. However, when the mechanism is on-chain user simple, there is truly no information regarding the distribution that can be inferred from the users' behaviour. The users will always bid their values in equilibrium. We leave as an open problem whether the same impossibility can be proved purely via off-chain influence proofness, without having to rely on on-chain user simplicity.}

\begin{theorem} \label{thm:NoAdvice}
    Consider the class $\RegDist$ of smooth regular distributions.
    Then, an on-chain user simple and off-chain influence proof mechanism that is both prior-independent and advice-independent for all $\Distr \in \RegDist$ must be the trivial mechanism that never allocates any user irrespective of the block capacity.
\end{theorem}

We prove \autoref{thm:NoAdvice} in \autoref{sec:ProofofNoAdvice}.

Second, we also rule out plain-text mechanisms that are simple to participate for a sufficiently rich class of distributions.
At a high level, suppose there exists a distribution $\Distr$ and a value profile $\vec{v}$ for which the miner is awarded a revenue larger than the block reward from building an empty block.
For another distribution $\hat{\Distr}$ with a monopoly reserve much larger than all values in $\vec{v}$, by condition \ref{Bul:2} from \autoref{thm:MainReduction}, all users must receive no allocation when their value profile equals $\vec{v}$.
However, the miner can deviate from the equilibrium behaviour, submit the advice $\adv^{\Distr}$ corresponding to $\Distr$ and trick the mechanism into believing that the users' values are drawn from $\Distr$.
The miner will then receive a revenue larger than the block reward, contradicting on-chain miner simplicity.

\begin{theorem} \label{thm:NoCrypto}
    Consider the class $\RegDist$ of smooth regular distributions $\Distr$ with no point masses, i.e, there exists no value $\hat v$ such that $\Pr_{v \sim \Distr}[v = \hat{v}] > 0$.
    Then, a simple-to-participate plain-text prior-independent mechanism for all $\Distr \in \RegDist$ must be the trivial mechanism that almost never allocates any user irrespective of the block capacity or the distribution $\Distr$.
\end{theorem}

We defer the proof to \autoref{sec:ProofNoCrypto}

%% file: 60-priordependentimpossibilitiesandresults.tex
\section[Prior-Dependent, OnC-US, OnC-MS, and OffCIP TFMs]{Prior-Dependent, On-Chain User and Miner Simple, Off-Chain Influence Proof, Plain-Text Mechanisms} \label{sec:ImpossibilitesandPossibilities}

In this section, we consider the case of prior-dependent mechanisms, i.e., we assume the block-building process $\bBuild$ knows the distribution $\Distr$ of user values. 
We focus on plain-text TFMs.
While the two main desiderata of \autoref{sec:PriorIndependent} are no longer relevant for this section,\footnote{
    \autoref{sec:PriorIndependent} considers whether the TFM needs to use miner advice, and whether it needs to use cryptography.
    A prior-dependent TFM does not require miner advice, since (for example) the block-building process itself can now know the right reserve price to set based on the distribution $\Distr$.
    Similarly, the simple-to-participate cryptographic variant of the second-price auction discussed by \citet{GaneshTW24} for the prior-independent setting can easily be modified for the prior-dependent setting --- by simulating the miner's advice for the distribution $\Distr$.
    So, we focus solely on the plain-text case.
} other considerations emerge from the analysis as important.
Specifically, we explore whether simple-to-participate mechanisms exist even for finite blocks, whether randomization is necessary for a TFM and finally, whether the distribution can have an arbitrary unbounded support.
While we uncover additional TFMs in some settings, our main finding is that \emph{posted-price mechanisms} (for an infinite block) as discussed in \autoref{sec:PostedPrice} are the only mechanisms that are deterministic, plain-text and simple to participate.

Throughout this section, we find it convenient to rank users $(1), \dots (n)$ in descending order of bids.
For notational convenience, we will denote the quantities corresponding to the $i$\textsuperscript{th} largest bid by $(i)$ (for example, $v^{(i)}$, $\OnCAlloc^{(i)} = \nabla_{(i)}^{\vv} U^n$, etc).

%% file: 61-priordependentconverse.tex
\subsection[A Sufficient Condition for Prior-Dependent OnCUS and OffCIP Mechanisms]{A Sufficient Condition for Prior-Dependent On-Chain User Simple and Off-Chain Influence Proof Mechanisms}
\label{sec:OnCUS-plus-OffCIP}

Before constructing simple-to-participate mechanisms, it would be useful to have a sufficient condition that ensures that a candidate mechanism $\DirReveal$ is indeed on-chain user simple and off-chain influence proof.
We will begin with such a sufficient condition very similar to the three properties in \autoref{thm:MainReduction}.

The three properties in \autoref{thm:MainReduction} do not characterize Myerson-in-Range mechanisms due to the following reason --- the three conditions do not automatically guarantee that, for some block-building process $\bBuild = (\AllocRule, \Pay, \Burn)$, $\vv(\vec{v}) \, \OnCAlloc(\vec{v}) - \OnCBurn(\vec{v}) \geq \vv(\vec{v}) \, \AllocRule - \BurnB$ for all $(\AllocRule, \BurnB) \in \Feasibility$.
It only ensures that $\vv(\vec{v}) \, \OnCAlloc(\vec{v}) - \OnCBurn(\vec{v}) \geq \vv(\vec{v}) \, \AllocRule - \BurnB$ for $(\AllocRule, \BurnB) = (\OnCAlloc(\vec{w}), \OnCBurn(\vec{w}))$ for some value profile $\vec{w}$.
However, suppose that the mechanism is on-chain user simple and $(\OnCAlloc(\vec{v}), \OnCPay(\vec{v}), \OnCBurn(\vec{v})) = (\AllocRule(\vec{v}), \Pay(\vec{v}), \Burn(\vec{v}))$ for all $\vec{v}$.\footnote{Remember that the mechanism is prior-dependent and we are assuming that the miner does not submit any advice to the block-building process $(\AllocRule, \Pay, \Burn)$. Thus, $\AllocRule$, $\Pay$ and $\Burn$ take as input a profile of bids and output the allocations, payments and the burn respectively.}
Then, $\vv(\vec{v}) \, \OnCAlloc(\vec{v}) - \OnCBurn(\vec{v}) \geq \vv(\vec{v}) \, \AllocRule - \BurnB$ for all $(\AllocRule, \BurnB) \in \Feasibility$.
In this section, we argue ensuring that
the direct-revelation mechanism realizes all feasible outcomes in $\Feasibility$, along with the three conditions from \autoref{thm:MainReduction}, is sufficient for the mechanism to be Myerson-in-Range.

\begin{lemma} \label{thm:MainReductionConverse}
    Let $\Distr$ be a smooth regular distribution.
    For a prior-dependent block-building process $\bBuild = (\AllocRule, \Pay, \Burn)$, suppose a direct-revelation mechanism $\DirReveal$ satisfies $\DirReveal = (\AllocRule, \Pay, \Burn)$.
    Further, let $\DirReveal$ satisfy the three conditions from \autoref{thm:MainReduction} for any number $n$ of users and for all value profiles belonging to $\supp(\Distr^n)$.
    \begin{enumerate}[(A)]
        \item (Optimal for $n$ users.) For all $\vec{v}, \vec{w} \in \supp(\Distr^n)$, we have
        $$\sum_{i = 1}^n \vv(v_i) \, \OnCAlloc_i(\vec{v}) - \OnCBurn(\vec{v}) \geq \sum_{i = 1}^n \vv(v_i) \, \OnCAlloc_i(\vec{w}) - \OnCBurn(\vec{w}).$$
        \item (Negative $\vv$'s are suboptimal.) For $\vec{v} \in \supp(\Distr^{n})$, $\OnCAlloc_i(\vec{v}) = 0$ whenever $\vv(v_i) < 0$.
        \item (No censoring or fabricating.) For $\vec{v} = (v_1, \dots, v_n, v_{n+1}, \dots, v_{n+t}) \in \supp(\Distr^{n}) \times \Distr_{\vv \leq 0}^t$ and $\vec{w} = (v_1, \dots, v_{n})$,
        $$\sum_{i = 1}^{n} \vv(v_i) \OnCAlloc_i(\vec{v}) \allowbreak - \OnCBurn(\vec{v}) \allowbreak = \sum_{i = 1}^{n} \vv(v_i) \OnCAlloc_i(\vec{w}) \allowbreak - \OnCBurn(\vec{w}).$$
        Then, $\DirReveal$ is Myerson-in-Range.
    \end{enumerate}    
\end{lemma}
We defer the proof to \autoref{sec:ProofofMainReductionConverse}.

Note that the direct-revelation mechanism $\DirReveal$ satisfying the requirements in \autoref{thm:MainReductionConverse} is on-chain user simple and off-chain influence proof.
For on-chain user simplicity, observe that the block-building process takes in bids and already implements the payment rule corresponding to the DSIC direct-revelation mechanism.
For off-chain influence proofness, once again, equilibrium payments are charged from the users by the payment rule $\Pay = \OnCPay$ on-chain, and thus, the miner will not have to make any extraneous off-chain transfers to steer users towards the underlying BNE.
Thus, the Myerson-in-Range mechanism $\DirReveal$ has a trivial off-chain component and by \autoref{thm:MyersonEnvelopeIffOffCIP}, the mechanism is off-chain influence proof.

Remember that the payment identity in the multi-item environment (\autoref{thm:Rochet}) is both necessary and sufficient for a mechanism to be DSIC for a monopsonist.
Applying \autoref{thm:Rochet} for a virtual utility maximizing miner, we get a partial converse to \autoref{thm:UtilityVersionMainReduction}.

\begin{theorem}[Partial converse to the burn identity] \label{thm:VirtualPartialConverse}
    For a prior-dependent block-building process $\bBuild = (\AllocRule, \Pay, \Burn)$ and a direct-revelation mechanism $\DirReveal = (\AllocRule, \Pay, \Burn)$, suppose there exists a family of virtual utility functions $\seq{U^n}{n \in \N}$ 
    such that 
    \begin{align*}
        \notag
        \OnCAlloc(\vec{v}) = \nabla^{\vv} U^n(\vv(\vec{v})), \hspace{0.3cm} \OnCPay_i(\vec{v}) = v_i \nabla^{\vv}_i U(\vv(\vec{v})) - \int_0^{v_i} \nabla^{\vv}_i U(\vv(\upsilon),& \vv(\vec{v}_{-i})) \,d \upsilon \\
        \text{ and } \hspace{0.15cm} \OnCBurn(\vec{v}) = \sum_{i = 1}^n \vv(v_i) \, \nabla^{\vv}_i U(\vv(\vec{v})) - U(\vv(\vec{v}))
    \end{align*}
    satisfying
\begin{enumerate}
    \item $U^n$ is convex and non-decreasing as a function of the virtual values of the bids,
    \item $\nabla_i^{\vv} U^n (\vv(\vec{v})) = 0$ whenever $\vv(v_i) < 0$, and,
    \item $U^n(\vv(\vec{v})) = U^{n+t}(\vv(\vec{v}), \vv(\hat{v}))$ for all $\hat{v} = (\hat{v}_{n+1}, \dots, \hat{v}_{\hat{n}})$ such that $\vv(\hat v_{n+i}) \leq 0$.
\end{enumerate}
for all $n \in \N$.
Then, $\DirReveal$ is on-chain user simple and off-chain influence proof.
\end{theorem}

%% file: 62-deterministicmechanisms.tex
\subsection{Deterministic Simple-to-Participate Mechanisms} \label{sec:NoDeterministicMechanisms}

In this section, we argue that the only deterministic plain-text mechanism that is simple to participate is the class of all uniform posted-price mechanisms.\footnote{Remember that, by definition, we require our TFMs to be anonymous. Strictly speaking, a deterministic, anonymous, virtual utility maximizing mechanism is not possible for a block with a finite capacity. For example, for a block of capacity $1$, if there are two users with equal values, both larger than the monopoly reserve, the mechanism either has to randomize or has to break ties in the favour of one bidder. However, we will consider mechanisms that are deterministic and anonymous up to tie-breaking. Our impossibility results are agnostic to how ties are broken, and do not rely on the edge-cases requiring tie-breaking.}

\begin{theorem} \label{thm:NoDeterministicMechanisms}
    Let $\Distr$ be a smooth regular distribution.
    Then, a deterministic simple-to-participate mechanism allocates all bids above some threshold $p$ at least the monopoly reserve, charges $p$ from each allocated user and burns $\vv(p) \geq 0$ per included bid apart from awarding a block reward to the miner independent of the contents of the created block.
\end{theorem}

As a first step towards proving \autoref{thm:NoDeterministicMechanisms}, we will show that a deterministic off-chain influence proof mechanism has a simple burn rule parametrized only by the number of users included in the block.
Concretely, we will argue that there exists a sequence of burns $\seq{\Threshold_t}{t \in \N}$ such that $\Threshold_t$ is burnt whenever the miner builds a block containing $t$ users.
As a sanity check, note that the block reward equals the net revenue from building an empty block, which in turn equals $-\Threshold_0$.

\begin{lemma} \label{thm:CharacterizeDetMechanism}
    For a smooth regular distribution $\Distr$ and any deterministic off-chain influence proof mechanism, there exists a sequence of burns $\seq{\Threshold_t}{t \in \N}$ such that for a value profile $\vec{v} = (v^{(1)}, \dots, v^{(n)})$, the direct-revelation mechanism $(\OnCAlloc, \OnCPay, \OnCBurn)$ almost surely includes users $(1), \dots, (t)$ for
    $$t = \argmax_{\hat{t}} \sum_{i = 1}^{\hat{t}} \vv(v^{(i)}) - \Threshold_{\hat t}$$
    and burns $\Threshold_t$ for including exactly $t$ users.
\end{lemma}

Intuitively, the proof argues that from the monopsonist lens, a deterministic, item-symmetric (i.e, anonymous) mechanism to sell multiple items to a single buyer can only post a price for each bundle of items, the price depending only on the number of items in the bundle.
We defer the proof to \autoref{sec:ProofofCharDetMech}.

For the remainder of the section, we will show that $\Threshold_t = t \, (\Threshold_1 - \Threshold_0) + \Threshold$ for simple-to-participate deterministic mechanisms.
In other words, the \emph{marginal burn} $\MargThreshold_t := \Threshold_t - \Threshold_{t-1}$, i.e, the burn for including the $t$\textsuperscript{th} user conditioned on including $t-1$ users, must be a constant $\MargThreshold$.
Proving a constant marginal burn will conclude the proof of \autoref{thm:NoDeterministicMechanisms} for $\vv(p) = \MargThreshold$ and $p = \vv^{-1}(\MargThreshold)$.

To begin, observe that we can assume without loss of generality that $\Threshold_t \leq t \, \sup \Distr + \Threshold_0$.
Supposing that $\Threshold_t > t \sup \Distr + \Threshold_0$, the mechanism never allocates exactly $t$ users.
The virtual utility from creating a block with $t$ users is dominated by that of creating an empty block.
$$\sum_{i = 1}^t \vv(v^{(i)}) - t \sup \Distr - \Threshold_0 \leq - \Threshold_0.$$
However, note that the virtual utility from the empty block continues to (weakly) dominate including exactly $t$ users even if we decrease $\Threshold_t$ to $t \sup \Distr + \Threshold_0$.
Thus, setting $\Threshold_t = t \sup \Distr + \Threshold_0$ does not change the allocation rule, and therefore, leaves the payment and the burn rules unchanged too.
As a consequence, we can assume that $\Threshold_t \in [\Threshold_0, t \sup \Distr + \Threshold_0]$.

To show that the marginal burn $\MargThreshold_t$ is a constant, we will first argue that the marginal burn is non-increasing in $t$  (\autoref{thm:DecreasingMargBurns}), and then, for a non-increasing sequence $\seq{\MargThreshold_t}{t \in \N}$, we will show that they cannot decrease either (\autoref{thm:IncreasingMargBurns}).
This is similar to the approach adopted by \citet{GafniY24Discrete} in proving  their Theorem 3.9 (which establishes a similar result for deterministic global strong collusion proof mechanisms). 
However, while these properties follow fairly easily in the setting of  \citet{GafniY24Discrete}, adapting them to our setting poses significant additional challenges. 
For example, to check whether a mechanism is on-chain miner simple, we will have to calculate the change in burn and the change in payments for various deviations by the miner.
While the burns are naturally expressed in virtual value space (\autoref{thm:UtilityVersionMainReduction}), they are to be compared to users' payments expressed in value space (\autoref{thm:PaymentIdentity}).

\begin{lemma} \label{thm:DecreasingMargBurns}
    Let $\Distr$ be a smooth regular distribution.
    Then, a deterministic TFM with marginal burns $\seq{\MargThreshold_t}{t \in \N}$ that is simple to participate must satisfy $\MargThreshold_t \geq \MargThreshold_{t+1}$.
\end{lemma}

We give a pictorial representation of the proof in \autoref{fig:DecreasingMarginalBurns}.
We provide a proof sketch below and defer the full proof to \autoref{sec:ProofofDecreasingMargBurns}.

\begin{proof}[Proof sketch]
    Let $t$ be the smallest number for which $\MargThreshold_t < \MargThreshold_{t+1}$.
    Consider the value profile $\vec{v}$ such that $\vv(v^{(i)}) = \MargThreshold_i + \varepsilon$ for $\varepsilon > 0$.
    Choose $\varepsilon$ sufficiently small so that $\MargThreshold_t < \vv(v^{(t)}) < \MargThreshold_{t+1}$.
    We will show that the miner can increase her revenue by fabricating a bid $\hat{v}^{(t+1)}$ just smaller than $v^{(t)}$, with virtual value $\vv(\hat{v}^{(t+1)}) \in (\MargThreshold_t, \vv(v^{(t)}))$.

    Irrespective of whether the miner fabricates $\hat{v}^{(t+1)}$, note that the virtual utility maximizing allocation includes users $(1), \dots, (t)$.
    By construction, the virtual surplus $\vv(v^{(i)})$ from allocating user $(i)$ is at least the marginal burn $\MargThreshold_i$, and thus, the miner only increases her virtual utility by allocating all $t$ users.
    However, $\vv(\hat{v}^{(t+1)}) < \MargThreshold_{t+1}$, and thus, $\hat{v}^{(t+1)}$ is not included.
    The allocation and the burn remains constant before and after the miner inserts $\hat{v}^{(t+1)}$.

    However, observe that the payment charged to user $(t)$ strictly increases after the miner fabricates $\hat{v}^{(t+1)}$.
    With only $(t)$ bids, user $(t)$ is allocated as long as $v^{(t)} \geq \vv^{-1}(\MargThreshold_t)$.
    On the other hand, with $t+1$ bids, user $(t)$ is allocated precisely when $v^{(t)} \geq \hat{v}^{(t+1)} > \vv^{-1}(\MargThreshold_t)$.
    Therefore, the critical bid for user $(t)$ increases and as a result, the miner extracts a larger payment from $(t)$.
    In the full proof in \autoref{sec:ProofofDecreasingMargBurns}, we also argue that the critical bids of users $(1), \dots, (t-1)$ do not decrease after the miner fabricates $\hat{v}^{(t+1)}$ (intuitively, the $(t+1)$\textsuperscript{th} bid only increases the competition for the $t$ slots, and the critical bids will only increase).
    The miner strictly improves her revenue by inserting the fake bid $\hat{v}^{(t+1)}$ and thus, the mechanism cannot be on-chain miner simple.
\end{proof}

\begin{figure}
    {
    \centering
    \begin{subfigure}[b]{0.48\textwidth}
        \includegraphics[width=0.98\linewidth]{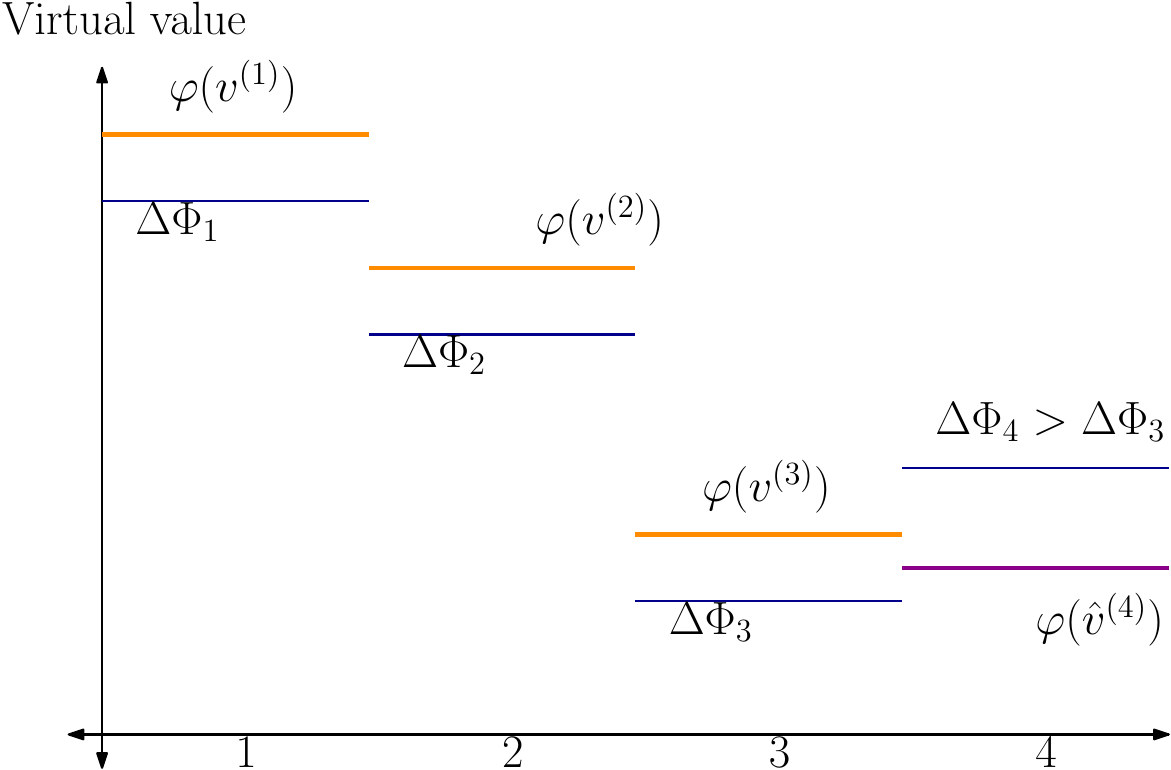}
        \caption{Illustration of the proof of \autoref{thm:DecreasingMargBurns}}
        \label{fig:DecreasingMarginalBurns}
    \end{subfigure}
        \begin{subfigure}[b]{0.48\textwidth}
        \includegraphics[width=0.98\linewidth]{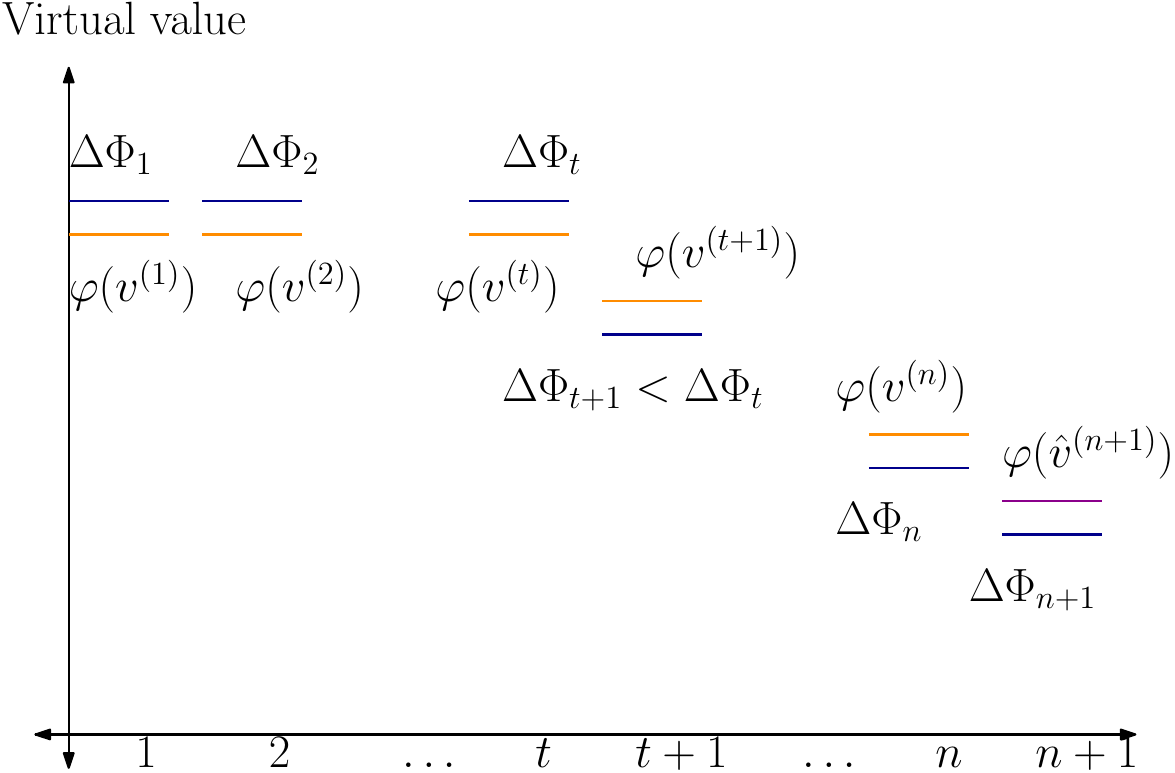}
        \caption{Illustration of the proof of \autoref{thm:IncreasingMargBurns}}
        \label{fig:IncreasingMarginalBurns}
    \end{subfigure}   
    \caption{ Main arguments used to prove \autoref{thm:NoDeterministicMechanisms}.}
    \label{fig:NoDeterministicMechanisms}
    }
    {\footnotesize \textbf{Notes:} The marginal burn $\seq{\MargThreshold_t}{t \in \N}$ must be a constant for a deterministic simple-to-participate mechanism. In \autoref{fig:DecreasingMarginalBurns}, we argue that $\MargThreshold_t \geq \MargThreshold_{t+1}$ for all $t \in \N$. For example, suppose $\MargThreshold_3 < \MargThreshold_4$. For a value profile $\vec{v}$ such that $\vv(v^{(i)}) > \MargThreshold_i$ for $1 \leq i \leq3$, the miner can increase her revenue by fabricating a bid $\hat{v}^{(4)}$ marginally smaller than $v^{(3)}$.
    Observe that with or without $\hat{v}^{(4)}$, all $3$ users will get allocated.
    Further, $\hat{v}^{(4)}$ is not allocated since $\vv(\hat{v}^{(4)}) < \MargThreshold_4$.
    Without $\hat{v}^{(4)}$, user $(3)$ gets included as long as $v^{(3)} \geq \vv^{-1}(\Threshold_3)$. Thus, user $(3)$'s critical bid equals $\vv^{-1}(\MargThreshold_3)$.
    Once the miner fabricates $\hat{v}^{(4)}$, $v^{(3)}$ needs to be larger than $\hat{v}^{(4)}$ to be included, and thus is charged $\hat{v}^{(4)} > \vv^{-1}(\MargThreshold_3)$, increasing the miner's revenue and thereby contradicting on-chain miner simplicity.
    In \autoref{fig:IncreasingMarginalBurns} we argue that the marginal burns cannot decrease either. Suppose $\MargThreshold_t > \MargThreshold_{t+1}$. Let $\vv(v^{(i)}) = \MargThreshold_i - \varepsilon$ for $1 \leq i \leq t$. For a very large $n$ and small $\delta > 0$, $\vv(v^{(i)}) = \MargThreshold_i + \delta$ for $t+1 \leq i \leq n$. Choose $\delta$ such that the virtual utility $\sum_{i = 1}^n \vv(v^{(i)}) - \sum_{i = 1}^n \MargThreshold_i - \Threshold_0$ from including all $n$ users is smaller than the virtual utility $- \Threshold_0$ from creating the empty block. None of the $n$ users will get allocated.
    Consider the miner fabricating $\hat{v}^{(n+1)}$ such that all $n+1$ bids get allocated, i.e, $\sum_{i = 1}^n \vv(v^{(i)}) + \vv(\hat{v}^{(n+1)}) - \sum_{i = 1}^{n+1} \MargThreshold_i - \Threshold_0 > - \Threshold_0$. For a large enough $n$, the payments collected from the $n$ users will outweigh the burn $\sum_{i = 1}^{n+1} \MargThreshold_i$, contradicting on-chain miner simplicity.
    \par}
\end{figure}

\begin{lemma} \label{thm:IncreasingMargBurns}
    Let $\Distr$ be a smooth regular distribution.
    A deterministic simple-to-participate TFM with a sequence of non-increasing marginal burns $\seq{\MargThreshold_t}{t \in \N}$ must satisfy $\MargThreshold_t = \MargThreshold$ for all $t \in \N$. 
\end{lemma}

We give a pictorial proof in \autoref{fig:IncreasingMarginalBurns}.
We sketch a proof here and defer the formal proof to \autoref{sec:ProofofIncreasingMargBurns}.

\begin{proof}[Proof sketch]
    Let $t \in \N$ be such that $\MargThreshold_t > \MargThreshold_{t+1}$.
    At a high level, we will construct a value profile with $n$ users for $n \gg t$ such that the virtual surplus from allocating any subset of users is slightly smaller than the burn to include them, resulting in the mechanism creating an empty block.
    However, the miner will be able to fabricate a single bid such that, from the mechanism's view, the virtual surplus from allocating all $n+1$ bids is larger than $\sum_{i = 1}^{n+1} \MargThreshold_i$ and thus, all $n+1$ bids are allocated.
    We will show that, for a sufficiently large $n$, the payments charged to the users from the above deviation is larger than the total burn, resulting in a larger revenue for the miner.

    Construct the value profile $\vec{v}$ as follows.
    Set $v^{(1)}, \dots, v^{(t)}$ so that $\vv(v^{(i)}) = \MargThreshold_i - \varepsilon$.
    Further, choose $v^{(t+1)}, \dots, v^{(n)}$ such that $\vv(v^{(i)}) = \MargThreshold_i + \delta$ for $(n-t) \, \delta$ slightly smaller than $t \, \varepsilon$.
    By construction, $$\sum_{i = 1}^n \vv(v^{(i)}) = \sum_{i = 1}^n \MargThreshold_i - t \varepsilon + (n-t) \, \delta < \sum_{i = 1}^n \MargThreshold_i$$ and thus, no users are allocated.
    The miner receives only the block reward $-\Threshold_0$.

    Now suppose that the miner fabricates a bid $\hat{v}^{(n+1)} < v^{(n)}$ such that $\vv(\hat{v}^{(n+1)}) = \MargThreshold_{n+1} + t \varepsilon - (n-t) \, \delta$.
    Once again, by construction, $\sum_{i = 1}^{n} \vv(v^{(i)}) + \vv(\hat{v}^{(n+1)}) = \sum_{i = 1}^{n+1} \MargThreshold_i$, and thus, all $n+1$ bids are allocated.

    Also, note that the critical bid for each user $(i)$ equals 
    $\vv^{-1}(\vv(v^{(i)}))$.
    This follows since for any virtual value $\vv(v) < \vv(v^{(i)})$, $$\sum_{j = 1, j \neq i}^{n+1} \vv(v^{(i)}) + \vv(\hat{v}^{(n+1)}) + \vv(v) < \sum_{i = 1}^{n} \vv(v^{(i)}) + \vv(\hat{v}^{(n+1)}) = \sum_{i = 1}^{n+1} \MargThreshold_i$$ and thus, the empty block is created, and in particular, user $(i)$ is not allocated.

    Thus, by inserting the fake bid $\hat{v}^{(n+1)}$, the miner collects $\sum_{i = 1}^n \vv^{-1}(\vv(v^{(i)})) = \sum_{i = 1}^t \vv^{-1}(\MargThreshold_i - \varepsilon) + \sum_{i = t+1}^n \vv^{-1}(\MargThreshold_i + \delta)$ as payments, but loses $\sum_{i = 1}^{n+1} \MargThreshold_i$ as burn.
    The miner's net revenue equals
    $$\sum_{i = 1}^t \Big(\vv^{-1}(\MargThreshold_i - \varepsilon) - \MargThreshold_i\Big) + \sum_{i = t+1}^n \Big(\vv^{-1}(\MargThreshold_i + \delta) - \MargThreshold_i\Big) - \MargThreshold_{n+1}.$$
    Remember that $\vv^{-1}(\MargThreshold_i) \geq \MargThreshold_i$ and intuitively, the revenue diverges as $n \xrightarrow{} \infty$.
    In \autoref{sec:ProofofIncreasingMargBurns}, we will show that for a sufficiently large $n$, the above revenue is strictly positive.
\end{proof}

\begin{remark}
When $\MargThreshold_t = \MargThreshold$ for all $t \in \N$, the above proof breaks down since we cannot construct $v^{(t)}$ such that $\vv(v^{(t)}) = \MargThreshold_t - \varepsilon$ and $v^{(t+1)}$ such that $\vv(v^{(t+1)}) = \MargThreshold_{t+1} + \delta$.
Then, $v^{(t+1)}$ is larger than $v^{(t)}$, which cannot be the case since we rank the values $v^{(1)} > \dots > v^{(n)}$ in descending order.
\end{remark}

\begin{remark}
    Note that we assume there can exist an arbitrary influx of users participating in the TFM to prove \autoref{thm:IncreasingMargBurns} and thereby, \autoref{thm:NoDeterministicMechanisms}.
    However, if there exists some $N \in \N$ such that the number of users $n$ is at most $N$, then, there can exist deterministic, simple-to-participate TFMs with non-constant marginal burns.
    We give such an example for the exponential distribution for $N = 2$ in \autoref{sec:Bounded}.
    
    Similarly, assuming that the virtual value function is continuous is important for \autoref{thm:NoDeterministicMechanisms}.
    In \autoref{sec:Discontinuities}, we construct a distribution with a discontinuous virtual value function for which there exists a deterministic mechanism that is simple to participate for a non-constant marginal burn $\seq{\MargThreshold_t}{t \in \N}$.
\end{remark}

For a block with a finite capacity, if the number of users with values more than the threshold $p$ from \autoref{thm:NoDeterministicMechanisms} is larger than the capacity of the block, the mechanism cannot allocate all the users feasibly.
Thus, there exists no deterministic mechanism that is simple to participate for a finite block.

\begin{corollary}
\label{thm:impossibility-deterministic-finite}
    When the block has a finite capacity, there exists no non-trivial deterministic simple-to-participate mechanism for a smooth regular distribution $\Distr$.
\end{corollary}

%% file: 63-positionauctionwithburns.tex
\subsection{Position Auctions with Burn}
\label{sec:rand-auction}

Position auctions are mechanisms where the allocation probability for a user does not depend on the exact bid submitted by the user, but only on the rank of the user's bid amongst all other submitted bids.
We consider position auctions where $\Threshold_t$ is burnt to allocate $t$ users.

\begin{definition}[Position auctions]
    For a regular distribution $\Distr$, a \emph{position auction} is given by a sequence of non-increasing allocation probabilities $\seq{x^{(i)}}{i \in \N}$ and \emph{marginal burn per unit allocation} $\seq{\MargThreshold_i}{i \in \N}$.
    Conditioned on allocating $t$ users, the users are allocated with probabilities $x^{(1)}, \dots, x^{(t)}$ in descending order of their bids and $\Threshold_t = \sum_{i = 1}^t \MargThreshold_i \, x^{(i)} + \Threshold_0$ is burnt.
    For $\vec{v} \sim \Distr^n$, the mechanism almost surely maximizes the miner's virtual utility
    $$\max_t \sum_{i = 1}^t \big( \vv(v^{(i)}) - \MargThreshold_i \big) \, x^{(i)} - \Threshold_0 = \max_t \sum_{i = 1}^t \vv(v^{(i)}) \, x^{(i)} - \Threshold_t.$$
\end{definition}

Similar to deterministic mechanisms, we will argue that the marginal burn per unit allocation $\MargThreshold_t$ must be a constant $\MargThreshold$, independent of the number of allocated users $t$.

\begin{lemma} \label{thm:ConstMargBurnPerUnit}
    Let $\Distr$ be a smooth regular distribution.
    Then, the marginal burn per unit allocation $\seq{\MargThreshold_i}{i \in \N}$ of a simple-to-participate position auction is a constant sequence $\seq{\MargThreshold}{i \in \N}$.
\end{lemma}

We provide a proof sketch in \autoref{sec:ProofofConstMargBurnPerUnit}.
The proof is extremely similar to \autoref{thm:DecreasingMargBurns} and \autoref{thm:IncreasingMargBurns}.

In \autoref{thm:NoPositionAuction}, we show that there cannot exist simple-to-participate position auctions that are not uniform posted-price mechanisms for unbounded distributions, even for a block with infinite capacity. 
On the other hand, \autoref{thm:BoundedPositionAuction} shows the existence of simple-to-participate position auctions beyond posted-price mechanisms for bounded smooth regular distributions and finite blocks, and further, characterizes all such mechanisms.

\begin{theorem} \label{thm:NoPositionAuction}
For an unbounded smooth regular distribution $\Distr$, consider a position auction given by the allocation rule $\seq{x^{(i)}}{i \in \N}$ and a marginal burn per unit allocation $\MargThreshold$.
If the mechanism is simple-to-participate, then, $x^{(i)} = x$ for a constant $0 \leq x \leq 1$ for all $i \in \N$.
\end{theorem}
\begin{proof}
    If possible, let $t$ be the smallest natural number such that $x^{(t)} > x^{(t+1)}$.
    Then $x^{(1)} = \dots = x^{(t)} = x$.
    At a high level, when the users have an extremely high value $v$, the miner will be able to fabricate a single bid just below $v$ and increase the payments made by all $t$ users by a substantial margin (this substantial margin increases as a function of $v$).
    On the other hand, the additional burn from fabricating a bid equals $\MargThreshold \, x^{(t+1)}$ and is independent of $v$.
    For a sufficiently large $v$, the increase in payments will outweigh the additional burn, incentivizing the miner to fabricate bids.
    
    Formally, for a large $v \in \R_{\geq 0}$ and small $\varepsilon > 0$, consider $v^{(1)}, \dots, v^{(t)}$ all belonging to $[v, v+\varepsilon]$ (i.e, draw $v_1, \dots, v_t$ and re-arrange them in descending order).
    Indeed, the probability of realizing a value profile $\vec{v}$ with each $v_i \in [v, v+\varepsilon]$ is positive (since the distribution is regular and unbounded).

    We will calculate the payments made by the users.
    Each user $(i)$ is allocated $x$ for $1 \leq i \leq t$, and is allocated as long as the virtual value $\vv(v^{(i)})$ is at least $\MargThreshold$.
    Thus, for all $1 \leq i \leq t$, user $(i)$ makes a payment $\vv^{-1}(\MargThreshold) \, x$ to be included, resulting in a net revenue $t \, \vv^{-1}(\MargThreshold) \, x - t \, \MargThreshold \, x$.

    Suppose the miner fabricates $\hat{v}^{(t+1)} \in (v^{(t)}, v^{(t)} - \varepsilon]$.
    An additional $\MargThreshold \, x^{(t+1)}$ is burnt to include the fabricated bid $\hat{v}^{(t+1)}$.
    However, the payments made by $(1), \dots, (t)$ also increase.
    The users are allocated $x^{(t+1)} < x$ whenever $\vv^{-1}(\MargThreshold)\leq v^{(i)} \leq \hat{v}^{(t+1)}$ and $x$ when $v^{(i)} > \hat{v}^{(t+1)}$.
    By the payment identity (\autoref{item:payment-identity} in \autoref{thm:myerson}), the users pay $\vv^{-1}(\MargThreshold) \, x^{(t+1)} + \hat{v}^{(t+1)} \, (x^{(t)} - x^{(t+1)})$
    each.
    Thus, the net revenue equals
    $$t \, \Big[\vv^{-1}(\MargThreshold) \, x^{(t+1)} + \hat{v}^{(t+1)} \, (x^{(t)} - x^{(t+1)})\Big] - (t x + x^{(t+1)}) \, \MargThreshold.$$

    The increase in revenue from fabricating $\hat{v}^{(t+1)}$ equals
    $$t \, (x - x^{(t+1)}) \, (\hat{v}^{(t+1)} - \vv^{-1}(\MargThreshold)) - \MargThreshold \, x^{(t+1)}.$$
    We have $\hat{v}^{(t+1)} \geq v - \varepsilon$.
    For a large enough $v$, the revenue increases from fabricating $\hat{v}^{(t+1)}$, contradicting on-chain miner simplicity.    
\end{proof}

\begin{theorem} \label{thm:BoundedPositionAuction}
    For any bounded smooth regular distribution $\Distr$ and a block with a finite capacity $\Omega$, a position auction given by the allocation probabilities $\seq{x^{(t)}}{t \in \N}$ and marginal burn per unit allocation $\MargThreshold$ is on-chain miner simple if and only if (and thereby simple-to-participate if and only if)
    $\sum_{t = 1}^{\infty} x^{(t)} \leq \Omega$ and
    \begin{equation} \label{eqn:SufficientOnMS}
    t \, (x^{(t)} - x^{(t+1)}) \, (\sup \Distr - \vv^{-1}(\MargThreshold)) < \MargThreshold \, x^{(t+1)}
\end{equation}
for all $t \in \N$.
\end{theorem}

As we will see in the proof, the left hand side will be the increase in the miner's revenue when there are $t$ users with a value $\sup \Distr$ and the miner fabricates a bid $\hat{v}^{(t+1)}$ just below $\sup \Distr$.
The right hand side is the excess burn from inserting the fake bid $\hat{v}^{(t+1)}$.
Intuitively, the miner's ability to increase her revenue by fabricating bids is the highest when all users submit bids close to the supremum of the distribution, i.e, all users are willing to pay extremely high amounts but the revenue is small due to a lack of competition.
We defer the proof to \autoref{sec:ProofofBoundedPositionAuction}.

For the remainder of the section, we will construct a position auction that satisfies the conditions in \autoref{thm:BoundedPositionAuction} for a block capacity $\Omega = 1$.

\begin{example}[Position auction for block capacity $1$.]
\label{ex:rand-auction}
We will construct a position auction such that
(a) $\sum_{t = 1}^{\infty} x^{(t)} = 1$ and $t \, (x^{(t)} - x^{(t+1)}) \leq 2 x^{(t+1)}$ and (b) $\MargThreshold$ satisfies $(\sup \Distr - \vv^{-1}(\MargThreshold)) < \MargThreshold/2$, so that the conditions from \autoref{thm:BoundedPositionAuction} hold.
The latter condition is easy to satisfy.
As $\MargThreshold \xrightarrow{} \sup \Distr$, we have $\sup \Distr - \vv^{-1}(\MargThreshold) \xrightarrow{} 0$ and $\MargThreshold \xrightarrow{} \sup \Distr$.
Thus, for some $\MargThreshold$ sufficiently close to $\sup \Distr$, we have $(\sup \Distr - \vv^{-1}(\MargThreshold)) < \MargThreshold/2$.

For the former, set $x^{(t)} = \frac{1}{t \, (t+1)} \times x^{(1)}$ for $t > 1$.
We have 
$$\sum_{t = 1}^{\infty} x^{(t)} = x^{(1)} \Big( 1 +  \sum_{t = 1}^{\infty} \frac{1}{t \, (t+1)}\Big) = 2 x^{(1)}.$$
For $x^{(1)} = 1/2$, we see that the TFM can allocate any number of users even if the block capacity is $1$.
Indeed, $$t \, (x^{(t)} - x^{(t+1)}) = t \left(\frac{1}{t\, (t+1)} - \frac{1}{(t+1) \, (t+2)}\right) \, x^{(1)} = 2 \times \frac{1}{(t+1) \, (t+2)} = 2 \, x^{(t+1)}.$$

This concludes the construction of a position auction that is simple to participate for a finite block capacity.
\end{example}

%% file: 64-generalizedpositionauctions.tex
\subsection{Generalized Position Auctions}
So far, we have argued that the posted-price mechanism with a suitable burn is simple-to-participate for an infinite block.
In contrast, when the distribution is bounded, we showed that there exists position auctions with an allocation rule meaningfully different from that of the posted-price mechanism and are simple-to-participate even for a finite block.
In this section, we propose \emph{generalized position auctions} that are also meaningfully different from the posted-price mechanism, and are simple-to-participate for infinite blocks and value distributions with an unbounded support.
However, we show generalized position auctions cannot be simple to participate for finite blocks and thus, we will have to search in a much larger class of mechanisms to construct simple-to-participate mechanisms when the block capacity is bounded.

At a high level, the allocation rule for a user in a position auction depends only its rank when all bids are arranged in descending order.
This can be extended to a generalized position auction by choosing a class of single-agent allocation rules $\seq{x^{(t)}}{t \in \N}$, where user $(i)$ with the $i$\textsuperscript{th} largest bid and value $v^{(i)}$ is allocated with probability $x^{(i)}(v^{(i)})$.

\begin{definition}[Generalized position auctions]
    A \emph{generalized position auction} is given by a sequence of single-agent allocation rules $\seq{x^{(t)}}{t \in \N}$, where $x^{(t)}: \R_{\geq 0} \xrightarrow{} [0, 1]$.
    For a value profile $\vec{v}$, user $(i)$ receives an allocation $x^{(i)}(v^{(i)})$.
    The payment and the burn rules are chosen to satisfy the payment (\autoref{item:payment-identity} in \autoref{thm:myerson}) and burn identities (\autoref{thm:UtilityVersionMainReduction}) respectively.
\end{definition}

Note that for the allocation rule to be monotone, we require each $x^{(t)}$ to be monotone and further, $x^{(t)}(v) \geq x^{(t+1)}(v)$ at all values $v \in \R_{\geq 0}$.
Additionally, condition \ref{Bul:2} from \autoref{thm:UtilityVersionMainReduction} suggests that $x^{(t)}(v) = 0$ for all $t$ and values $v$ smaller than the monopoly reserve.

To begin, we will compute the payment and burn rules according to the payment and burn identities respectively so that the mechanism is on-chain user simple and off-chain influence proof.
For a value profile $\vec{v} \in \supp(\Distr^n)$, user $(i)$ is charged a payment
\begin{align}
    \OnCPay_i(\vec{v}) = v^{(i)} \, x^{(i)}(v^{(i)}) - \sum_{j = i}^n \int_{v^{(j+1)}}^{v^{(j)}} x^{(j)}(z) \,d z \label{eqn:PaymentGenPos}
\end{align}
(for convenience, we set $v^{(n+1)} = \vv^{-1}(0)$).
Similarly,
\begin{align}
    \OnCBurn(\vec{v}) = \sum_{i = 1}^n \vv(v^{(i)}) \, x^{(i)}(v^{(i)}) - \int_{0}^{\vv(v^{(i)})} x^{(i)}(z) \,d \vv(z) \label{eqn:BurnGenPos}
\end{align}
is burnt almost surely for $\vec{v} \sim \Distr^n$.
Thus, the miner's smoothened utility function is given by
\begin{equation}
    \notag
    \sum_{i = 1}^n \vv(v^{(i)}) \, x^{(i)}(v^{(i)}) - \Big(\sum_{i = 1}^n \vv(v_i) \, x^{(i)}(v^{(i)}) - \int_{0}^{\vv(v^{(i)})}x^{(i)}(z) \,d \vv(z) \Big) = \sum_{i = 1}^n \int_{0}^{\vv(v^{(i)})} x^{(i)}(z) \,d\vv(z).
\end{equation}
The partial derivative $\nabla^{\vv}_i \sum_{i = 1}^n \int_{0}^{\vv(v^{(i)})} x^{(i)}(z) \,d \vv(z) = x^{(i)}(v^{(i)})$, satisfying the burn identity.

Next, we will construct simple-to-participate generalized position auctions for an infinite block.
The strategy to show that the mechanism is indeed simple to participate is similar to \autoref{thm:BoundedPositionAuction}.

\begin{theorem} \label{thm:UnboundedPositionAuction}
    Let $\Distr$ be a smooth regular distribution with an unbounded virtual value function.
    Suppose that an on-chain user simple and off-chain influence proof generalized position auction with allocation rule $\seq{x^{(t)}}{t \in \N}$ satisfies
    \begin{align}
        t \, \Big(x^{(t)}(w) - x^{(t+1)}(w) \Big) \geq (t+1) \, \Big(x^{(t+1)}(w) - x^{(t+2)}(w) \Big) \label{eqn:AllocConverge}
    \end{align}
    for all $t \in \N$, and
    \begin{align}
        t \times \int_{\vv^{-1}(0)}^w \Big(x^{(t)}(z) - x^{(t+1)}(z) \Big) \,dz \leq \vv(w) \, x^{(t+1)}(w) - \int_0^{\vv(w)} x^{(t+1)}(z) \,d\vv(z) \label{eqn:GenPosIneq}
    \end{align}
    for all $t \in \N$ and $w \geq \vv^{-1}(0)$.
    Then, the mechanism is also on-chain miner simple.
\end{theorem}
We defer the proof to \autoref{sec:ProofofUnboundedPositionAuction}.

In \autoref{ex:GenPos}, we construct a simple-to-participate generalized position auction for smooth regular distributions.

\begin{example}[Simple-to-participate generalized position auction for unbounded distributions] \label{ex:GenPos}
    Let $\Distr$ be a smooth regular distribution with a value $\Gamma$ such that $\vv(\Gamma) > 0$.
    Consider the generalized position auction given by
    $$x^{(t)}(w) = \begin{cases}
        0 & \text{ for } w \in [0, \Gamma), \\
        1 - \frac{1}{2e^{(w - \Gamma)}} \sum_{i = 1}^t \frac{1}{i \, (i+1)} & \text{ for } w \in [\Gamma, \infty).
    \end{cases}$$
    We will begin by verifying \autoref{eqn:AllocConverge}.
    The inequality obviously holds for $w < \Gamma$.
    For $w \geq \Gamma$,
    $$t \times \Big(x^{(t)}(w) - x^{(t+1)}(w)\Big) = \frac{1}{2e^{(w - \Gamma)}} \cdot \frac{t}{(t+1)(t+2)}$$
    which is indeed decreasing in $t \in \N$.

    To verify \autoref{eqn:GenPosIneq}, once again note that the inequality holds trivially for $w < \Gamma$.
    For $w \geq \Gamma$, the left hand side can be upper bounded by
    $$\int_{\Gamma}^{\infty} t \times \Big(x^{(t)}(w) - x^{(t+1)}(w)\Big) \, dz = \frac{t}{2(t+1)(t+2)} \int_{\Gamma}^{\infty} \frac{1}{e^{(w - \Gamma)}} \, dz = \frac{t}{2(t+1)(t+2)}.$$
    We can lower bound the right hand side as follows.
    $$\vv(w) \, x^{(t+1)}(w) - \int_0^{\vv(w)} x^{(t+1)}(z) \,d\vv(z) = \int_0^{\vv(w)} x^{(t+1)}(w) - x^{(t+1)}(z) \,dz \geq \vv(\Gamma) \times \frac{1}{2}.$$
    The final inequality follows since $x^{(t)}(w) \geq 1/2$ for $w \geq \Gamma$.
    Choose $\Gamma$ such that $\vv(\Gamma) \geq 2$ so that the right hand side is at least $1$, which is clearly larger than the left hand side.\footnote{For the construction in this example to work, there needs to exist $\Gamma$ such that $\vv(\Gamma) \geq 2$. However, it is not hard to modify our construction for arbitrary distributions. All we need is the existence of some $\Gamma$ such that $\vv(\Gamma) > 0$.}   
\end{example}

We conclude our discussion by arguing that the only simple-to-participate generalized position auction for a finite block and an unbounded distribution is the trivial mechanism.

\begin{theorem} \label{thm:GenPosFinite}
    Let $\Distr$ be a smooth regular distribution.
    Then, the only simple-to-participate generalized position auction for a finite block when user values are drawn from $\Distr$ is the trivial auction that almost never allocates any user.
\end{theorem}

We defer the proof to \autoref{sec:ProofofGenPosFinite}.

We leave as an open question whether there exist simple-to-participate mechanisms for unbounded distributions for blocks with a finite capacity.

%% file: 70-priordependentpostedprice.tex
\section{Discussion: Data-Driven Prior-Dependent Posted-Price Mechanisms}
\label{sec:discussion}

To summarize our discussion thus far, we have argued that simple-to-participate mechanisms need to be prior-dependent if it does not have access to heavy cryptographic machinery.
In this section, we consider various challenges in algorithmically estimating the prior $\Distr$ even for the straightforward posted-price mechanism for an infinite block.

For the posted-price mechanism with a burn informed by an algorithmic inference procedure to be off-chain influence proof, the distribution estimated empirically by the procedure and the miner's true beliefs should not be significantly different.
We discuss possible reasons for mismatch in the on-chain and off-chain information available that can cause large gaps in learning the prior.

As a baseline, we go over various considerations of Ethereum's incumbent transaction fee mechanism EIP-1559 \citep{ButerinRLP19, Roughgarden20}, which posts a price, burns all the payments collected from the users and leaves the miner with only the block reward.
In steady state, EIP-1559 aims at maintaining a high rate of block space utilization (or \emph{gas}, in Ethereum's case) subject to ensuring a low probability of demand exceeding supply (a trade-off formalized by \citealp{GaneshHSvM24}).
However, in case of extreme demand spikes such as immediately after the launch of an NFT collection, EIP-1559 is content with defecting to a first-price auction to determine inclusion in the block.

The prior estimation protocol can aim for a similar performance.
In steady state, when the drift in the distribution of user values between blocks is small, we should hope that the protocol computes a distribution that is $\varepsilon$-close to the true prior with high probability.
For example, if the goal is to implement the revenue optimal posted-price mechanism, \citet{DhangwatnotaiRY15} show an estimation procedure to learn a price that obtains a $(1-\varepsilon)$-approximation to the optimal revenue for $\varepsilon \xrightarrow{} 0$ as the number of available samples tends to infinity.
Alternatively, the mechanism could also shoot for maintaining a target allocation probability or utilizing a target fraction $\gamma$ of the block space (for example $\gamma = 1/2$ for EIP-1559).
We leave bounding the sample complexity of estimation procedures to learn the price and burn for both, achieving a target allocation probability and a target block space utilization, as an open problem.
During extreme events when the distribution shifts significantly between blocks, the data in previous blocks cannot be used to infer meaningful estimates for the current block, and thus, the mechanism must rely on other guardrails similar to the first-price auction of EIP-1559.

While deterring a miner from running mechanisms off-chain, using data from previous blocks to estimate prices and corresponding burns has the following downsides when compared to EIP-1559.
First, the price-adjustment mechanism of EIP-1559 is extremely simple --- the price for the next block is purely a function of the current price and the fraction $\gamma$ of block space utilized.
However, estimating virtual values is far more nuanced and might require (at least temporarily) storing transactions that were not included in the previous blocks.
Second, for prices strictly larger than the monopoly reserve (which require burning a strictly positive fraction of the payments collected), a naive inference procedure might not be robust against a far-sighted miner trying to reduce the burn for future blocks.
For example, for a distribution with  CDF $F$ and a PDF $f$, suppose the miner expects a price $p$ and a burn $\vv(p) = p - \frac{1-F(p)}{f(p)}$ for the next block.
The miner can then pay small bribes to users with values smaller than $p$ to shade their bids so that for the inferred distribution with a virtual value $\hat{\vv}$, $\hat{\vv}(p) = 0$.\footnote{It is possible to manipulate probability mass below $p$ such that the virtual value at $p$ of the manipulated distribution equals $0$. The monopoly reserve $r$ with virtual value $0$ of a regular distribution satisfies $r = \arg \max_{\hat{r}} \hat{r} \times (1 - F(\hat{r}))$ (see \citealp{HartlineBook}, for example). Shading bids smaller than $p$ so that the maximum of $\hat{r} \times (1 - F(\hat{r}))$ is attained at $p$ can trick a naive learning protocol to mistake $\hat{\vv}(p) = 0$.}
Users with values smaller than $p$ do not get allocated by the mechanism anyways, and should be happy to take any positive bribe offered by the miner.
Thus, the miner can now earn $p$ per included user, without having to burn any portion of the revenue.

We leave as an open question designing tamper-resistant inference procedures and more generally, mitigating potential reasons for mismatch between the on-chain protocol's estimates and data available off-chain.

%% file: 91-appendix-extra-content.tex
\section{Mechanism Design Preliminaries} \label{sec:MechDesign}

To begin, we give a very brief review of mechanism design.
For an accessible introduction focused on the application to TFMs, see \cite{GaneshTW24};
For a complete treatment, see resources such as \citet{roughgarden2010algorithmic} and \citet{HartlineBook}.

We consider single-parameter quasi-linear Bayesian environments specified by some set $[n] = \{ 1,\ldots, n\}$ of users, a set of feasible allocations $\FeasAlloc\subseteq 2^{\{1,\ldots,n\}}$ of which users are allocated (i.e., included in a block), and distributions $\TypeDistr_1 \times \dots \times \TypeDistr_n$  of the users' valuations for being allocated; in this paper, we consider the i.i.d. environment with $\TypeDistr_i = \TypeDistr$ for all bidders.
An \emph{outcome} in such an environment is a tuple in $\Outcomes = \FeasAlloc \times \R^n$, where outcome $(X, \Pay_1,\ldots, \Pay_n)$ specifies a set $X$ of allocated users, and each user $i$ pays $P_i$ for a total utility of $v_i \cdot \1{ i \in X } - \Pay_i$.
Users are risk-neutral; i.e., their utility for a distribution of outcomes equals their expected utility.
A (sealed-bid) \emph{auction} (or a \emph{mechanism}) over some environment is any mapping $\mathcal{M} : \T^{n} \to \Delta(\Outcomes)$ from bids of each user to a distribution over outcomes.

We are often interested in incentive compatible mechanisms, i.e., ones where the user is incentivized to bid their true value.
For some mechanism $\mathcal{M}$ and a value profile $\vec v$, let $X_i(\vec v)$ denote the probability that $i$ is allocated and $P_i(\vec v)$ be his expected payment (where the expectation is taken only over the randomness in the mechanism); we call $X_i$ and $P_i$ the \emph{allocation rule} and \emph{payment rule} of $i$.
The mechanism $\mathcal{M}$ is \emph{DSIC} (dominant-strategy incentive compatible) if for each user $i$, bidding $v_i$ is $i$'s dominant strategy irrespective of the bids of the other users. 
Formally, for all 
$\vec v = (v_1, \dots, v_n)$ and $\widetilde{v}_i$,
\[
 v_i\cdot X_i(v_i, \vec v_{-i}) - P_i(v_i, \vec v_{-i}) 
\ge
v_i\cdot X_i(\widetilde{v}_i, \vec v_{-i}) - P_i(\widetilde{v}_i, \vec v_{-i}).
\]

\citet{Myerson81} argues that $\AllocRule$ can be the allocation rule of a DSIC mechanism if and only if $\AllocRule_i$ is monotone in user $i$'s value $v_i$.
Further, the payment rule $\Pay$ for which $(\AllocRule, \Pay)$ is DSIC can be characterized as a function of the allocation rule $\AllocRule$.

\begin{theorem}[Payment identity] \label{thm:PaymentIdentity}
    A mechanism $\mathcal{M}$ with an allocation and payment rule $\AllocRule$ and $\Pay$ respectively is DSIC if and only if:
    \begin{itemize}
        \item $\AllocRule_i(\cdot, \vec{v}_{-i})$ is monotone non-decreasing in $v_i$.
        \item The payment rule $\Pay$ satisfies, for some constant $c_i$,
        $$\Pay_i(v_i, \vec v_{-i}) = \int_0^{v_i} z \AllocRule_i'(z, \vec v_{-i}) \, dz + c_i.$$
    \end{itemize}
\end{theorem}

For deterministic mechanisms, an equivalent version of the payment identity can be stated in terms of the \emph{critical bid}---a concept we use in our paper.
Given bids $\vec{v}_{-i}$ of the other users, the critical bid $\overline{v}_i$ for user $i$ is the smallest bid for which $i$ is allocated by the mechanism.
Then, the payment identity states that the DSIC mechanism $\mathcal{M}$ charges the critical bid $\overline{v}_i$ from user $i$ whenever he is allocated.

A weaker incentive compatibility constraint called \emph{Bayesian incentive compatibility (BIC)} is also considered in literature, which says that when user $i$ takes expectation over the bids of the other users, he maximizes his utility by bidding $v_i$.
That is, if we define the \emph{interim} allocation and payment rules to be
\[ x_i(v_i) = \Es{\vec v_{-i} \sim \TypeDistr_{-i}}{X_i(v_i, \vec v_{-i})} \text{\hspace{1cm} and \hspace{1cm}} p_i(v_i) = \Es{\vec v_{-i} \sim \TypeDistr_{-i}}{P_i(v_i, \vec v_{-i})},
\]
then the mechanism is BIC if for all $v_i, \widetilde{v}_i \in \R_{\ge 0}$,
\[
v_i\cdot x_i(v_i) - p_i(v_i) \ge v_i\cdot x_i(\widetilde{v}_i) - p_i(\widetilde{v}_i).
\]

We also consider non-incentive-compatible auctions in which users best-respond to each others' Bayesian strategies, i.e., users play a \emph{BNE} (Bayes-Nash equilibrium).
Such a mechanism $\mathcal{M}$ accepts bids from user $i$ drawn from some message space $\mathsf{Bid}_{i}$, and maps each input in 
$\mathsf{Bid}_{1}\times\ldots\times\mathsf{Bid}_{n}$ to a feasible allocation and payments for each user.
User $i$'s \emph{strategy} $s_i$ is a map from $i$'s value to her bid in $\mathsf{Bid}_i$.
A profile of strategies $(s_1, \dots, s_n)$ is a BNE if for all values $v_i$ and alternative bids $b_i$, we have 
\begin{align*}
 & \Es{\vec v_{-i}\sim \mathcal{T}_{-i}}{ v_i \cdot X_i(s_i(v_i), s_{-i}(\vec v_{-i})) - P_i(s_i(v_i), s_{-i}(\vec v_{-i})) }
\\ & \qquad \qquad \ge \Es{v_{-i}\sim \mathcal{T}_{-i}}{ v_i \cdot X_i(b_i, s_{-i}(\vec v_{-i})) - P_i(b_i, s_{-i}(\vec v_{-i})) }.
\end{align*}
The interim allocation and payment rules of a BNE are defined as in a BIC mechanism, i.e., 
    \[
    x_i(v_i) = \Es{\vec v_{-i} \sim \TypeDistr_{-i}}{X_i(s_i(v_i), s_{-i}(\vec v_{-i}))} \text{\hspace{1cm} and \hspace{1cm}} p_i(v_i) = \Es{\vec v_{-i} \sim \TypeDistr_{-i}}{P_i(s_i(v_i), s_{-i}(\vec v_{-i}))}.
    \]
The classic revelation principle says that for every BNE of every mechanism, there exists a different BIC mechanism with the same interim allocation and price rules.

A payment identity similar to DSIC mechanisms can be derived for BNEs as well.
Further, Bayesian mechanism design gives rich tools for characterizing BNEs and analyzing their revenue.
Much of this work uses the concept of a virtual value. 
When user $i$'s with value distribution $\TypeDistr_i$ has cdf $F_i(t) = \Pl{v_i \le t}$ and corresponding pdf $f_i$,
his virtual value is defined as $\varphi_i(v_i) = v_i - (1-F_{\TypeDistr}(v_i))/f_{\TypeDistr}(v_i)$.

\begin{theorem}[\citealp{Myerson81}]
\label{thm:myerson}
Consider any environment and any mechanism $\mathcal{M}$ with strategies $\big(s_i(\cdot)\big)_{i\in [n]}$ that defines interim allocation and payment rules $\big(x_i(\cdot)\big)_{i\in [n]}$ and $\big(p_i(\cdot)\big)_{i\in [n]}$.
We have:
\begin{enumerate}[(1)]
\item \label{item:payment-identity}
(\textbf{Payment identity.})
    The strategies $(s_i)_{i\in[n]}$ are a BNE if and only if the following holds:
    for each $i$, the allocation rule $x_i(\cdot)$ is monotonically non-decreasing, and the payment rule $p_i(v_i)$ satisfies $p_i(v_i) = \int_0^{v_i} zx_i'(z) \, dz + c$ for some constant $c$.
\item \label{item:revenue-equals-virtual-welfare}
(\textbf{Myerson's lemma / revenue equivalence.})
    The expected revenue (i.e., the sum of the payments) satisfies
    \[ \Es{v \sim \TypeDistr^n}{\sum_{i = 1}^n p_i(v_i)} 
    = \Es{v \sim \TypeDistr^n}{\sum_{i = 1}^n \varphi_i(v_i) \cdot x_i(v_i)}.
    \]
    In particular, every mechanism (or BNE) with interim allocation rules $(x_i)_{i\in[n]}$ has the same expected revenue.
\end{enumerate}
\end{theorem}

Bayesian mechanism design often studies classes of value distributions such that revenue-maximization becomes more tractable.
    
\begin{definition}[Regular Distributions and Monopoly Reserve] \label{def:Regular}
    A \emph{regular} distribution $\TypeDistr$ is one in which the virtual value function $\varphi(\cdot)$ is monotonically non-decreasing.
    The \emph{monopoly reserve} of the regular distribution $\TypeDistr$ is the supremum of all values $\mathsf{r}$ such that $\varphi(\mathsf{r}) < 0$.
\end{definition}

For a regular distribution $\Distr$ and all $v < \sup \Distr$, $\Pr_{\hat{v} \sim \Distr}[\hat{v} = v] = 0$.
However, there exist distributions $\Distr$ for which there is a positive probability mass on the supremum $\sup \Distr$.
As a word of caution, note that if there are multiple values such that their virtual values equal zero, note that $\vv^{-1}(0)$ is not the monopoly reserve by our definition.
The monopoly reserve is the largest value whose virtual value equals zero.

Throughout the paper, we will follow the convention that $\vv^{-1}(\phi)$ is the supremum over all values $v$ such that $\vv(v) < \phi$.

\section{A Zoo of Posted-Price Mechanisms}
\label{sec:posted-zoo}

In this section, we review various designs of posted-price mechanisms for an infinite block and the properties the respective equilibria satisfy.
We consider different variations of the block-building process $\bBuild$, including whether $\bBuild$ sets the price or receives it as input from the miner, whether the mechanism is cryptographic or plain-text, and whether the payments are burnt or passed on to the miner.
Across all variants of the posted-price mechanism, we will consider the user BNE $\sigma^{\onCG}$ where users bid their values truthfully and the miner, (if applicable) sets a price equal to the monopoly reserve, and does not censor or fabricate any bids.
Unless specified otherwise, all the mechanisms discussed below are plain-text.

\paragraph{EIP-1559:} For a price $p$ set by the block-building process $\bBuild$, EIP-1559 includes all users bidding above $p$, charges $p$ from each included user and burns all revenue.
\citet{GaneshTW24} argue that $\sigma^{\onCG}$ is on-chain user simple, on-chain miner simple and \emph{strong collusion proof}, but not off-chain influence proof.

On-chain user simplicity is fairly straightforward.
A user with value greater than $p$ can only get himself unallocated by bidding below his value, and similarly, a user with value at most $p$ would not want to overbid, get allocated and pay $p$ to get included.
Similarly, on-chain miner simplicity follows since the miner can only make a non-positive revenue on-chain irrespective of its strategy, and it makes zero revenue in $\sigma^{\onCG}$.

$\sigma^{\onCG}$ is not off-chain influence proof since any off-chain mechanism $\Moff$ where the miner demands an entry fee to be paid off-chain on top of the payment $p$ made on-chain will increase the miner's expected revenue.
The miner receives no revenue in the equilibrium $\sigma^{\onCG}$.
On the other hand, suppose the miner demands an entry fee $\rho$ off-chain.
Then, users with a value greater than $(p + \rho)$ will want to pay the entry fee and get included in the block.
The miner will make an expected revenue $\rho \, \Pr_{v \sim \Distr}[v \geq p+\rho]$ from each user, which is strictly greater than zero as long as $\rho + p \in \supp(\Distr)$.

EIP-1559 satisfies a notion of collusion-resistance called strong collusion proofness defined by \citet{GaneshTW24}, based on side-contract proofness \citep{ChungS23}.
At a high level, the equilibrium $\sigma^{\onCG}$ is strong collusion proof if the miner and a subset of users cannot increase their joint utility by \emph{integrating} into a single entity --- truthfully exchange their values amongst themselves and deviate to a strategy that increases the joint utility of the colluding cartel.
Strong collusion proofness is a consequence of the following argument.
Colluding does not increase the net revenue the miner receives from users outside the cartel.
For a user $i$ in the cartel with a value $v_i$, the user and the miner jointly make a utility $(v_i-p)$ by including user $i$, which is positive exactly when $v_i$ exceeds $p$.
$\sigma^{\onCG}$ includes user $i$ exactly when $v_i \geq p$.
Thus, the cartel cannot increase their joint utility by deviating from their equilibrium strategies and $\sigma^{\onCG}$ is strong collusion proof.

\paragraph{Prior-dependent posted-price mechanism with no burns:}
Suppose that the block-building process $\bBuild$ is prior-dependent, i.e, has an estimation procedure built-in to estimate the distribution $\Distr$ of user values and posts the monopoly reserve $p$ of $\Distr$.
All payments made by the users are transferred to the miner.

By an argument identical to EIP-1559, $\sigma^{\onCG}$ is on-chain user simple for the prior-dependent posted-price mechanism.
To argue on-chain miner simplicity, note that fabricating bids does not change the allocation or payment charged to any user, and censoring a user $i$ can only decrease the payment collected from $i$.
Further, for an infinite block, posting the monopoly reserve is the revenue optimal mechanism, and thus, the miner cannot induce a different equilibrium off-chain that extracts a higher expected revenue.
As a consequence $\sigma^{\onCG}$ is off-chain influence proof.

However, $\sigma^{\onCG}$ is not strong collusion proof.
For a user $i$ with a value $v_i$ smaller than the monopoly reserve $p$, the miner can collude with user $i$ as follows --- ask $i$ to place a bid larger than $p$ and offer a rebate $(p - v_i) + \varepsilon$.
The miner increases her revenue by colluding since she receives an additional payment $v_i - \varepsilon$ from user $i$, and user $i$ has a higher utility since he is now included and pays $v_i - \varepsilon < v_i$.
The cartel has increased their joint utility by colluding.

\citet{GaneshTW24} show that conditioned on satisfying on-chain user and miner simplicity, it is impossible to simultaneously satisfy both strong collusion proofness and off-chain influence proofness.
Thus, choosing between EIP-1559 and the prior-dependent posted-price mechanism can depend on whether the blockchain's primary concern is collusion (specifically, strong collusion-proofness) or the miner defaulting to an off-chain communication channel (specifically, off-chain influence-proofness).

While the prior-dependent posted-price mechanism does not satisfy strong collusion proofness, it satisfies a weaker notion of collusion resistance, which  \citet{GaneshTW24} prove is implied by off-chain influence proofness, called \emph{trustless collusion proofness}.
Instead of the cartel integrating into a single entity and coordinating with \emph{purely} the cartel's interests in mind, consider a version of collusion where the colluding parties still prioritize their own interests.
In other words, collusion %
happens though a specific contract of profit-sharing between the cartel, and users best-respond to the contract and each other, and hence play an equilibrium.
In particular, the miner should never receive a higher expected reward in this equilibrium if the TFM is trustless collusion proof.
If $\sigma^{\onCG}$ is off-chain influence proof, the miner's expected revenue in $\sigma^{\onCG}$ is at least the expected revenue from any equilibrium induced by the cartel, and thus, $\sigma^{\onCG}$ is also trustless collusion proof.
For more details and discussion, see \citet{GaneshTW24}.

As proved in \autoref{thm:PostedPrice}, posting a price $p$ larger than the monopoly reserve and burning $\vv(p)$ per allocated user also yields an on-chain user and miner simple and off-chain influence proof mechanism for any regular distribution $\Distr$.

\paragraph{Prior-independent posted-price mechanism with no burns:}
Instead of the block-building process estimating the distribution $\Distr$ and posting a price equal to the monopoly reserve, the prior-independent mechanism lets the miner set the price on-chain via advice.

By arguments identical to the prior-dependent variant, $\sigma^{\onCG}$ is on-chain user simple, off-chain influence proof and is not strong collusion proof.

On-chain miner simplicity of $\sigma^{\onCG}$ depends on the exact implementation of the mechanism.
As in the model considered by \citet{GaneshTW24}, if users submit their bids first and the miner decides the price after seeing the bids, $\sigma^{\onCG}$ is not on-chain miner simple.
By setting a price $p$, the miner receives a revenue $\Rev(p, \vec{v}) = p \times (\#\text{users with a bid at least } p)$ and the miner will set a price $p = \argmax_p \Rev(p, \vec{v})$, which can vary with the set of submitted bids $\vec{v}$.

However, suppose the miner is first made to commit to a price $p$, and users respond by submitting their values.
Then, the price $p$ cannot depend on the users' bids.
The miner cannot increase its revenue by fabricating or censoring bids either.
Thus, $\sigma^{\onCG}$ is on-chain miner simple.

\paragraph{The cryptographic prior-independent posted price mechanism with no burns:}
Similar to the prior-independent mechanism, the miner sets the price $p$ and all payments received from the users are transferred to the miner.
However, users submit encrypted bids and reveal their bids only after the miner announces the price $p$.
Thus, the price $p$ must be independent of the contents of the bids.

Note that the cryptographic prior-independent posted-price mechanism is similar to the plain-text prior-independent version where the miner first sets the price and then users reveal their bids.
Cryptography enforces the time lag between when the miner posts the price and when she learns of the users' bids.
Thus, $\sigma^{\onCG}$ is on-chain user and miner simple, and off-chain influence proof, but not strong collusion proof.

We summarize our discussion in \autoref{tab:all-combinations-of-properties}.

\begin{table}[htb]
    \begingroup
    \renewcommand{\arraystretch}{1.5} 
    \begin{center}
    \begin{tabular}{ccccc}
    \toprule
    & \ \ \multirow{2}{*}{\makecell{Off-Chain\\}} & \multicolumn{2}{c}{On-Chain} & \ \ \multirow{2}{*}{\makecell{Strong\\}}
    \\
    \cmidrule{3-4}
    Posted-Price Variant &  \ \ \makecell{Influence \\Proof} & \makecell{User\\Simple} \ \ & \makecell{Miner\\Simple} & \ \ \makecell{Collusion \\ Proof} \\
    \midrule
     EIP-1559 & \xmark & \cmark & \cmark & \cmark \\
     Prior-dependent & \cmark & \cmark & \cmark & \xmark \\
     Prior-independent, miner responds to users' bids & \cmark & \cmark & \xmark & \xmark \\
     Prior-independent, users responds to miner's price & \cmark & \cmark & \cmark & \xmark \\
     Cryptographic prior-independent & \cmark & \cmark & \cmark & \xmark \\
    \bottomrule
    \end{tabular}
    \end{center}
    \endgroup
    
    \caption{Desiderata satisfied by variants of posted-price mechanisms %
    (and their respective equilibria).}
    
    {\footnotesize 
    \textbf{Notes:} The block-building process sets the price in both EIP-1559 and the prior-dependent posted-price mechanism, while the miner sets the price via advice in all the prior-independent variants. All payments collected from the users are transferred to the miner except in EIP-1559, where all the payments are burnt and the miner receives no revenue. All mechanisms are plain-text except for the cryptographic prior-independent mechanism.
    In the two plain-text versions of the prior-independent posted-price mechanism, the users first submit their bids and the miner sets a price as a function of the bids it sees in one and the miner sets a price first and users respond by placing bids in the other. \par}
    \label{tab:all-combinations-of-properties}
\end{table}

\section{Deterministic Simple-to-Participate Mechanisms for Bounded Number of Bidders and Virtual Values with Discontinuities}

\subsection{Bounded Number of Bidders} \label{sec:Bounded}

Consider a TFM that supports at most two bids when the distribution $\Distr$ equals the exponential distribution with expected value $1$.\footnote{The CDF and the density of the exponential distribution equal $1 - e^{-x}$ ad $e^{-x}$ respectively. The virtual value function is given by $\vv(v) = v - 1$, with a monopoly reserve $1$.}
Suppose that $\MargThreshold_1 = 4$ and $\MargThreshold_2 = 3$.
We plot the allocation rule of the virtual utility maximizing mechanism in \autoref{fig:TwoUser}.

\begin{figure}
    {
    \centering
    \includegraphics[width=0.6\linewidth]{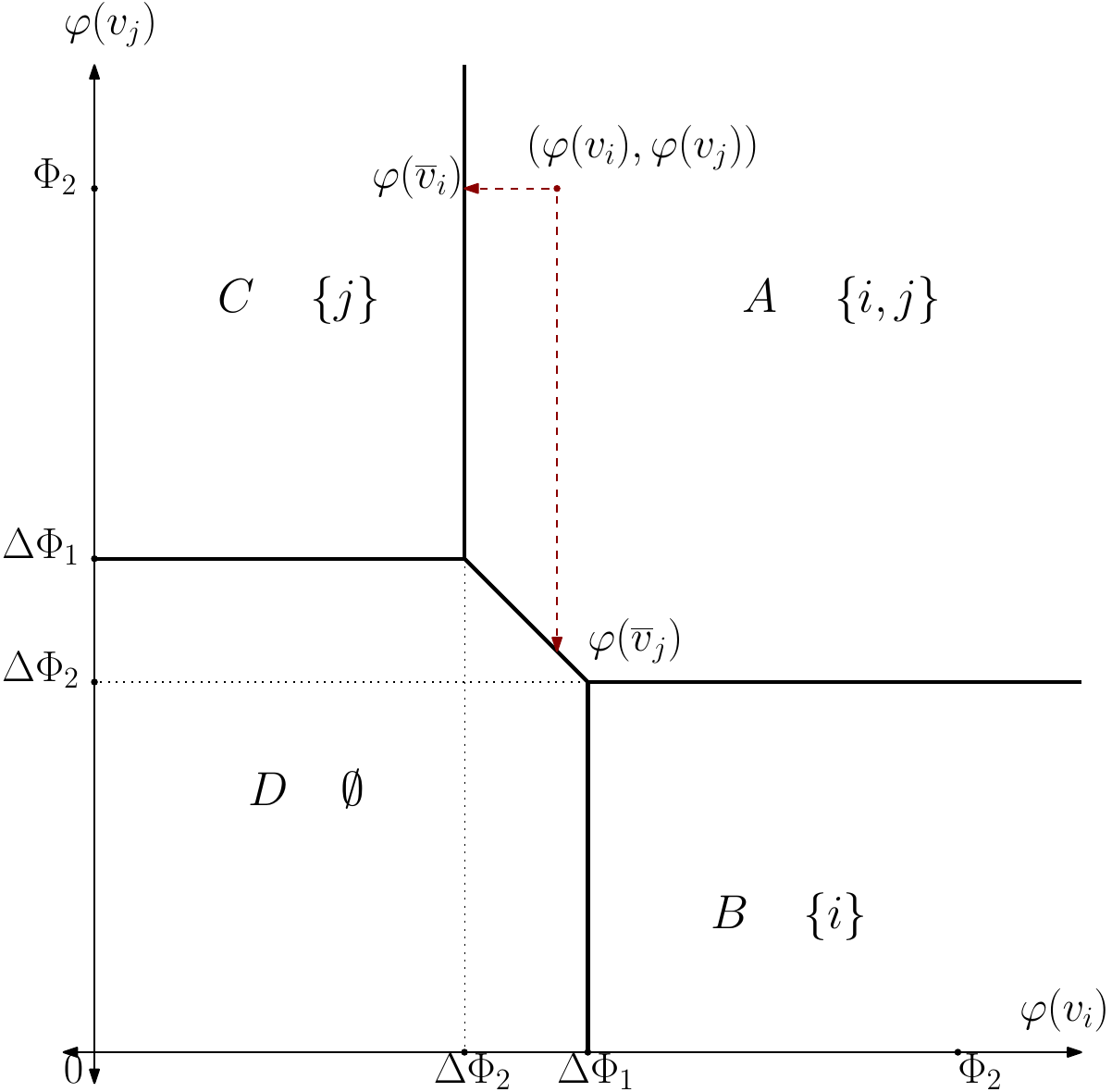}
    \caption{The allocation rule of a TFM supporting at most $2$ bids.}
    \label{fig:TwoUser}
    }
    
    {\footnotesize \textbf{Notes:} Each region contains the name of the region and the allocation when the virtual values of the bids lie in that region. In our example, we set $\MargThreshold_1 = 4$, $\MargThreshold_2 = 3$, and thus, $\Threshold_2 = 7$. For any value profile of the two users, their critical bids can be inferred from the plot. For example, for a profile $(v_i, v_j) \in A$, user $i$'s critical bid can be computed by finding the smallest value $\overline{v}_i$ for which $(\overline{v}_i, v_j)$ lies in the region $C$, where the user $i$ is not allocated (showcased by the horizontal red line).\par}
\end{figure}

By \autoref{thm:VirtualPartialConverse}, the mechanism is on-chain user simple and off-chain influence proof.
For each region on the plot, we separately argue that the miner cannot increase her revenue by censoring or fabricating bids.
\begin{enumerate}
    \item $(v_i, v_j) \in A$ and both users are allocated: To begin, note that each user must have a virtual value at least $\MargThreshold_2$ to get allocated.
    Thus, by not being strategic, the miner receives a total payment at least $2 \vv^{-1}(\MargThreshold_2) = 2 \, (\MargThreshold_2 + 1)$ (substituting the virtual value function of the exponential distribution).
    $\MargThreshold_1 + \MargThreshold_2$ is burnt to include both the users.
    The miner's net revenue is therefore larger than
    $$2 \, (\MargThreshold_2 + 1) - \MargThreshold_1 - \MargThreshold_2 = \MargThreshold_2 - \MargThreshold_1 + 2 = 1.$$
    We substitute $\MargThreshold_1 = 4$ and $\MargThreshold_2 = 3$ for the last equality.
    
    Now, suppose that the miner is strategic.
    The TFM does not support more than $2$ bids and therefore, the miner will have to censor one of them in order to fabricate one of its own.
    Suppose that the miner censors user $i$ and fabricates a fake bid $\hat{v}_i$.
    From \autoref{fig:TwoUser}, it can be observed that the miner can extract a payment of at most $\vv^{-1}(\MargThreshold_1) = \MargThreshold_1 + 1$ from user $j$, equality holding when $(\hat{v}_i, v_j) \in C$.
    The miner must include at least one user to obtain a positive revenue, and thus, at least $\MargThreshold_1$ is burnt.
    Therefore, the miner's net revenue is at most
    $$(\MargThreshold_1 + 1) - \MargThreshold_1 = 1,$$
    which, in turn, is smaller than the revenue from the compliant strategy that does not censor or fabricate bids.
    The miner cannot increase her revenue when $(v_i, v_j) \in A$.
    \item $(v_i, v_j) \in B$ and only user $i$ is allocated ($(v_i, v_j) \in C$ is analogous): User $i$ is charged a payment $\vv^{-1}(\MargThreshold_1)$ to be included.
    Note that over all $(v_i, v_j)$, the maximum payment charged to a user is exactly $\vv^{-1}(\MargThreshold_1)$.
    Thus, being strategic can only decrease the payment collected by the miner.
    If the miner fabricates a bid and ends up including two bids, the burn can only increase to further hurt the miner's revenue.
    Thus, the miner's revenue is maximized by being compliant and not fabricating any bids.
    \item $(v_i, v_j) \in D$ and no user is allocated:
    Suppose that the miner censors $i$ and fabricates $\hat{v}_i$.
    It cannot be the case that only user $j$ gets allocated, since otherwise, the mechanism would allocate $j$ even before user $i$ was censored.
    Thus, $\MargThreshold_1 + \MargThreshold_2 = 7$ is burnt.
    Further, as argued before, the maximum payment that can be extracted from user $j$ equals $\vv^{-1}(\MargThreshold_1) = 4 + 1 = 5$.
    Thus, the miner's net revenue is at most $5 - 7 < 0$.
    The miner is better off by being compliant.
\end{enumerate}

\subsection{Virtual Values with Discontinuities} \label{sec:Discontinuities}

In this section, we will construct a deterministic simple-to-participate auction for an infinite block that is not a posted-price mechanism for a distribution with a discontinuous virtual value function.

Consider the distribution $\Distr$ with a CDF $F$ given by
\begin{equation*}
    \notag
    F(v) = \begin{cases}
        1 & \text{ for } v = 2+\varepsilon, \\
        1 - \frac{\varepsilon}{2 \, (v-2+\varepsilon)} & \text{ for } v \in [2, 2+\varepsilon), \\
        \frac{v}{v+2} & \text{ for } v \in [0, 2).
    \end{cases}
\end{equation*}
The virtual value function for $\Distr$ equals
\begin{equation*}
    \notag
        \vv(v) = \begin{cases}
            2 + \varepsilon & \text{ for } v = 2+\varepsilon, \\
            2 - \varepsilon & \text{ for } v \in [2, 2+\varepsilon), \\
            - 2 & \text{ for } v \in [0, 2).
    \end{cases}
\end{equation*}

Now consider the mechanism determined by the sequence $\seq{\MargThreshold_t}{t \in \N}$ of marginal burns where
$\MargThreshold_1 = 2 + \varepsilon$ and $\MargThreshold_t = 0$ for all $t > 1$.
The mechanism allocates a single bidder with a value $2 + \varepsilon$, or every bidder with a positive virtual value if two or more bidders have a value at least $2$.
By \autoref{thm:VirtualPartialConverse}, the mechanism is on-chain user simple (we will be charging payments according to the payment identity, \autoref{item:payment-identity} in \autoref{thm:myerson}) and off-chain influence proof.

We will argue that the mechanism is also on-chain miner simple.
If there are more than two bidders with value at least $2$, the mechanism allocates all of them, charges $2$ from each of them and burns $2+\varepsilon$ in aggregate.
Fabricating bids will neither change the payments collected from the users nor reduce the burn.
However, censoring a bid will change the highest bidder's payment if there are only two users.
Suppose the highest bidder has a value $2+\varepsilon$ and there is one other user with a value at least $2$.
Censoring the other user will increase the payment collected from the highest bidder to $2+\varepsilon$.
However, the miner loses out on the $2$ that woud have been charged to the other user if he was not censored.
Thus, the miner cannot increase her revenue by fabricating or censoring bids.

Similarly, suppose that there is exactly one user with a value larger than the monopoly reserve.
If the value of the user equals $2+\varepsilon$, then the user is already included in the block, and fabricating or censoring bids does not increase the miner's revenue.
On the other hand, if the user has a virtual value $2-\varepsilon$, the miner could consider fabricating a bid with value $2$ and get both bids included.
The miner will earn a revenue $2$ from the user, but loses $2+\varepsilon$ to the burn.

Therefore, we conclude that the mechanism is on-chain miner simple and thereby, simple to participate.

\section{Reducing Global Strong Collusion Proofness to Maximal-in-Range Multi-Item Single-Buyer Mechanism} \label{sec:MIRC}

In this section, we change our focus from off-chain influence proofness to global strong collusion proofness.
We begin by reviewing prior work and highlighting our contributions.
A growing body of work has been looking at the design of collusion-resistant mechanism design 
(\citealp{Roughgarden20, ChungS23, ChungRS24, GafniY24}; \citealp{GafniY24Discrete}).
\citet{Roughgarden20} defines global strong collusion proofness (albeit, called off-chain agreement proofness), and \citet{ChungRS24} and \citet{GafniY24Discrete} design and prove impossibilities on the existence of on-chain simple and global strong collusion proof mechanisms.
We answer the following questions from previous works.
We characterize global strong collusion proofness (\autoref{thm:MIRCUtilityVersionMainReduction}), generalizing the characterization from \citet{GafniY24Discrete} for deterministic mechanisms (Theorem 3.1 in \citealp{GafniY24Discrete}).
We also design a mechanism that is on-chain simple and global strong collusion proof that yields the miner a strictly positive revenue (\autoref{ex:MIRCGenPos}).
Further, if the users' values have a bounded support, we design on-chain simple and global strong collusion proof mechanisms even for finite blocks (\autoref{ex:MIRC-rand-auction}).
\citet{GafniY24Discrete} also construct mechanisms for finite blocks that provide the miner a positive revenue.
However, they assume that the users' values are drawn from a discrete space, and the miner's revenue approaches zero as the discretization parameter approaches zero.

We adopt the model of randomness from \citet{GafniY24Discrete}, which differs from \citet{ChungRS24} as follows.
\citet{GafniY24Discrete} require the agents to decide their actions prior to learning the randomness in the mechanism.
In contrast, \citet{ChungRS24} allow users and the miner to act ex-post, after the randomness is realized.
\citet{ChungRS24} prove the impossibility of designing on-chain simple and global strong collusion proof mechanisms for a finite block in their model of randomness.

\subsection{Characterizing Global Strong Collusion Proof Equilibria}

We obtain a characterization of global strong collusion proof mechanisms through a reduction similar to \autoref{thm:UtilityVersionMainReduction}.

We start by formally defining global strong collusion proofness, which was briefly mentioned in \autoref{sec:MOdel}.
Recall that an equilibrium $\sigma^{\onCG}$ in the on-chain game of a TFM is global strong collusion proof if the \emph{global coalition} consisting of the miner and all users bidding in the TFM cannot increase their joint utility by deviating from $\sigma^{\onCG}$ for all value profiles $\vec{v}$ of the users.

To define global strong collusion proofness, we begin by specifying the strategy space $\mathcal{S}^{\onCG}_{\coal}$ of the global coalition.
Each strategy profile $\sigma^{\onCG}_{\coal} = (s_\mi^{\onCG, \coal}, s_{\usr, 1}^{\onCG, \coal},\allowbreak \dots, \allowbreak s_{\usr, n}^{\onCG, \coal}) \in \mathcal{S}^{\onCG}_{\coal}$ maps the values $\vec{v}$ realized by the users to the strategies played by each agent in the coalition.
In particular, unlike off-chain influence proofness where the users and the miner continue to remain strategic with each other, the colluding users truthfully reveal their values to each other and the miner.
Therefore, the agents' strategies can depend on the private information $\vec{v}$ of all $n$ users.
Since any payment made by the users are transferred to the miner, the global coalition's joint utility is the total welfare of the users minus the portion of the payments burnt by the TFM.

\begin{definition}[Global strong collusion proofness]
    For a distribution $\Distr$ of user values, an equilibrium $\sigma^{\onCG}$ of a TFM is global strong collusion proof if for all value profiles $\vec{v} \in \supp(\Distr^n)$ of the users and any strategy profile ${\sigma}_{\mathsf{coal}}^{\onCG} \in \mathcal{S}^{\onCG}_{\coal}$
    \[
    \sum_{i = 1}^n v_i \Alloc_i(\sigma^{\onCG}(\vec{v})) - \Burn(\sigma^{\onCG}(\vec{v}))b\geq \sum_{i = 1}^n v_i \Alloc_i(\sigma^{\onCG}_{\mathsf{coal}}(\vec{v})) - \Burn( \sigma^{\onCG}_{\mathsf{coal}}(\vec{v})).\footnote{Note the following key distinction between off-chain influence proofness and global strong collusion proofness. While the miner should be maximizing her expected revenue in the off-chain game by playing a strategy $\sigma^{\offCG}$ for the equilibrium to be off-chain influence proof, the global coalition should maximize their joint utility by playing $\sigma^{\onCG}$ for all values $\vec{v}$ for global strong collusion proofness. This difference between the definitions makes the results for global strong collusion proofness much simpler than the analogous results for off-chain influence proofness --- to prove that an equilibrium $\sigma^{\onCG}$ is not global strong collusion proof, it is sufficient if we find one valuation profile $\vec{v}$ for which the global coalition profits by deviating, as opposed to finding an entire collection of valuation profiles with a positive probability measure.} \textsuperscript{, } \footnote{The definition of global strong collusion proofness extends strong collusion proofness from \citet{GaneshTW24}, where the miner and the colluding users have fixed arbitrary valuations (chosen worst-case to break strong collusion proofness), while the valuations of other users are drawn from $\Distr$.
    Indeed, for the global coalition, all users are part of the coalition and thus, the value profile $\vec{v}$ is chosen worst-case.}
    \]
\end{definition}

For a block-building process $\bBuild = (\Alloc, \Pay, \Burn)$ and a strategy profile ${\sigma}_{\mathsf{coal}}^{\onCG} \in \mathcal{S}^{\onCG}_{\coal}$, we define the direct-revelation mechanism corresponding to ${\sigma}_{\mathsf{coal}}^{\onCG}$ to be $(\OnCAlloc, \OnCBurn) = \Big(\Alloc({\sigma}_{\mathsf{coal}}^{\onCG}), \Burn({\sigma}_{\mathsf{coal}}^{\onCG})\Big)$.
Note that we slightly abuse notation to overload the term direct-revelation mechanism.
The miner and the users are required to play an equilibrium in the off-chain game in the definition of the direct-revelation mechanism (\autoref{def:DirectRevMechanism}).
On the other hand, we allow strategy profiles ${\sigma}_{\mathsf{coal}}^{\onCG}$ that are not even valid strategies in the on-chain game since the agents' actions can depend on each others values.

Similar to the definition of Myerson-in-Range, we find it convenient to describe the allocations and burn as a function of the values realized by the users.
We say that the direct-revelation mechanism is \emph{Maximal-in-Range for Coalitions} if it is induced by a strategy profile $\sigma^{\onCG}_{\mathsf{coal}}$ that maximizes the global coalition's joint utility.

\begin{definition}[Maximal-in-Range for Coalitions (MIRC)]
    For a block-building process $\bBuild = (\Alloc, \Pay, \Burn)$ and a distribution $\Distr$ of user valuations, let $\sigma^{\onCG}_{\coal} \in \mathcal{S}^{\onCG}_{\coal}$ pointwise maximize the joint utility $\sum_{i = 1}^n v_i \Alloc_i(\sigma^{\onCG}_{\mathsf{coal}}(\vec{v})) \allowbreak - \Burn(\sigma^{\onCG}_{\mathsf{coal}}(\vec{v}))$ of the global coalition for all $\vec{v} \in \supp(\Distr^n)$.
    Then, the corresponding direct-revelation mechanism $(\OnCAlloc, \OnCBurn)$ is Maximal-in-Range for Coalitions corresponding to the on-chain game $\onCG$ and the distribution $\Distr$.
\end{definition}

\autoref{thm:MIRCGSCP} follows immediately from the definitions. An equilibrium $\sigma^{\onCG}$ is global strong collusion proof if and only if the corresponding direct revelation mechanism is MIRC.

\begin{lemma} \label{thm:MIRCGSCP}
    An on-chain equilibrium $\sigma^{\onCG}$ is global strong collusion proof for the block-building process $\bBuild = (\AllocRule, \Pay, \Burn)$ if and only if its direct-revelation mechanism $(\OnCAlloc, \OnCBurn)$ is MIRC.
\end{lemma}

As with off-chain influence proofness, we abuse notation to say that the direct-revelation mechanism satisfies global strong collusion proofness without explicitly specifying the underlying equilibrium $\sigma^{\onCG}$.

We derive an analog of \autoref{thm:MainReduction} for global strong collusion proofness.
For a given number $n$ of users with a valuation profile $\vec{v}$, the global coalition should not be able to increase its joint utility by misreporting their values as $\vec{w}$ to the direct-revelation mechanism (analogous to Condition \ref{Bul:1}).
Additionally, the global coalition should also not be able to increase their joint utility by censoring or fabricating bids (analogous to Condition \ref{Bul:3}).
Note that the analog to Condition \ref{Bul:2} holds trivially for a global strong collusion proof mechanism --- indeed, the users' values are always non-negative and thus, the mechanism does not have to ensure that users with negative values are never allocated.

\begin{theorem} \label{thm:MIRCMainReduction}
    Let $\Distr$ be a distribution of user valuations with a connected support such that $0 \in \supp(\Distr)$.
    Then, a direct revelation mechanism $(\OnCAlloc, \OnCBurn)$ is MIRC only if for any number $n$ of users:
    \begin{enumerate}[(A)]
        \item (Optimal for $n$ users.) For any value profiles $\vec{v}, \vec{w} \in \supp(\Distr)^n$,
        $$\sum_{i = 1}^n v_i \, \OnCAlloc_i(\vec{v}) - \OnCBurn(\vec{v}) \geq \sum_{i = 1}^n v_i \, \OnCAlloc_i(\vec{w}) - \OnCBurn(\vec{w}).$$
        \addtocounter{enumi}{1}
        \item (No censoring or fabricating) For all $\vec{v} = (v_1, \dots, v_n, v_{n+1}, \dots, v_{n+t}) \in \supp(\Distr)^n \times \{0\}^t$ and $\vec{w} = (v_1, \dots, v_{n})$, 
        $$\sum_{i = 1}^{n} v_i \OnCAlloc_i(\vec{v}) \allowbreak - \OnCBurn(\vec{v}) \allowbreak = \sum_{i = 1}^{n} v_i \OnCAlloc_i(\vec{w}) \allowbreak - \OnCBurn(\vec{w}).$$
    \end{enumerate}
    Further, if the mechanism is also on-chain user simple and $(\OnCAlloc, \OnCBurn) = (\AllocRule, \Burn)$, the above two conditions imply that the mechanism is global strong collusion proof.
\end{theorem}
We skip the proof since the proofs of the forward direction and the (partial) converse are almost identical to \autoref{thm:MainReduction} and \autoref{thm:MainReductionConverse} respectively.

The reduction from designing global strong collusion proof mechanisms to designing DSIC multi-item single-buyer mechanisms is much simpler than \autoref{thm:UtilityVersionMainReduction}.
Remember that $\Feasibility$ is the set of all feasible outcomes supported by the block-building process $\bBuild$.
Upon realizing values $\vec{v}$, the global coalition plays an action so as to pick an outcome $(\AllocRule, \BurnB) \in \Feasibility$ so as to maximize $\sum_{i = 1}^n v_i \AllocRule_i - \BurnB$.
In other words, the global coalition picks an outcome $(\AllocRule, \BurnB)$ identically to a buyer with value $v_i$ for item $i$ in a multi-item auction with a menu of allocations and payments given by $\Feasibility$.

\autoref{thm:MIRCUtilityVersionMainReduction} follows by combining \autoref{thm:Rochet} with \autoref{thm:MIRCMainReduction}.
We skip the proof since it is almost identical to \autoref{thm:UtilityVersionMainReduction} and \autoref{thm:VirtualPartialConverse}.

\begin{theorem}[Burn identity for MIRC mechanisms] \label{thm:MIRCUtilityVersionMainReduction}
    Let $(\OnCAlloc, \OnCBurn)$ be a MIRC mechanism for a distribution $\Distr$ with a connected support containing $0$.
    Then, there exists a family of utility functions $\seq{U^n}{n \in \N}$, $U^n: \R^n \xrightarrow{} \R$, such that for all $n \in \N$,
    \begin{enumerate}
        \item $U^n$ is convex and non-decreasing as a function of the bids, and,
        \item for all $\vec{v} \in \supp(\Distr)^n$,
        $$\OnCAlloc(\vec{v}) = \nabla U^n(\vec{v}) \text{ and } \OnCBurn(\vec{v}) = \sum_{i = 1}^n v_i \, \nabla_i U(\vec{v}) - U(\vec{v}).$$
    \end{enumerate}
    Further, if $(\OnCAlloc, \OnCBurn) = (\AllocRule, \Burn)$, the existence of the family of utility functions satisfying the above properties implies that $(\OnCAlloc, \OnCBurn)$ is MIRC.
\end{theorem}

\subsection{Constructing Global Strong Collusion Proof Mechanisms Beyond Posted-Price}

While \citet{ChungRS24} and \citet{GafniY24Discrete} consider the space of on-chain simple and global strong collusion proof randomized mechanisms, they do not provide an explicit construction for the same beyond posted-price mechanisms.
We use the characterization from \autoref{thm:MIRCUtilityVersionMainReduction} to construct such mechanisms.
Specifically, we construct position auctions with burns even for finite blocks when the value distribution $\Distr$ is bounded.
For unbounded distributions, we design generalized position auctions that satisfy all of the three properties listed above, albeit, for infinite blocks.

\subsubsection{MIRC Position Auctions with Burn for Bounded Distributions}

In this section, we design position auctions such that the resulting TFM is on-chain simple and global strong collusion proof.
Note that MIRC position auctions are slightly different from position auctions in that the global coalition will want to include every user with a value larger than the marginal burn.
In comparison, a miner looking to maximize revenue in the off-chain game will include a user only when his virtual value is larger than the marginal burn.

\begin{definition}[MIRC position auctions]
    A \emph{MIRC position auction} is given by a sequence of non-increasing allocation probabilities $\seq{x^{(i)}}{i \in \N}$ and \emph{marginal burns per unit allocation} $\seq{\MargThreshold_i}{i \in \N}$.
    Conditioned on allocating $t$ users, the users are allocated with probabilities $x^{(1)}, \dots, x^{(t)}$ in descending order of their bids and $\Threshold_t = \sum_{i = 1}^t \MargThreshold_i \, x^{(i)} + \Threshold_0$ is burnt.
    For $\vec{v} \in \R_{\geq 0}^n$, the mechanism maximizes the global coalition's utility
    $$\max_t \sum_{i = 1}^t \big( v^{(i)} - \MargThreshold_i \big) \, x^{(i)} - \Threshold_0 = \max_t \sum_{i = 1}^t v^{(i)} \, x^{(i)} - \Threshold_t.$$
\end{definition}

We construct an on-chain simple and global strong collusion proof MIRC position auction with a constant marginal burn $\MargThreshold_i = \MargThreshold$.
We skip the proof since it is almost identical to \autoref{thm:BoundedPositionAuction}.

\begin{theorem} \label{thm:BoundedMIRCPositionAuction}
    For any bounded distribution $\Distr$ and a block with a finite capacity $\Omega$, a position auction given by the allocation probabilities $\seq{x^{(t)}}{t \in \N}$ and a constant marginal burn per unit allocation $\MargThreshold$ is on-chain miner simple if and only if (and thereby simple-to-participate if and only if)
    $\sum_{t = 1}^{\infty} x^{(t)} \leq \Omega$ and
    \begin{equation} \label{eqn:MIRCSufficientOnMS}
    t \, (x^{(t)} - x^{(t+1)}) \, ( \sup \Distr - \MargThreshold) < \MargThreshold \, x^{(t+1)}
\end{equation}
for all $t \in \N$.
\end{theorem}

We can easily modify \autoref{ex:rand-auction} to satisfy the conditions specified in \autoref{thm:BoundedMIRCPositionAuction}.

\begin{example}[MIRC position auction for block capacity $1$.]
\label{ex:MIRC-rand-auction}
Let $x^{(t)} = \frac{1}{2 \, t \, (t+1)}$ and $\MargThreshold = \tfrac{2}{3} \sup \Distr + \varepsilon$, for some small $\varepsilon > 0$.
We then have 
$$t \, (x^{(t)} - x^{(t+1)}) \, ( \sup \Distr - \MargThreshold) < \frac{1}{2 \, (t+1) \, (t+2)} \times \frac{2}{3} \sup \Distr < \MargThreshold \, x^{(t+1)},$$
satisfying \autoref{eqn:MIRCSufficientOnMS}.
Indeed, $\sum_{t = 1}^{\infty} x^{(t)} = 1$ and thus, the mechanism is feasible even for a block of capacity $1$.
The MIRC position auction corresponding to the allocation rule described above is both on-chain simple and global strong collusion proof.
\end{example}

\subsubsection{MIRC Generalized Position Auctions with Burn for Unbounded Distributions}

Finally, we move to designing on-chain simple and global strong collusion proof mechanisms beyond just posting a price.
Specifically, we consider the space of all generalized position auctions.
As with MIRC position auctions, we provide a slight modification to the burn rule so that the resulting mechanism is MIRC.

\begin{definition}[MIRC generalized position auctions]
    A \emph{MIRC generalized position auction} is given by a sequence of single-agent allocation rules $\seq{x^{(t)}}{t \in \N}$, where $x^{(t)}: \R_{\geq 0} \xrightarrow{} [0, 1]$, each $x^{(t)}$ is monotone and $x^{(t)}(v) \geq x^{(t+1)}(v)$ at all values $v \in \R_{\geq 0}$.
    For a value profile $\vec{v}$, user $(i)$ receives an allocation $x^{(i)}(v^{(i)})$.
    The payment and the burn rules are chosen to satisfy the payment (\autoref{item:payment-identity} in \autoref{thm:myerson}) and MIRC burn identities (\autoref{thm:MIRCUtilityVersionMainReduction}) respectively.
\end{definition}

We use the convention that $v^{(n+1)} = 0$.
Then, the payment charged to user $(i)$ is given by
\begin{align*}
    \OnCPay_i(\vec{v}) = v^{(i)} \, x^{(i)}(v^{(i)}) - \sum_{j = i}^n \int_{v^{(j+1)}}^{v^{(j)}} x^{(j)}(z) \,d z.
\end{align*}
Similarly, the burn rule can be calculated as follows.
\begin{align*}
    \OnCBurn(\vec{v}) = \sum_{i = 1}^n v^{(i)} \, x^{(i)}(v^{(i)}) - \int_{0}^{v^{(i)}} x^{(i)}(z) \,dz.
\end{align*}
The global coalition's utility upon realizing values $\vec{v}$ equals
\begin{equation}
    \notag
    \sum_{i = 1}^n v^{(i)} \, x^{(i)}(v^{(i)}) - \Big(\sum_{i = 1}^n v_i \, x^{(i)}(v^{(i)}) - \int_{0}^{v^{(i)}}x^{(i)}(z) \,d z \Big) = \sum_{i = 1}^n \int_{0}^{v^{(i)}} x^{(i)}(z) \,dz,
\end{equation}
whose gradient wrt $v^{(i)}$ promptly equals the allocation $x^{(i)}$ to user $(i)$, satisfying the MIRC burn identity.

First, we argue that the total payments collected from the users is larger than the burn, thereby ensuring a non-negative revenue to the miner.
\begin{align*}
\sum_{i = 1}^n \OnCPay_i(\vec{v}) &= \sum_{i = 1}^n v^{(i)} \, x^{(i)}(v^{(i)}) - \sum_{j = i}^n \int_{v^{(j+1)}}^{v^{(j)}} x^{(j)}(z) \,d z \\
&= \sum_{i = 1}^n v^{(i)} \, x^{(i)}(v^{(i)}) - \sum_{i = 1}^n i \times \int_{v^{(i+1)}}^{v^{(i)}} x^{(i)}(z) \,d z \\ 
&\geq \sum_{i = 1}^n v^{(i)} \, x^{(i)}(v^{(i)}) - \sum_{i = 1}^n \sum_{j = 1}^i \int_{v^{(i+1)}}^{v^{(i)}} x^{(j)}(z) \,d z \\
&= \sum_{i = 1}^n v^{(i)} \, x^{(i)}(v^{(i)}) - \int_{0}^{v^{(i)}} x^{(i)}(z) \,dz = \OnCBurn(\vec{v}).
\end{align*}
The inequality follows since $x^{(i)}(z) \leq x^{(j)}(z)$ for all $i \geq j$ and $z \geq 0$.

\autoref{thm:UnboundedMIRCPositionAuction} provides a sufficient condition for a generalized position auction to be on-chain miner simple alongside on-chain user simplicity and global strong collusion proofness.
We skip the proof since it is very similar to that of \autoref{thm:UnboundedPositionAuction}.

\begin{theorem} \label{thm:UnboundedMIRCPositionAuction}
    Let $\Distr$ be a distribution with $\supp(\Distr) = \R_{\ge 0}$.
    Suppose that an on-chain user simple and global strong collusion proof generalized position auction with allocation rule $\seq{x^{(t)}}{t \in \N}$ satisfies
    \begin{align}
        t \, \Big(x^{(t)}(w) - x^{(t+1)}(w) \Big) \geq (t+1) \, \Big(x^{(t+1)}(w) - x^{(t+2)}(w) \Big) \label{eqn:MIRCAllocConverge}
    \end{align}
    for all $t \in \N$, and
    \begin{align}
        t \times \int_{0}^w \Big(x^{(t)}(z) - x^{(t+1)}(z) \Big) \,dz \leq w \, x^{(t+1)}(w) - \int_0^{w} x^{(t+1)}(z) \,dz \label{eqn:MIRCGenPosIneq}
    \end{align}
    for all $t \in \N$ and $w \geq 0$.
    Then, the mechanism is also on-chain miner simple.
\end{theorem}

\autoref{ex:GenPos} can easily be modified to construct an on-chain simple and global strong collusion proof mechanism.

\begin{example}[On-chain simple and global strong collusion proof MIRC generalized position auction for unbounded distributions] \label{ex:MIRCGenPos}
    Let $\Distr$ be an unbounded distribution.
    For $\Gamma = 2 + \varepsilon$ for some small $\varepsilon > 0$, consider the generalized position auction given by
    $$x^{(t)}(w) = \begin{cases}
        0 & \text{ for } w \in [0, \Gamma), \\
        1 - \frac{1}{2e^{(w - \Gamma)}} \sum_{i = 1}^t \frac{1}{i \, (i+1)} & \text{ for } w \in [\Gamma, \infty).
    \end{cases}$$
    As argued in \autoref{ex:GenPos}, the above allocation rule satisfies \autoref{eqn:MIRCAllocConverge}.
    Verifying \autoref{eqn:MIRCGenPosIneq} is identical to \autoref{eqn:GenPosIneq} in \autoref{ex:GenPos} too.
    The inequality holds trivially for $w < \Gamma$.
    For $w \geq \Gamma$, the left hand side can be upper bounded by
    $$\int_{\Gamma}^{\infty} t \times \Big(x^{(t)}(w) - x^{(t+1)}(w)\Big) \, dz = \frac{t}{2(t+1)(t+2)} \int_{\Gamma}^{\infty} \frac{1}{e^{(w - \Gamma)}} \, dz = \frac{t}{2(t+1)(t+2)}.$$
    We can lower bound the right hand side as follows.
    $$w \, x^{(t+1)}(w) - \int_0^{w} x^{(t+1)}(z) \,dz = \int_0^{w} x^{(t+1)}(w) - x^{(t+1)}(z) \,dz \geq \Gamma \times \frac{1}{2}.$$
    The final inequality follows since $x^{(t)}(w) \geq 1/2$ for $w \geq \Gamma$.
    Since we choose $\Gamma > 2$, the right hand side is at least $1$, which is clearly larger than the left hand side.   
\end{example}

%% file: 92-appendix-omitted-proofs.tex
\section{Omitted Proofs}

\subsection{Omitted Proofs from \autoref{sec:MonopsonistLens}}

\subsubsection{Proof of \autoref{thm:MainReduction}} \label{sec:ProofMainReduction}

    Condition \ref{Bul:1} is a direct consequence of \autoref{thm:MyerImpliesVirtualUtilityOpt} --- if the direct-revelation mechanism  $(\OnCAlloc, \OnCPay, \OnCBurn)$ is Myerson-in-Range, it pointwise optimizes for its virtual utility except for a measure zero set of value profiles.

    The proof of condition \ref{Bul:2} is pretty similar to the intuition provided before \autoref{thm:MainReduction}.
    For each value profile $\vec{v} = (v_1, \dots, v_n)$ for which there exists $\vv(v_i) < 0$ such that the user is allocated with a positive probability, the miner can censor $v_i$ and replace it with an identical fabricated bid of its own.
    The allocation rule for the remaining users remains unchanged and the burn remains a constant too.
    The miner's virtual utility increases by $- \vv(v_i) \OnCAlloc_i(\vec{v}) > 0$ (the gain from not allocating user $i$).
    Thus, the set of such value profiles $\vec{v}$ for which a user with a strictly negative virtual value is allocated cannot have a positive probability measure.
    
    For condition \ref{Bul:3}, we will consider two value profiles $\vec{v} = (v_1, \dots, v_{n}, v_{n+1}, \dots, v_{n+t})$ and $\vec{w} = (v_1, \dots, v_n)$ such that $\vv(v_{n+1}), \dots, \vv(v_{n+t}) \leq 0$.
    Essentially, the miner should almost surely neither be incentivized to fabricate $v_{n+1}, \dots, v_{n+t}$ upon seeing the bids $\vec{w}$ nor want to censor $v_{n+1}, \dots, v_{n+t}$ when it sees $\vec{v}$.
    The former yields
    $$\sum_{i = 1}^{n} \vv(v_i) \OnCAlloc_i(\vec{w}) - \OnCBurn(\vec{w}) \geq \sum_{i = 1}^{n} \vv(v_i) \OnCAlloc_i(\vec{v}) - \OnCBurn(\vec{v})$$
    and the latter,
    $$\sum_{i = 1}^{n+t} \vv(v_i) \OnCAlloc_i(\vec{v}) - \OnCBurn(\vec{v}) \geq \sum_{i = 1}^{n} \vv(v_i) \OnCAlloc_i(\vec{w}) - \OnCBurn(\vec{w}).$$
    Applying condition \ref{Bul:2}, except with zero probability, we get $\vv(v_{n+i}) \, \OnCAlloc_n(\vec{v}) = 0$ for $1 \leq i \leq t$ since $\vv(v_{n+i}) \leq 0$.
    Combining the above two inequalities, we have
    $$\sum_{i = 1}^{n} \vv(v_i) \OnCAlloc_i(\vec{v}) - \OnCBurn(\vec{v}) = \sum_{i = 1}^{n} \vv(v_i) \OnCAlloc_i(\vec{w}) - \OnCBurn(\vec{w}).$$
    for $\vec{v} \sim \Distr^{n} \times \Distr_{\vv \geq 0}^t$.

\subsubsection{Proof of \autoref{thm:SmootheningVirtualUtilityMaximization}} \label{sec:ProofofSmoothening}

    At a high level, for all value profiles $\vec{v}$ for which the miner can increase her virtual utility by deviating from the equilibrium strategy, we will construct $(\widetilde{\AllocRule}, \widetilde{\BurnB})$ by taking the limit point of the direct-revelation mechanism $(\OnCAlloc, \OnCBurn)$ as the value profile tends to $\vec v$.
    The meat of the proof argues that such a limit point exists and further, incentive compatibility at other valuation profiles is not destroyed by modifying $(\OnCAlloc, \OnCBurn)$ at $\vec v$.

    For a fixed number of users $n$, let $\ICSet_n$ be the set of valuations $\vec v$ for which
    $$\sum_{i = 1}^n \vv(v_i) \, \OnCAlloc_i(\vec{v}) - \OnCBurn(\vec{v}) \geq \sum_{i = 1}^n \vv(v_i) \, \OnCAlloc_i(\vec{w}) - \OnCBurn(\vec{w}).$$
    for all $\vec{w} \in \supp(\Distr^n)$.
    
    To begin, for all points $\vec v \not \in \ICSet_n$, we will construct a sequence of valuation profiles $\vec v = \seq{\vec \upsilon_i}{i \in \N} \subset \ICSet_n$ that converges to $\vec v$ (we abuse notation to identify a sequence with its limit point).
    We can then define $(\widetilde{\AllocRule}, \widetilde{\BurnB})$ to be the limit points of $\seq{\OnCAlloc(\vec \upsilon_i)}{i \in \N}$ and $\seq{\OnCBurn(\vec \upsilon_i)}{i \in \N}$ respectively.
    Since the mechanism is Myerson-in-Range, the complement of $\ICSet_n$ in the support of $\Distr$ is measure zero.
    As the distribution $\Distr$ is regular, the support $\supp(\Distr)$ is a closed set. Thus, the closure of the set $\ICSet_n$ after deleting measure zero number of points from $\supp(\Distr^n)$ equals $\supp(\Distr^n)$.
    Therefore, for all points $\vec v \not \in \ICSet_n$, there exists a sequence $\vec v$ fully contained in $\ICSet_n$ that converges to $\vec v$.
    For convenience, let $\vec{v}$ be the constant sequence $\seq{\vec v}{i \in \N}$ for all $\vec v \in \ICSet$.

    Next, we argue that there exists some subsequence of $\vec v$ (indexed by $\mathcal{I}$) such that $\seq{\OnCAlloc(\vec \upsilon_i), \OnCBurn(\vec \upsilon_i)}{i \in \mathcal{I}}$ converges.
    If we prove that $\seq{\OnCAlloc(\vec \upsilon_i), \OnCBurn(\vec \upsilon_i)}{i \in \N}$ is bounded, then by the Bolzano–Weierstrass theorem (see, for example, \citealp{Rudin64}), there exists some convergent subsequence as required.
    Indeed, the allocation rule $\OnCAlloc$ belongs to the interval $[0, 1]$.
    Moreover, the burn rule cannot vary significantly in comparison to the norm $||\vec \upsilon_i||$ of the valuation profile, i.e, $\OnCBurn(\emptyset) \leq \OnCBurn(\vec \upsilon_i) \leq ||\vec \upsilon_i|| + \OnCBurn(\emptyset)$.
    Otherwise, if $\OnCBurn(\emptyset) > \OnCBurn(\vec \upsilon_i)$, the miner can increase her revenue when there are no users (which is a probability $1$ event when the number $n$ of users equals $0$) by fabricating the bids $\vec \upsilon_i$.
    On the other hand, if $ \OnCBurn(\vec \upsilon_i) > \OnCBurn(\emptyset) + ||\vec \upsilon_i||$, the miner generates a larger virtual utility by leaving out all $n$ users.
    Since the virtual value is at most the value for all points in the support of $\Distr$,
    $$\sum_{j = 1}^n \vv(\upsilon_{i, j}) \, \OnCAlloc(\vec \upsilon_i) - \OnCBurn(\vec \upsilon_i) \leq ||\upsilon_i|| - \OnCBurn(\vec{\upsilon_i}) < - \OnCBurn(\emptyset).$$    
    Hence, $ - \OnCBurn(\emptyset) \geq ||\vec \upsilon_i|| - \OnCBurn(\vec \upsilon_i) \geq  \vv(\vec \upsilon_i) \, \OnCAlloc(\vec \upsilon_i) - \OnCBurn(\vec \upsilon_i).$
    Since the sequence $\vec v$ converges, it is also bounded and thus, $\OnCBurn(\vec \upsilon_i) \in [\OnCBurn(\emptyset), \OnCBurn(\emptyset) + ||\vec \upsilon_i||]$ is also bounded. 

    Re-index the sequence $\vec{v}$ so that $\seq{\OnCAlloc(\vec \upsilon_i), \OnCBurn(\vec \upsilon_i)}{i \in \N}$ converges.
    We are now ready to define $\widetilde{\AllocRule}$ and $\widetilde{\BurnB}$.
    For the sequence $\vec v = \seq{\vec \upsilon_i}{i \in \N}$,
    $$\widetilde{\AllocRule}(\vec v) = \lim_{i \xrightarrow{} \infty} \OnCAlloc(\vec \upsilon_i) \text{ and } \widetilde{\BurnB}(\vec v) = \lim_{i \xrightarrow{} \infty} \OnCBurn(\vec \upsilon_i).$$

    Ir remains to verify that $(\widetilde{\AllocRule}, \widetilde{\BurnB})$ satisfies $\sum_{i = 1}^n \vv(v_i) \, \widetilde{\AllocRule}_i(\vec v) - \widetilde{\BurnB}(\vec v) \geq \sum_{i = 1}^n \vv(v_i) \, \widetilde{\AllocRule}_i(\vec{w}) - \widetilde{\BurnB}(\vec{w})$ for all $\vec v, \vec{w} \in \supp(\Distr^n)$.
    Let $\vec v = \seq{\vec \upsilon_i}{i \in \N}$ and let $\vec{w} = \seq{\vec{\omega}_i}{i \in \N}$.
    For all $i \in \N$, since $\vec \upsilon_i \in \ICSet_n$, we have
    \begin{equation} \label{eqn:CompareSeqi}
        \sum_{j = 1}^n \vv(\upsilon_{i, j}) \, \widetilde{\AllocRule}_j(\vec \upsilon_i) - \widetilde{\BurnB}(\vec \upsilon_i) \geq \sum_{j = 1}^n \vv(\upsilon_{i, j}) \, \widetilde{\AllocRule}_i(\vec{\omega}_i) - \widetilde{\BurnB}(\vec{\omega}_i)
    \end{equation}
    By construction of the sequence $\vec v$, $||[\sum_{j = 1}^n \vv(v_j) \, \widetilde{\AllocRule}_j (\vec v) - \widetilde{\BurnB} (\vec v)] - [\sum_{j = 1}^n \vv(\upsilon_{i, j}) \, \OnCAlloc_j(\vec \upsilon_i) - \OnCBurn(\vec \upsilon_i)]||$ tends to zero as $i \xrightarrow{} \infty$.
    Indeed, $\lim_{i \xrightarrow{} \infty } \vv(\vec \upsilon_i) = \vv(\vec{v})$ (by continuity of the virtual value function), $\lim_{i \xrightarrow{} \infty }\OnCAlloc(\vec \upsilon_i) = \widetilde{\AllocRule} (\vec v)$ and $\lim_{i \xrightarrow{} \infty }\OnCBurn(\vec \upsilon_i) = \widetilde{\BurnB} (\vec v)$ (by definition). 
    Both $\vec v$ and $\seq{\upsilon_i}{i \in \N}$ have a bounded norm, and thus, for a sufficiently large $i_{\varepsilon}$ and all $i \geq i_{\varepsilon}$
    \begin{equation} \label{eqn:VecSeqiTov}
        ||[\sum_{j = 1}^n \vv(v_j) \, \widetilde{\AllocRule}_j (\vec v) - \widetilde{\BurnB} (\vec v)] - [\sum_{j = 1}^n \vv(\upsilon_{i, j}) \, \OnCAlloc_j(\vec \upsilon_i) - \OnCBurn(\vec \upsilon_i)]|| \leq \varepsilon
    \end{equation}
    Similarly, for all $i \geq i'_{\varepsilon}$ for a large enough $i'_{\varepsilon}$,
    \begin{equation} \label{eqn:HatSeqiTov}
        ||[\sum_{j = 1}^n \vv(v_j) \, \widetilde{\AllocRule}_j (\vec{w}) - \widetilde{\BurnB} (\vec{w})] - [\sum_{j = 1}^n \vv(\upsilon_{i,j}) \, \OnCAlloc_j(\vec{\omega}_i) - \OnCBurn(\vec{\omega}_i)]|| \leq \varepsilon
    \end{equation}
    Combining \autoref{eqn:CompareSeqi}, \autoref{eqn:VecSeqiTov} and \autoref{eqn:HatSeqiTov}, we have 
    $$\sum_{j = 1}^n \vv(v_j) \, \widetilde{\AllocRule}_j (\vec v) - \widetilde{\BurnB} (\vec v) \geq  \sum_{j = 1}^n \vv(v_j) \, \widetilde{\AllocRule}_j (\vec{w}) - \widetilde{\BurnB} (\vec{w}) - 2 \varepsilon$$
    for all $\varepsilon > 0$.
    The result follows.    

\subsection{Omitted Proofs from \autoref{sec:PriorIndependent}}

\subsubsection{Proof of \autoref{thm:NoAdvice}} \label{sec:ProofofNoAdvice}

Suppose there exists a value profile $\vec{v}$ for which some user is allocated with a positive probability.
By monotonicity of the value allocation rule $\OnCAlloc$, some user must receive a positive allocation for all value profiles $\vec{w}$ coordinate-wise larger than $\vec{v}$.
Let $W$ be the set of all such value profiles.

Choose a smooth regular distribution $\Distr$ with a sufficiently large monopoly reserve such that $\{\vec{w} \in W | \vv(w_i) < 0 \text{ for all } 1 \leq i \leq n\}$ has a positive measure.
By on-chain user simplicity, users bid their values $\vec{w}$ truthfully irrespective of the distribution $\Distr$ and thus, are allocated with a positive probability.
Then, condition \ref{Bul:2} in \autoref{thm:MainReduction} is violated with a positive probability for $\Distr$ and thus, the direct-revelation mechanism $\DirReveal$ cannot be Myerson-in-Range and consequently, off-chain influence proof.

\subsubsection{Proof of \autoref{thm:NoCrypto}} \label{sec:ProofNoCrypto}

Suppose that for some distribution $\Distr$, there exists a collection $V$ of values profiles with a positive probability measure such that
the direct-revelation mechanism $(\OnCAlloc^{\Distr}, \OnCPay^{\Distr}, \OnCBurn^{\Distr})$ leaves the miner with a revenue strictly larger than the block reward $-\OnCBurn^{\Distr}(\emptyset)$ from building an empty block for all $\vec{v} \in V$.

Now consider another distribution $\hat{\Distr}$ with a large enough monopoly reserve so that the set $W = \{\vec{w} \in V | \vv^{\hat{\Distr}}(w_i) < 0 \text{ for all } 1 \leq i \leq n\}$ has a positive measure.\footnote{The set $W$ can have measure zero if the distribution $\Distr$ has a point mass. Suppose that, for the distribution $\Distr$ with a positive probability mass on its supremum $\sup \Distr$, the miner receives a revenue larger than the block reward exactly when the users' values equal $\vec{v} = (\sup \Distr, \dots, \sup \Distr)$. A regular distribution $\hat{\Distr}$ can only have a point mass on its supremum, and thus, the set $W$ will have a measure zero. However, assuming that $\Distr$ has no point masses solves the above issue. For any set $\hat{W}$ with a positive probability measure under the distribution $\Distr$, we can easily find another distribution $\hat{\Distr}$ with a sufficiently large monopoly reserve so that $\hat{W}$ also has a positive probability measure under $\hat{\Distr}$.}
By condition \ref{Bul:2} in \autoref{thm:MainReduction}, no user with a value profile $\vec{w} \in W$ can be allocated by the Myerson-in-Range mechanism $\DirReveal$, except with probability zero.
Thus, the miner receives only the block reward $-\OnCBurn^{\hat{\Distr}}(\emptyset)$ corresponding to the distribution $\hat{\Distr}$.

Supposing that we show $-\OnCBurn^{\hat{\Distr}}(\emptyset) = -\OnCBurn^{\Distr}(\emptyset)$, then the miner can increase her revenue for the value profiles $\vec{w} \in \hat{W}$ when the users' value are drawn from $\hat{\Distr}$.
The miner can trick the mechanism by submitting the advice $\adv^{\Distr}$ corresponding to $\Distr$.
By on-chain user simplicity, users bid their values truthfully irrespective of the value distribution.
Thus, just using the bids submitted by users, the mechanism cannot distinguish between values really being drawn from $\Distr$ and the miner tricking it by submitting $\adv^{\Distr}$.
Thus, for $\vec{w} \in W$, the miner is paid a revenue as if the distribution was $\Distr$, which is strictly larger than the revenue when the distribution equals $\hat{\Distr}$, contradicting on-chain user simplicity.

It only remains to argue $-\OnCBurn^{\hat{\Distr}}(\emptyset) = -\OnCBurn^{\Distr}(\emptyset)$, which follows straightaway through on-chain miner simplicity when the number of users $n = 0$.
If possible, let $-\OnCBurn^{\Distr}(\emptyset) > -\OnCBurn^{\hat \Distr}(\emptyset)$.
Then, when values are drawn from $\hat{\Distr}$, the miner can strictly increase her revenue by submitting $\adv^{\Distr}$, getting a revenue $-\OnCBurn^{\Distr}(\emptyset) > -\OnCBurn^{\hat \Distr}(\emptyset)$, thereby contradicting on-chain miner simplicity.

\subsection{Omitted Proofs from \autoref{sec:ImpossibilitesandPossibilities}}

\subsubsection{Proof of \autoref{thm:MainReductionConverse}} \label{sec:ProofofMainReductionConverse}

    For $\vec{v} \in \supp(\Distr^n)$, we will argue that the miner cannot increase her virtual utility by (a) coercing some users to bid $w_i$ instead of $v_i$, (b) censoring some users and finally, (c) fabricating some bids $\hat{w}$.

    For a value profile $\vec{v} = (v_1, \dots, v_n)$, suppose that the miner coerces users $1, \dots, t$ to bid $\vec{w}_{\mathsf{coerce}} = (w_1, \dots, w_t)$, censors users $t+1, \dots, n$ and fabricates bids $\vec{w}_{\mathsf{fake}} = (\hat w_{n+1}, \dots, \hat w_{\hat{n}})$.
    Now consider the value profile $(\vec{v}, \vec{v}_{\mathsf{fake}}) = (v_1, \dots, v_n, \hat{v}_{n+1}, \dots, \hat{v}_{\hat{n}})$ where $\vv(\hat{v}_{n+1}) = \dots = \vv(\hat{v}_{\hat{n}}) = 0$ (remember that there exists such a value $\hat{v}_i$ from the smoothness of $\Distr$).

    We then have
    \begin{align*}
        \notag
        \sum_{i = 1}^n \vv(v_i) \, \OnCAlloc_i(\vec{v}) - \OnCBurn(\vec{v}) &= \sum_{i = 1}^n \vv(v_i) \, \OnCAlloc_i(\vec{v}, \vec{v}_{\mathsf{fake}}) + \sum_{i = n+1}^{\hat n} \vv(\hat v_i) \, \OnCAlloc_i(\vec{v}, \vec{v}_{\mathsf{fake}}) - \OnCBurn(\vec{v}, \vec{v}_{\mathsf{fake}}) \tag*{(Condition \ref{Bul:3})}
    \end{align*}

    Next, consider the profile $(\vec{v}_{\mathsf{coerce}}, \vec{v}_{\mathsf{censor}}, \vec{v}_{\mathsf{fake}}) = (v_1, \dots, v_t, \hat{v}_{t+1}, \dots, \hat{v}_{\hat{n}})$ obtained by replacing $v_{t+1},\dots, v_n$ by $\hat{v}_{t+1}, \dots, \hat{v}_n$ defined as follows --- $\hat{v}_i = v_i$ if $\vv(v_i) < 0$; otherwise, $\hat{v}_i$ equals the monopoly reserve.
    Let $(\vec{v}_{\mathsf{coerce}}, \vec{v}_{\mathsf{fake}})$ be obtained by censoring $\vec{v}_{\mathsf{censor}}$ from $(\vec{v}_{\mathsf{coerce}}, \vec{v}_{\mathsf{censor}}, \vec{v}_{\mathsf{fake}})$.
    Then,

\begin{align*}
    \notag
    &\sum_{i = 1}^n \vv(v_i) \, \OnCAlloc_i(\vec{v}, \vec{v}_{\mathsf{fake}}) + \sum_{i = n+1}^{\hat n} \vv(\hat v_i) \, \OnCAlloc_i(\vec{v}, \vec{v}_{\mathsf{fake}}) - \OnCBurn(\vec{v}, \vec{v}_{\mathsf{fake}}) \\
    &\qquad \geq \sum_{i = 1}^n \vv(v_i) \, \OnCAlloc_i(\vec{v}_{\mathsf{coerce}}, \vec{v}_{\mathsf{censor}}, \vec{v}_{\mathsf{fake}}) + \sum_{i = n+1}^{\hat n} \vv(\hat v_i) \, \OnCAlloc_i(\vec{v}_{\mathsf{coerce}}, \vec{v}_{\mathsf{censor}}, \vec{v}_{\mathsf{fake}}) - \OnCBurn(\vec{v}_{\mathsf{coerce}}, \vec{v}_{\mathsf{censor}}, \vec{v}_{\mathsf{fake}}) \tag*{(Condition \ref{Bul:1})} \\
    &\qquad \geq \sum_{i = 1}^t \vv(v_i) \, \OnCAlloc_i(\vec{v}_{\mathsf{coerce}}, \vec{v}_{\mathsf{censor}}, \vec{v}_{\mathsf{fake}}) + \sum_{i = t+1}^n \vv(\hat v_i) \, \OnCAlloc_i(\vec{v}_{\mathsf{coerce}}, \vec{v}_{\mathsf{censor}}, \vec{v}_{\mathsf{fake}}) \\
    &\qquad \qquad + \sum_{i = n+1}^{\hat n} \vv(\hat v_i) \, \OnCAlloc_i(\vec{v}_{\mathsf{coerce}}, \vec{v}_{\mathsf{censor}}, \vec{v}_{\mathsf{fake}}) - \OnCBurn(\vec{v}_{\mathsf{coerce}}, \vec{v}_{\mathsf{censor}}, \vec{v}_{\mathsf{fake}}) \\
    &\qquad =  \sum_{i = 1}^t \vv(v_i) \, \OnCAlloc_i(\vec{v}_{\mathsf{coerce}}, \vec{v}_{\mathsf{fake}}) + \sum_{i = n+1}^{\hat n} \vv(\hat v_i) \, \OnCAlloc_i(\vec{v}_{\mathsf{coerce}}, \vec{v}_{\mathsf{fake}}) - \OnCBurn(\vec{v}_{\mathsf{coerce}}, \vec{v}_{\mathsf{fake}}) \tag*{(Condition \ref{Bul:3})}
\end{align*}
The inequality in the third line follows since we replace all $\vv(v_i) > 0$ by $\vv(\hat{v}_i) = 0$ for $t+1 \leq i \leq n$.
Finally, applying condition \ref{Bul:1} for the value profiles $(\vec{v}_{\mathsf{coerce}}, \vec{v}_{\mathsf{fake}})$ and $(\vec{w}_{\mathsf{coerce}}, \vec{w}_{\mathsf{fake}})$, we get
\begin{align*}
    \notag
    &\sum_{i = 1}^t \vv(v_i) \, \OnCAlloc_i(\vec{v}_{\mathsf{coerce}}, \vec{v}_{\mathsf{fake}}) + \sum_{i = n+1}^{\hat n} \vv(\hat v_i) \, \OnCAlloc_i(\vec{v}_{\mathsf{coerce}}, \vec{v}_{\mathsf{fake}}) - \OnCBurn(\vec{v}_{\mathsf{coerce}}, \vec{v}_{\mathsf{fake}}) \\
    &\qquad \geq \sum_{i = 1}^t \vv(v_i) \, \OnCAlloc_i(\vec{w}_{\mathsf{coerce}}, \vec{w}_{\mathsf{fake}}) + \sum_{i = n+1}^{\hat n} \vv(\hat v_i) \, \OnCAlloc_i(\vec{w}_{\mathsf{coerce}}, \vec{w}_{\mathsf{fake}}) - \OnCBurn(\vec{w}_{\mathsf{coerce}}, \vec{w}_{\mathsf{fake}}) \\
    &\qquad = \sum_{i = 1}^t \vv(v_i) \, \OnCAlloc_i(\vec{w}_{\mathsf{coerce}}, \vec{w}_{\mathsf{fake}}) - \OnCBurn(\vec{w}_{\mathsf{coerce}}, \vec{w}_{\mathsf{fake}})
\end{align*}
The final equality follows since $\vv(\hat{v}_i) = 0$.

In summary, for any value profile $\vec{v}$ and any deviation $(\vec{w}_{\mathsf{coerce}}, \vec{w}_{\mathsf{fake}})$ by the miner,
$$\sum_{i = 1}^n \vv(v_i) \, \OnCAlloc_i(\vec{v}) - \OnCBurn(\vec{v}) \geq \sum_{i = 1}^t \vv(v_i) \, \OnCAlloc_i(\vec{w}_{\mathsf{coerce}}, \vec{w}_{\mathsf{fake}}) - \OnCBurn(\vec{w}_{\mathsf{coerce}}, \vec{w}_{\mathsf{fake}})$$
establishing that $\DirReveal$ is Myerson-in-Range.

\subsubsection{Proof of \autoref{thm:CharacterizeDetMechanism}} \label{sec:ProofofCharDetMech}

    We will compute the smoothened virtual utility function of a deterministic off-chain influence proof mechanism and then deduce the allocation and burn rules via the burn identity (\autoref{thm:UtilityVersionMainReduction}).

    Consider a value profile $\vec{v} \sim \Distr^n$ for which the $t$ largest bids are allocated.
    By the burn identity, we have $\OnCAlloc^{(i)} = \nabla_{(i)}^{\vv} U^n (\vv(\vec{v})) = 1$ for $1 \leq i \leq t$ and $\OnCAlloc^{(i)} = \nabla_{(i)}^{\vv} U^n (\vv(\vec{v})) = 0$ for $t+1 \leq i \leq n$ with probability $1$.
    Re-constructing the smoothened virtual utility $U$ from its gradient, we see that $U(\vv(\vec{v})) = \sum_{i = 1}^t \vv(v^{(i)}) - \Threshold_t(n)$ for some $\Threshold_t(n)$, possibly depending on the number $n$ of users, but independent of $\vec{v}$.
    Once again from the burn identity, the burn to include $t$ users equals
    $$\OnCBurn(\vec{v}) = \sum_{i = 1}^n \vv(v^{(i)}) \, \nabla^{\vv}_{(i)} U(\vv(\vec{v})) - U(\vv(\vec{v})) = \Threshold_t(n).$$
    Thus, the allocation rule $\OnCAlloc$ includes users $(1), \dots, (t)$ such that
    $$t = \argmax_{\hat{t}} \sum_{i = 1}^{\hat{t}} \vv(v^{(i)}) - \Threshold_{\hat t}(n).$$
    It only remains to argue that there exists $\Threshold_t$ such that $\Threshold_t(n) = \Threshold_t$ for all $t, n \in \N$.

    Suppose that for $t \in \N$, there exists $n \in \N$ such that for $\vec{v} \sim \Distr^n$, the probability of exactly $t$ users getting allocated is strictly positive (if no such $n$ exists, we can without loss of generality set $\Threshold_t(n) = \infty$ for all $n$ without changing the allocation rule).
    We first show that $\Threshold_t(n) = \Threshold_t$ for each such number $n$ of users.
    We will then argue that $\Threshold_t(\hat n) = \Threshold_t$ for all $\hat{n} = \N$.

    Consider $n \in \N$ such that with a positive probability, exactly $t$ users are allocated when $\vec{v} \sim \Distr^n$.
    We will show that $\Threshold_t(n) \leq \Threshold_t(\hat{n})$ for any $\hat{n}$ for which there exists some $\hat{v} = (\hat{v}^{(1)}, \dots, \hat{v}^{(\hat{n})}) \in \supp(\Distr^{\hat{n}})$ such that exactly $t$ users are allocated.
    Otherwise, if $\Threshold_t(n) > \Threshold_t(\hat n)$, the miner can increase her virtual utility by (a) coercing users $(1), \dots, (t)$ into reporting $\hat{v}^{(1)}, \dots, \hat{v}^{(t)}$ and (b) herself fabricating bids $\hat{v}^{(t+1)}, \dots, \hat{v}^{(\hat{n})}$.
    Exactly $t$ users are allocated, but only $\Threshold_t(\hat{n})$ is burnt instead of $\Threshold_t(n)$.
    As an immediate consequence, we must therefore have $\Threshold_t(n) = \Threshold_t(\hat{n}) = \Threshold_t$ for all $n, \hat{n} \in \N$ for which the probability of exactly $t$ users getting allocated is non-zero.

    Now consider $\hat n \in \N$ such that the probability of including exactly $t$ users when $\hat{v} \sim \Distr^{\hat n}$ equals zero.
    We will show that setting $\Threshold_t(\hat n) = \Threshold_t$ will not violate off-chain influence proofness.
    Assume for contradiction that the mechanism was originally off-chain influence proof, but updating $\Threshold_t(\hat n)$ to $\Threshold_t$ violates off-chain influence proofness.
    The set of all value profiles $\vec{w} \in \supp(\Distr^{\hat{n}})$ for which exactly $t$ users are allocated has a measure zero and thus, the miner potentially deviating from its equilibrium strategy upon seeing $\vec{w}$ will not violate off-chain influence proofness.
    It should rather be the case that the miner wants to induce the outcome $(\OnCAlloc(\vec{w}), \OnCBurn(\vec{w}))$ upon seeing a different value profile $\vec{v}$.
    We assumed the existence of $n \in \N$ such that when $\overline{v} \sim \Distr^n$, the probability of exactly $t$ users getting allocated is positive for which $\Threshold_t(n) = \Threshold_t$.
    Then, the miner could have included exactly $t$ users and burnt $\Threshold_t$ by having users $(1), \dots, (t)$ bid $\overline{v}^{(1)}, \dots, \overline{v}^{(t)}$ and by fabricating $\overline{v}^{(t+1)}, \dots, \overline{v}^{(n)}$.
    Thus, any deviation that increases the miner's virtual utility existed even prior to updating $\Threshold_t(\hat{n})$, contradicting off-chain influence proofness of the original mechanism.
    
    Thus, $\Threshold_t(n) = \Threshold_t$ for all $n \in \N$, as required. 

\begin{remark}
    Suppose that a deterministic plain-text mechanism is simple-to-participate.
    In the proof of \autoref{thm:CharacterizeDetMechanism}, we argue that updating $\Threshold_t(\hat{n})$ to $\Threshold_t$ for all $\hat{n} \in \N$ for which the probability of allocating exactly $t$ users is zero for $\hat{v} \sim \Distr^{\hat{n}}$ does not violate off-chain influence proofness.
    It is not hard to see that it does not violate on-chain user and miner simplicity too.
    Updating $\Threshold_t(\hat{n})$ to $\Threshold_t$ keeps the allocation rule invariant, and thus, the payments remain unchanged too.
\end{remark}

\subsubsection{Proof of \autoref{thm:DecreasingMargBurns}} \label{sec:ProofofDecreasingMargBurns}

We will have to fill in the following details in the proof sketch:
(a) replicate the same proof even when $\MargThreshold_i = \sup \Distr$ for some $i$ and thus, there cannot exist $\vv(v_i) > \MargThreshold_i$ as required by our construction, (b) verify that fabricating $\hat{v}^{(t+1)}$ does not decrease the critical bids for users $(1), \dots, (t-1)$ and finally, (c) show the existence of a set of value profiles with a positive probability measure for which the miner can increase her revenue by fabricating a bid $\hat{v}^{(t+1)}$.

As in the proof sketch, let $t$ be the smallest index such that $\MargThreshold_t < \MargThreshold_{t+1}$.
Further, let $k$ be the largest index for which $\MargThreshold_1 \geq \dots \geq \MargThreshold_k \geq \sup \Distr$.
Remember that $\Threshold_k = \sum_{i = 1}^k \MargThreshold_i + \Threshold_0 \leq t \, \sup \Distr + \Threshold_0$, and thus, $\MargThreshold_1 = \dots = \MargThreshold_k = \sup \Distr$.
Further, $t > k$, since otherwise, $\MargThreshold_{t+1} > \MargThreshold_t = \sup \Distr$ and $\Threshold_{t+1} > t \, \sup \Distr + \Threshold_0$.

We will construct a family $V$ of value profiles for which on-chain miner simplicity will be violated.
For a small enough $\varepsilon$, consider values $v^{(1)}, \dots, v^{(k)}$ a little smaller than $\vv^{-1}(\MargThreshold_i) = \sup \Distr$, i.e, $\vv(v^{(i)}) \in [\MargThreshold_i - \tfrac{i}{k \,t} \varepsilon, \MargThreshold_i - \tfrac{i-1}{k \,t} \varepsilon]$.
For $k < i \leq t$, consider values $v^{(i)}$ such that $\vv(v^{(i)})$ is slightly larger than $\MargThreshold_i$, i.e, $\vv(v^{(i)}) \in [\MargThreshold_i + (1 + \frac{i-k-1}{t-k}) \varepsilon, \MargThreshold_i + (1 + \frac{i-k}{t-k}) \, \varepsilon]$.
Note that the probability of $v^{(1)}, \dots, v^{(t)}$ belonging to the corresponding intervals is strictly positive, which follows since $\Distr$ is a smooth regular distribution, and the virtual value function is continuous.

Next, we will argue that the virtual utility optimal allocation includes all $t$ users.
By construction, the virtual utility from allocation at most $k$ users is strictly smaller than building the empty block.
Conditioned on allocating more than $k$ users, the increase $\vv(v^{(i)})$ to the virtual surplus from allocating user $(i)$ is greater than the marginal burn $\MargThreshold_i$ for $k <  i \leq t$.
Thus, the allocation rule either includes all $t$ users or builds the empty block.
We have constructed the intervals for $\vv(v^{(1)}), \dots, \vv(v^{(t)})$ so that
$\sum_{i = 1}^t \Big(\vv(v^{(i)}) - \MargThreshold_i\Big) - \Threshold_0 > - \Threshold_0$, and thus, an off-chain influence proof allocation rule will allocate all $t$ users except with zero probability.

Now consider the miner fabricating a bid $\hat{v}^{(t+1)}$ just smaller than $v^{(t)}$, i.e, $\vv(\hat{v}^{(t+1)}) \in (\MargThreshold_t, \vv(v^{(t)}))$.
Note that $\vv(\hat{v}^{(t+1)}) < \MargThreshold_{t+1}$, and thus, the virtual utility optimal allocation will continue to include only the users $(1), \dots, (t)$ with virtual values $\vv(v^{(1)}), \dots, \vv(v^{(t)}) > \vv(\hat{v}^{(t+1)})$ even after the miner fabricates $\hat{v}^{(t+1)}$.
Therefore, the allocation rule and the total burn do not change as a consequence of fabricating $\hat{v}^{(t+1)}$. 

Finally, we argue that fabricating $\hat{v}^{(t+1)}$ will not decrease the critical bids of users $(1), \dots, (t-1)$ while strictly increasing the critical bid for user $(t)$.
Fix the values $\vec{v}_{-(i)}$ of users other than $(i)$.
Now suppose that user $(i)$ is included in the block for some value $v < v^{(i)}$ after the miner fabricates $\hat{v}^{(t+1)}$.
We will show that user $(i)$ will be included for the value profile $(v, \vec{v}_{-i})$ even before the miner fabricated $\hat{v}^{(t+1)}$, and thus, the critical bid cannot reduce from fabricating $\hat{v}^{(t+1)}$.
Note that the virtual utility from allocation all $(t+1)$ bids is dominated by including the $t$ largest bids for the value profile $(v, v_{-(i)}, \hat{v}^{(t+1)})$.
Thus, at most $t$ bids are included by the virtual utility maximizing allocation.
User $(i)$ is one such included user, by assumption and it cannot be the case that $\hat{v}^{(t+1)}$ is prioritized over a user $(j)$ since $\vv(v^{(j)}) > \vv(\hat{v}^{(t+1)})$.
The virtual utility optimal allocation is a subset of users $(1), \dots, (t)$, and thus, the same allocation is optimal even without the fabricated bid $\hat{v}^{(t+1)}$.
Hence, user $(i)$ is allocated for the value profile $(v, v_{-(i)})$, as required.

Arguing that user $(t)$'s critical bid strictly increases is much simpler.
Without $\hat{v}^{(t+1)}$, user $(t)$ is allocated whenever $v^{(t)} \geq \vv^{-1}(\MargThreshold_t)$.
However, once $\hat{v}^{(t+1)}$ is fabricated, user $(t)$'s value needs to be at least $\hat{v}^{(t+1)} > \vv^{-1}(\MargThreshold_t)$ for inclusion.
Thus, the critical bid and the payment charged from user $(t)$ strictly increase, contradicting on-chain miner simplicity.

To summarize, whenever there exists some $t$ for which the marginal burn is strictly increasing, the miner can increase her revenue with a positive probability by fabricating a bid after seeing the user's bids.
Thus, the sequence of marginal burns $\seq{\MargThreshold_t}{t \in \N}$ satisfies $\MargThreshold_t \geq \MargThreshold_{t+1}$ for a deterministic, simple-to-participate mechanism.

\subsubsection{Proof of \autoref{thm:IncreasingMargBurns}} \label{sec:ProofofIncreasingMargBurns}

We would have to fill in the following details from the proof sketch --- (a) there exists a set of values with a positive probability measure that satisfies the construction sketched above (for example, we will have to ensure details like $\hat{v}^{(n+1)} < v^{(n)}$ for each of the value profiles that we consider) and (b) the sum $\sum_{t+1}^n \vv^{-1}(\MargThreshold_i + \delta) - \MargThreshold_i$ indeed diverges.

Let $t \in \N$ be such that $\MargThreshold_t > \MargThreshold_{t+1}$.
Construct values $v^{(1)}, \dots, v^{(t)}$ as follows.
For a sufficiently small $\varepsilon$ and $\hat \varepsilon$, uniformly sample $t$ times from the interval $[\varepsilon, \varepsilon + \hat{\varepsilon}]$ and arrange them in ascending order to obtain $\varepsilon_1, \dots, \varepsilon_t$.
Choose $v^{(1)}, \dots, v^{(t)}$ so that $\vv(v^{(i)}) = \MargThreshold_i - \varepsilon_i$.
By continuity of the virtual value function, there exists such values and the probability measure of the set of such values is strictly positive.

For any $n \in \N$, we will next draw values $v^{(t+1)}, \dots, v^{(n)}$ as follows.
Similar to $\seq{\varepsilon_i}{1 \leq i \leq t}$, sample $(n-t)$ times uniformly from the interval $[\delta, \delta + \hat{\delta}]$ and arrange them in descending order to get $\delta_{t+1}, \dots, \delta_{n}$.
Define $v^{(i)}$ to have a virtual value $\vv(v^{(i)}) = \MargThreshold_i + \delta_i$.

We choose $\varepsilon, \hat \varepsilon, \delta$ and $\hat{\delta}$ such that (a) $(n-t) \, (\delta + \hat{\delta}) < t \varepsilon$ and (b) $t \, (\varepsilon + \hat{\varepsilon}) < (n-t) \, (\delta -\hat{\delta}) + \delta$.
It is not hard to see the existence of $\varepsilon, \hat \varepsilon, \delta$ and $\hat{\delta}$ that simultaneously satisfy both (a) and (b).
Specifically, choose $\hat \delta << \delta$ and $\hat \varepsilon << \varepsilon$ so that $t \varepsilon - (n-t) \, \delta$ is orders of magnitude smaller than both $\varepsilon$ and $\delta$.  

Condition (a) ensures that the virtual utility optimal allocation with bids $v^{(1)}, \dots, v^{(n)}$ is the empty block, since
\begin{align*}
    \notag
   \sum_{i = 1}^n \vv(v^{(i)}) &= \sum_{i = 1}^t \MargThreshold_i - \varepsilon_i + \sum_{i = t+1}^n \MargThreshold_i + \delta_i \\
   &\leq \sum_{i = 1}^n \MargThreshold_i - t \varepsilon + (n-t) \, (\delta + \hat{\delta}) \\
   &< \sum_{i = 1}^n \MargThreshold_i.
\end{align*}
The virtual utility $\sum_{i = 1}^n \vv(v^{(i)}) - \sum_{i = 1}^n \MargThreshold_i - \Threshold_0$ from including all $n$ users is at most the block reward $-\Threshold_0$ from creating an empty block. 

Now, consider the miner fabricating a bid $\hat{v}^{(n+1)}$ with a virtual value $\vv(\hat{v}^{(n+1)}) = \MargThreshold_{n+1} + \delta_{n+1}$ for $\delta_{n+1} \in \big( t \, (\varepsilon + \hat{\varepsilon}) - (n-t) \, \delta, t \, (\varepsilon + \hat{\varepsilon}) - (n-t) \, \delta + (n-t) \hat{\delta} \big)$.
Indeed, from condition (b), $\vv(\hat{v}^{(n+1)}) < \MargThreshold_{n+1} + t \, (\varepsilon + \hat{\varepsilon}) - (n - t) \, \delta + (n-t) \, \hat{\delta} \leq \MargThreshold_n + \delta \leq \vv(v^{(n)})$ and thus, $\hat{v}^{(n+1)}$ is the $n+1$\textsuperscript{th} largest bid.

The allocation rule should almost surely include all $n+1$ bids once the miner fabricates $\hat{v}^{(n+1)}$.
To prove this, first note that the virtual utility from building the empty block strictly dominates building a block with $\hat{n}$ users for $\hat{n} < n+1$.
Thus, it is sufficient to argue that including all $n+1$ bids dominates the empty block.
However, we have specifically chosen $\varepsilon_1, \dots, \varepsilon_t$ and $\delta_{t+1}, \dots, \delta_{n+1}$ such that allocating all $n+1$ bids has a higher virtual utility greater than the empty block.
\begin{align*}
    \notag
    \sum_{i = 1}^{n+1} \vv(v^{(i)}) - \sum_{i = 1}^{n+1} \MargThreshold_i - \Threshold_0 &= -\sum_{i = 1}^t (\varepsilon + \hat{\varepsilon} ) + \sum_{i = t+1}^n \delta + \delta_{n+1} - \Threshold_0 \\
    &> -\sum_{i = 1}^t (\varepsilon + \hat{\varepsilon} ) + \sum_{i = t+1}^n \delta + \big(t \, (\varepsilon + \hat{\varepsilon}) - (n-t) \, \delta\big) - \Threshold_0 \\
    &= - \Threshold_0.
\end{align*}

To summarize the proof thus far, the miner earns strictly positive payments from all $n$ users after fabricating $\hat{v}^{(n+1)}$, as opposed to the status quo, where the empty block was built by the TFM.
It only remains to argue that the miner's net revenue is larger than the block reward after fabricating $\hat{v}^{(t+1)}$.

We will start by calculating the critical bid of each user.
The critical bid $\vv(\overline{v}^{(i)})$ of user $(i)$ is the smallest value for which the virtual utility conditioned on allocating user $(i)$ is larger than the virtual utility conditioned on excluding user $(i)$.
We will show that the former allocation is the set of all $n+1$ bids and the latter is the empty block.
Consider any block which includes $\hat{n} < n$ users.
The virtual utility of such a block is maximized by allocating the $\hat{n}$ largest virtual values $\vv(v^{(1)}), \dots, \vv(v^{(\hat{n})})$.
But, by construction, we have $\sum_{j = 1}^{\hat{n}} \vv(j) < \sum_{i = 1}^{\hat{n}} \MargThreshold_j$ and thus, including $\hat{n} < n+1$ users is strictly dominated by creating an empty block.
Therefore, conditioned on excluding $(i)$, the empty block is built.
Conditioned on including $(i)$, unless all $(n+1)$ bids are included, building the empty block will dominate including user $(i)$.
Thus, it is sufficient to compare the virtual utility from including all $n+1$ bidders against that of the empty block to calculate bidder $(i)$'s critical bid.
User $(i)$'s critical bid satisfies
$$\sum_{j = 1, j \neq i}^{n} \vv(v^{(j)}) + \vv(\overline{v}^{(i)}) + \vv(\hat{v}^{(n+1)}) - \sum_{i = 1}^{n+1} \MargThreshold_j - \Threshold_0 = -\Threshold_0.$$

Substituting $\vv(v^{j}) = \MargThreshold_j - \varepsilon_j$ for $1 \leq j \leq t$ and $\vv(v^{j}) = \MargThreshold_j + \delta_j$ for $t+1 \leq j \leq n+1$, we get $\vv(\overline{v}^{(i)}) = \MargThreshold_i + \sum_{j = 1, j \neq i}^t \varepsilon_j - \sum_{j = t+1, j \neq i}^{n+1} \delta_j$.
Further, for $i > t$, 
\begin{align*}
    \notag
    \vv(\overline{v}^{(i)}) &\geq \MargThreshold_i +  \sum_{j = 1}^t \varepsilon - \sum_{j = t+1, j \neq i}^{n+1} (\delta + \hat{\delta}) - \delta_{n+1} \\
    &\geq \MargThreshold_i +  t \varepsilon - (n-t-1) (\delta + \hat{\delta}) - \Big(t \, (\varepsilon + \hat{\varepsilon}) - (n - t) \, (\delta - \hat{\delta})\Big) \\
    &\geq  \MargThreshold_i + \delta - t \hat{\varepsilon} - 2(n-t) \hat{\delta}.
\end{align*}
Remember that $\hat \varepsilon << \varepsilon$ and $\hat{\delta} << \delta$.
For such choices of $\delta, \hat \delta, \varepsilon$ and $\hat{\varepsilon}$, we have $\delta - t \hat{\varepsilon} - 2(n-t)\hat{\delta} > 0$, and thus, $\vv(\overline{v}^{i}) \geq \MargThreshold_i$ for $t+1 \leq i \leq n$.

Finally, we argue that for a large enough $n$, the payments received from the $n$ users after fabricating $\hat{v}^{(n+1)}$ compensates for the burn $\sum_{i = 1}^{n+1} \MargThreshold_i$, i.e,
$$\sum_{i = 1}^n \overline{v}^{(i)} > \sum_{i = 1}^{n+1} \MargThreshold_i.$$
In other words, we need
$$\sum_{i = 1}^n (\overline{v}^{(i)} - \MargThreshold_i) > \MargThreshold_{n+1}.$$

We will show that the series $\sum_{i = t+1}^{\infty} (\overline{v}^{(i)} - \MargThreshold_i)$ diverges.
Remember that $\overline{v}^{i}$ is at least $\vv^{-1}(\MargThreshold_i)$ equality holding if and only if $\MargThreshold_i = \sup \Distr$.
Also, the sequence of marginal burns $\seq{\MargThreshold}{i \in \N}$ is non-increasing, with $\sup \Distr \geq \MargThreshold_t > \MargThreshold_{t+1}$.
Thus, the sequence $\seq{\vv^{-1}(\MargThreshold_i) - \MargThreshold_i}{i \geq t+1}$ is strictly positive.
If this sequence is unbounded, then the series will diverge as required.
Supposing that the sequence is bounded, we will show that there exists a subsequence that does not converge to zero, which would automatically imply a divergent series.

When $\seq{\vv^{-1}(\MargThreshold_i) - \MargThreshold_i}{i \geq t+1}$ is bounded, by the Bolzano-Weierstrass theorem \citep{Rudin64}, there exists a convergent subsequence.
Relabel the index set of this subsequence to call it $\seq{\vv^{-1}(\MargThreshold_i) - \MargThreshold_i}{i \geq t+1}$.
The sequences $\seq{\MargThreshold_i}{i \geq t+1}$ and $\seq{\vv^{-1}(\MargThreshold_i)}{i \geq t+1}$ are both monotone non-increasing (the first follows since the marginal burns are non-increasing and we use the fact that the distribution $\Distr$ is regular for the second) bounded below by $\vv^{-1}(0)$ and $0$ respectively.
Thus, both the sequences must converge.
Let $\seq{\vv^{-1}(\MargThreshold_i)}{i \geq t+1}$ converge to $\upsilon$.
Since $\vv$ is continuous, $\seq{\MargThreshold_i}{i \geq t+1} = \seq{\vv \circ \vv^{-1}(\MargThreshold_i)}{i \geq t+1}$ must converge to $\vv(\upsilon)$.
Therefore, the sequence $\seq{\vv^{-1}(\MargThreshold_i) - \MargThreshold_i}{i \geq t+1}$ converges to $\upsilon - \vv(\upsilon) > 0$, as $\upsilon < \sup \Distr$.

To conclude, $\sum_{i > t+1} \vv^{-1}(\MargThreshold_i) - \MargThreshold_i$ diverges, and for a sufficiently large $n$,
$$\sum_{i = 1}^n (\overline{v}^{(i)} - \MargThreshold_i) \geq \sum_{i = 1}^t (\overline{v}^{(i)} - \MargThreshold_i) + \sum_{i = t+1}^n \vv^{-1}(\MargThreshold_i) - \MargThreshold_i)  > \MargThreshold_{n+1},$$
contradicting on-chain miner simplicity.
Thus, there cannot exist $t$ for which $\MargThreshold_t > \MargThreshold_{t+1}$, which in turn implies a constant marginal burn, as required.

\subsubsection{Proof Sketch of \autoref{thm:ConstMargBurnPerUnit}} \label{sec:ProofofConstMargBurnPerUnit}

In this proof sketch, we will locate a value profile $\vec{v} \in \supp(\Distr^n)$ for some $n \in \N$ for which the miner can increase her revenue by fabricating bids if $\MargThreshold_t$ is not a constant.
Ideas similar to the proofs of \autoref{thm:DecreasingMargBurns} and  \autoref{thm:IncreasingMargBurns} can be used to generalize our construction to obtain a set of value profiles with a positive probability measure.

\begin{lemma} \label{thm:DecreasingMarginalPosition}
    For a smooth regular distribution $\Distr$, a simple-to-participate position auction must satisfy $\MargThreshold_t \geq \MargThreshold_{t+1}$ for all $t \in \N$. 
\end{lemma}
We skip the proof for \autoref{thm:DecreasingMarginalPosition} since it is almost identical to \autoref{thm:DecreasingMargBurns}.

\begin{lemma}
    Let $\Distr$ be a smooth regular distribution.
    Then, a simple-to-participate position auction with non-increasing marginal burn per unit allocation must satisfy $\MargThreshold_t = \MargThreshold$ for all $t \in \N$.    
\end{lemma}
\begin{proof}[Proof sketch]
    We follow a strategy similar to \autoref{thm:IncreasingMargBurns}.
    For contradiction, let $t$ be the smallest index such that $\MargThreshold_t > \MargThreshold_{t+1}$.
    Consider a value profile $\vec{v} = (v^{(1)}, \dots, v^{(t)}, v^{(t+1)}, \dots, v^{(n)})$ for $n >> t$ such that
    $\vv(v^{(i)}) = \MargThreshold_i - \varepsilon_i$ for $1 \leq i \leq t$ and $\vv(v^{(i)}) = \MargThreshold_i + \delta_i$ for $t < i \leq n$.
    The virtual utility from allocating all $n$ bids equals
    $$\sum_{i = 1}^n \vv(v^{(i)}) \, x^{(i)} - \sum_{i = 1}^n \MargThreshold_i \, x^{(i)} - \Threshold_0 = - \sum_{i = 1}^t \varepsilon_i \, x^{(i)} +\sum_{i = t+1}^n \delta_i \, x^{(i)} - \Threshold_0.$$
    Choosing $- \sum_{i = 1}^t \varepsilon_i \, x^{(i)} +\sum_{i = t+1}^n \delta_i \, x^{(i)} < 0$, the virtual utility from allocation all $n$ users is dominated by the block reward from building an empty block.
    The virtual utility optimal allocation includes no user and the miner receives a revenue $- \Threshold_0$.

    Suppose that, for $\delta_{n+1} \, x^{(t+1)}$ very slightly larger than $\sum_{i = 1}^t \varepsilon_i \, x^{(i)} - \sum_{i = t+1}^n \delta_i \, x^{(i)}$, the miner fabricates $\hat{v}^{(n+1)} < v^{(n)}$ such that $\vv(\hat{v}^{(n+1)}) = \MargThreshold_{n+1} + \delta_{n+1}$.\footnote{It is fairly straightforward to choose $\varepsilon_1, \dots, \varepsilon_t, \delta_{t+1}, \dots, \delta_n$ such that $\hat{v}^{(n+1)}$ is the $n+1$\textsuperscript{th} largest bid. For example, choose $\varepsilon_1 = \dots, \varepsilon_t = \varepsilon$ and $\delta_{t+1} = \delta_{n} = \delta$. If $\delta_{n+1} > \delta$, decrease $\varepsilon$ and increase $\delta$ until $\delta_{n+1} < \delta$. $\varepsilon_1, \dots, \varepsilon_t, \delta_{t+1}, \dots, \delta_n, \delta_{n+1}$ can later be made arbitrarily small if necessary by scaling all of them by the same multiplicative constant.}\textsuperscript{,}\footnote{Note that such a choice of $\delta_{n+1}$ is not possible if $x^{(t+1)} = 0$. However, we can without loss of generality assume that $x^{(i)} > 0$ for all $i$. Setting $x^{(i)} = 0$ is similar to choosing a marginal threshold $\MargThreshold_i = \infty$ for any $x^{(i)} > 0$. The allocation rule will never allocate more than $i$ bids to avoid the exceptionally high marginal burn.}
    By construction, the virtual utility from allocating all $n+1$ bids is slightly larger than the block reward from the empty block.
    Formally,
    \begin{align*}
        \notag
        \sum_{i = 1}^n &\vv(v^{(i)}) \, x^{(i)} + \vv(\hat{v}^{(n+1)}) \, x^{(n+1)} - \sum_{i = 1}^{n+1} \MargThreshold_i \, x^{(i)} - \Threshold_0 \\
        &= \sum_{i = 1}^t (\MargThreshold_i - \varepsilon_i) \, x^{(i)} + \sum_{i = t+1}^n (\MargThreshold_i + \delta_i) \, x^{(i)} + (\MargThreshold_{n+1} + \delta_{n+1}) \, x^{(n+1)} - \sum_{i = 1}^{n+1} \MargThreshold_i \, x^{(i)} - \Threshold_0 \\
        &> -\Threshold_0
    \end{align*}
    Thus, the mechanism includes all $n$ users and the fabricated bid.
    
    By reasoning about critical bids similar to the proof of \autoref{thm:IncreasingMargBurns}, we can show that
    the critical bids $\overline{v}^{(i)}$ for users $1 \leq (i) \leq t$ satisfies $\vv(\overline{v}^{(i)}) \approx \MargThreshold_i - \varepsilon_i$ and for users $t+1 \leq (i) \leq n$, satisfies $\vv(\overline{v}^{(i)}) \geq \MargThreshold_i$.
    The $n$ users are charged a payment at least 
    \begin{align*}
        \notag
        \sum_{i = 1}^t \vv^{-1}(\MargThreshold_i - \varepsilon_i) \, x^{(i)} + \sum_{i = t+1}^n \vv^{-1}(\MargThreshold_i) \, x^{(i)}
    \end{align*}
    while $\sum_{i = 1}^{n+1} \MargThreshold_i \, x^{(i)}$ is burnt.
    The miner's net revenue, therefore, is greater than
    \begin{align*}
        \notag
        \sum_{i = t+1}^n \Big(\overline{v}^{(i)} - \Threshold_i \Big) &\, x^{(i)} + \sum_{i = t+1}^n \Big(\vv^{-1}(\MargThreshold_i) - \MargThreshold_i \Big) \, x^{(i)} - \MargThreshold_{n+1} \, x^{(n+1)} \\
        &\geq \sum_{i = t+1}^n \Big(\vv(\overline{v}^{(i)}) - \MargThreshold_i \Big) \, x^{(i)} + \sum_{i = t+1}^n \Big(\vv^{-1}(\MargThreshold_i) - \MargThreshold_i \Big) \, x^{(i)} - \MargThreshold_{n+1} \, x^{(n+1)} \\
        &\approx -\sum_{i = t+1}^n \varepsilon_i \, x^{(i)} + \sum_{i = t+1}^n \Big(\vv^{-1}(\MargThreshold_i) - \MargThreshold_i \Big) \, x^{(i)} - \MargThreshold_{n+1} \, x^{(n+1)}.
    \end{align*}
    For a sufficiently small $\varepsilon_1, \dots, \varepsilon_t$ and a sufficiently large $n$, $\sum_{i = t+1}^n \varepsilon_i \, x^{(i)}$ becomes arbitrarily close to zero while $\sum_{i = t+1}^n \Big(\vv^{-1}(\MargThreshold_i) - \MargThreshold_i \Big) \, x^{(i)}$ strictly dominates $\MargThreshold_{n+1} \, x^{(n+1)}$, contradicting on-chain miner simplicity.
\end{proof}

\subsubsection{Proof of \autoref{thm:BoundedPositionAuction}} \label{sec:ProofofBoundedPositionAuction}

As a first step towards the ``if'' direction, we will argue that if the miner can increase her revenue by fabricating a collection of bids for some value profile $\vec{v}$, then, there exists a value profile for which the miner can increase her revenue by fabricating exactly one bid smaller than all other bids.

\begin{lemma} \label{thm:PositionFakeBids}
    For a position auction specified by $\seq{x^{(t)}}{t \in \N}$ and $\MargThreshold$, suppose there exists a value profile for which fabricating bids increases the miner's revenue.
    Then, there exists another value profile for which the miner can increase her revenue by fabricating exactly one bid smaller than all bids submitted by the users.
    In such a scenario, if the miner fabricates a bid $w$ with $t$ users bidding greater than $w$, the miner's net revenue increases by
    $$t \, (w - \vv^{-1}(\MargThreshold)) \, (x^{(t)} - x^{(t + 1})) - \MargThreshold x^{(t+1)}.$$
\end{lemma}
\begin{proof}
    First, we will argue that the miner can increase her revenue by fabricating exactly one bid and next, we will argue that there exists some value profile for which the miner can increase her revenue by fabricating a bid smaller than all the users' bids.

\begin{claim} \label{thm:OneBid}
    For a position auction specified by $\seq{x^{(t)}}{t \in \N}$ and $\MargThreshold$, suppose there exists a value profile for which fabricating bids increases the miner's revenue.
    Then, there exists another value profile for which the miner can increase her revenue by fabricating exactly one bid.    
\end{claim}
\begin{proof}
    Suppose that, for a value profile $\vec{v}$ with users $(1), \dots, (n)$, the miner increases here revenue by fabricating bids $\vec{w} = (w^{(1)}, w^{(2)}, \dots, w^{(\hat{n})})$.
    Imagine that the miner fabricates the bids $w^{(1)}, w^{(2)}, \dots, w^{(\hat{n})}$ one-by-one.
    We will consider the change in the miner's net revenue when she inserts $w^{(i)}$ after already inserting $w^{(1)}, \dots, w^{(i-1)}$.
    Note that the increase in the miner's net revenue upon fabricating all $\hat{n}$ of the bids in $\vec{w}$ is positive, and thus, there exists some $(i)$ for which the miner's revenue increases upon fabricating $w^{(i)}$ conditioned on already fabricating $w^{(1)}, w^{(2)}, \dots, w^{(i-1)}$.
    Call this increase $\Delta \Rev$.    

    Now, suppose there are $n + (i-1)$ users, $(1), \dots, (n)$ with values $\vec{v}$ and $(\overline{1}), \dots, (\overline{i-1})$ with values $(w^{(1)}, \dots, w^{(i-1)})$.
    Consider the miner fabricating $w^{(i)}$.
    Note that the change in payments made by users $(1), \dots, (n)$ minus the marginal burn from fabricating $w^{(i)}$ is exactly equal to $\Delta \Rev$.
    Remember that the bids $w^{(1)}, \dots, w^{(i-1)}$ are no longer fabricated bids and belong to the users $(\overline{1}), \dots, (\overline{i-1})$.
    The payments made by these users will also change when the miner fabricates $w^{(i)}$.
    However, we have $w^{(j)} \geq w^{(i)}$ for $1 \leq j < i$, and fabricating $w^{(i)}$ only increases the payments made by the users $(\overline{1}), \dots, (\overline{i-1})$.

    Thus, the miner can increase her revenue by fabricating the single bid $w^{(i)}$ on seeing the value profile $(\vec{v}, w^{(1)}, \dots, w^{(i-1)})$.
\end{proof}

    Next, suppose there exists a value profile $\vec{v} \in \supp(\Distr^n)$ for which the miner increases her revenue by fabricating a single bid $w$.
    Without loss of generality, assume all values are at least $\vv^{-1}(\MargThreshold)$.
    Values smaller than $\vv^{-1}(\MargThreshold)$ will anyways not be allocated, and thus, can be dropped without changing the miner's revenue.
    We will show the existence of a different value profile with values all greater than $w$ for which the miner can increase her revenue by fabricating $w$.

    To compute the change in the miner's revenue from fabricating $w$, we will separately calculate the increase in payments from users $(1), \dots, (t)$ and the users $(t+1), \dots, (n)$, where $t$ is the number of users with a value at least $w$.
    The former are the set of users who will be allocated the same quantity irrespective of whether $w$ is fabricated and the latter are the users whose rank will decrease by $1$ due to $w$ being fabricated, thereby receiving a smaller allocation.

    For convenience, for the remainder of the section, we will define $v^{(n+1)} := \vv^{-1}(\MargThreshold)$.

\begin{claim} \label{thm:LargerThanw}
    For a position auction given by an allocation rule $\seq{x^{(i)}}{i \in \N}$ and marginal burn per allocation $\MargThreshold$, consider a profile of users' values $\vec{v}$.
    Suppose the miner inserts a fake bid $w$ with $t$ users having a value at least $w$.
    Then, the increase in the payments collected from users $(1), \dots, (t)$ equals
    \begin{align*}
        \notag
        \Delta \Pay^{\leq (t)} = t \times \Bigg[(w - v^{(t+1)}) \, (x^{(t)} - x^{(t+1)}) + \sum_{i = t+1}^{n} (v^{(i)} - v^{(i+1)}) \, (x^{(i)} - x^{(i+1)}) \Bigg].
    \end{align*}
\end{claim}
\begin{proof}
We will calculate the payments charged to users $(1), \dots, (t)$ with and without the fake bid $w$ using the payment identity (\autoref{item:payment-identity} in \autoref{thm:myerson}).
The allocation rules before and after fabricating $w$ are sketched in \autoref{fig:PositionAuction}.

\begin{figure}
    \centering
    \begin{subfigure}[b]{0.4\textwidth}
        \includegraphics[width=0.8\linewidth]{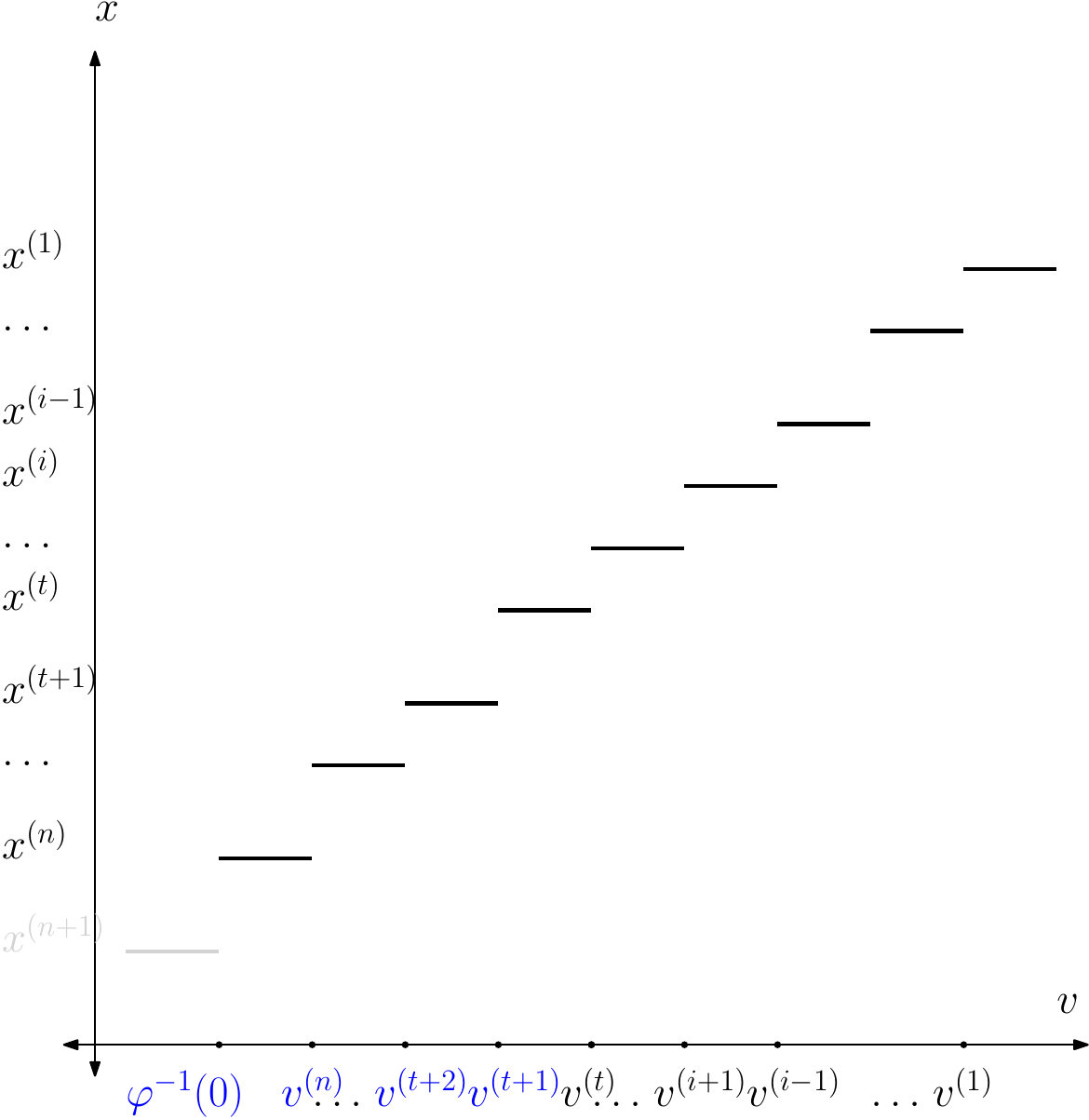}
        \caption{}
        \label{fig:PositionAuctionWithoutw}
    \end{subfigure}
        \begin{subfigure}[b]{0.4\textwidth}
        \includegraphics[width=0.8\linewidth]{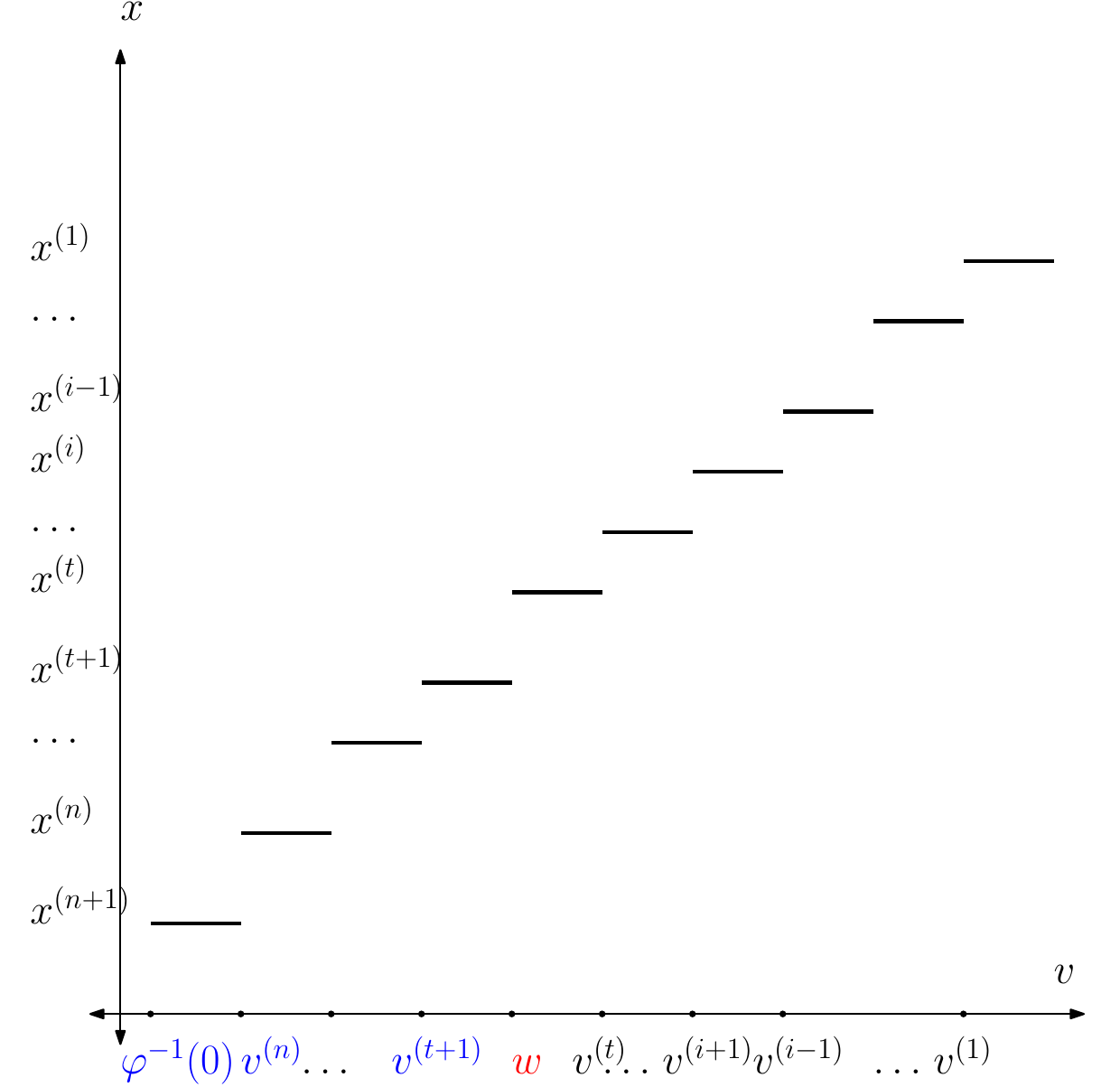}
        \caption{}
        \label{fig:PositionAuctionWithw}
    \end{subfigure}   
    \caption{\footnotesize The allocation rule for user $(i)$ in a position auction with and without a fabricated bid $w$. For a value profile $\vec{v}_{-(i)}$, \autoref{fig:PositionAuctionWithoutw} denotes the allocation to user $(i)$ as a function of his bid $v$. For $w$ such that $v^{(t+1)} \leq w \leq v^{(t)}$, \autoref{fig:PositionAuctionWithw} denotes the allocation rule for user $(i)$ after the miner fabricates a bid $w$. Note the change in allocation probabilities when user $(i)$ bids $v \in [\vv^{-1}(0), v^{(t)}]$.}
    \label{fig:PositionAuction}
\end{figure}

Prior to the miner inserting $w$, user $(j)$'s allocation rule for values $v \leq v^{(j)}$ is given by
\begin{equation*}
    \notag
    \OnCAlloc_{(j)}(v, \vec{v}_{-(j)}) = \begin{cases}
        0 & \text{ for } v \leq \vv^{-1}(\MargThreshold) = v^{(n+1)}, \\
        x^{(i)} & \text{ for } v \in [v^{(i+1)}, v^{(i)}), i \geq j.
    \end{cases} 
\end{equation*}
By the payment identity, user $(j)$ is charged a payment
\begin{align}
    \OnCPay_{(j)}(v, \vec{v}_{-(j)}) &= v^{(n+1)} \, x^{(n)} + \sum_{i = j}^{n-1} v^{(i+1)} \, (x^{(i)} - x^{(i+1)}) \label{eqn:OnMSPayment1}\\
    &= v^{(n+1)} \, x^{(n+1)} + \sum_{i = j}^{n} v^{(i+1)} \, (x^{(i)} - x^{(i+1)}). \label{eqn:OnMSPayment2}
\end{align}

After the bid $w$ is fabricated, the allocation rule for user $(j)$ changes as follows.
\begin{align*}
    \notag
    \OnCAlloc_{(j)}(v, \vec{v}_{-(j)}, w) = \begin{cases}
        0 & \text{ for } v \leq \vv^{-1}(\MargThreshold) = v^{(n+1)}, \\
        x^{(i)} & \text{ for } v \in [v^{(i)}, v^{(i-1)}), i > t+1. \\
        x^{(t+1)} & \text{ for } v \in [v^{(t+1)}, w), \\
        x^{(t)} & \text{ for } v \in [w, v^{(t)}), \\
        x^{(i)} & \text{ for } v \in [v^{(i+1)}, v^{(i)}), j \leq i \leq t-1.
    \end{cases} 
\end{align*}
Applying the payment identity once again, we get
\begin{align*}
    \notag
    \OnCPay_{(j)}(v, \vec{v}_{-(i)}, w) &= v^{(n+1)} \, x^{(n+1)} + \sum_{i = t+1}^n v^{(i)} \, (x^{(i)} - x^{(i+1)}) + w \, (x^{(t)} - x^{(t+1)}) + \sum_{i = j}^t v^{(i+1)} \, (x^{(i)} - x^{(i+1)}).
\end{align*}

Subtracting the two payments, the payment charged to user $(j)$ increases by
\begin{align*}
    \notag
    \OnCPay_{(j)}(v, \vec{v}_{-(j)}, w) - \OnCPay_{(j)}(v, \vec{v}_{-(j)}) &= (w - v^{(t+1)}) \, (x^{(t)} - x^{(t+1)}) + \sum_{i = t+1}^n (v^{(i)} - v^{(i+1)}) \, (x^{(i)} - x^{(i+1)}).
\end{align*}

Thus, the increase in payments from the users $(1), \dots, (t)$ with values larger than $w$ equals $t \times \big(\OnCPay_{(j)}(v, \vec{v}_{-(j)}, w) - \OnCPay_{(j)}(v, \vec{v}_{-(j)})\big)$, as required.
\end{proof}

\begin{claim} \label{thm:SmallerThanw}
    For a position auction given by an allocation rule $\seq{x^{(i)}}{i \in \N}$ and marginal burn per allocation $\MargThreshold$, consider a profile of users' values $\vec{v}$.
    Suppose the miner inserts a fake bid $w$ with $t$ users having a value at least $w$.
    Then, the (possibly negative) change in the payments collected from users $(t+1), \dots, (n)$ equals
    \begin{align*}
        \notag
        \Delta \Pay^{> (t)} = \sum_{i = t+2}^{n} (i - t- 1) \, (v^{(i)} - v^{(i+1)}) \, (x^{(i)} - x^{(i+1)}) - \sum_{i = t+1}^n v^{(i+1)} \, (x^{(i)} - x^{(i+1)}).
    \end{align*}
\end{claim}
\begin{proof}
    Similar to \autoref{thm:LargerThanw}, we will compute the payments for each user $(j)$ for $j > t$ with and without the fabricated bid $w$ using the payment identity (\autoref{item:payment-identity} in \autoref{thm:myerson}).

    By \autoref{eqn:OnMSPayment2}, the payment charged to user $(j)$ before fabricating $w$ equals
    \begin{align*}
        \notag
        \OnCPay_{(j)}(v, \vec{v}_{-(j)}) &= v^{(n+1)} \, x^{(n+1)} + \sum_{i = j}^{n} v^{(i+1)} \, (x^{(i)} - x^{(i+1)}).
    \end{align*}
    Post fabricating $w$, user $(j)$'s rank shifts below by $1$.
    Thus, the allocation rule becomes
    \begin{equation*}
    \notag
    \OnCAlloc_{(j)}(v, \vec{v}_{-(j)}, w) = \begin{cases}
        0 & \text{ for } v \leq \vv^{-1}(\MargThreshold) = v^{(n+1)}, \\
        x^{(i+1)} & \text{ for } v \in [v^{(i+1)}, v^{(i)}), i \geq j.
    \end{cases} 
    \end{equation*}
    User $(j)$'s payment then becomes (note that we use \autoref{eqn:OnMSPayment1} here)
    \begin{align*}
        \notag
        \OnCPay_{(j)}(v, \vec{v}_{-(j)}, w) &= v^{(n+1)} \, x^{(n+1)} + \sum_{i = j+1}^{n} v^{(i)} \, (x^{(i)} - x^{(i+1)}).
    \end{align*}

    Thus, the change in user $(j)$'s payment equals
    \begin{align*}
        \notag
        \OnCPay_{(j)}(v, \vec{v}_{-(j)}, w) - \OnCPay_{(j)}(v, \vec{v}_{-(j)}) &= \sum_{i = j+1}^n (v^{(i)} - v^{(i+1)}) \, (x^{(i)} - x^{(i+1)}) - v^{(j+1)} \, (x^{(j)} - x^{(j+1)}).
    \end{align*}

    While summing over all users $(j)$ for $t+1 \leq j \leq n$, we separately collect the $(v^{(i)} - v^{(i+1)}) \, (x^{(i)} - x^{(i+1)})$ terms (which features exactly once in the summation for all $i \geq j+1$) and the $- v^{(j+1)} \, (x^{(j)} - x^{(j+1)})$ terms separately. 
    Thus,
    \begin{align*}
        \notag
        \Delta \Pay^{> (t)} = \sum_{i = t+2}^{n} (i - t- 1) \, (v^{(i)} - v^{(i+1)}) \, (x^{(i)} - x^{(i+1)}) - \sum_{i = t+1}^n v^{(i+1)} \, (x^{(i)} - x^{(i+1)}).
    \end{align*}
\end{proof}

To summarize our progress so far, we have computed the changes in the users' payments from fabricating $w$.
We can combine \autoref{thm:LargerThanw} and \autoref{thm:SmallerThanw} to calculate the change $\Delta \Rev$ in the miner's net revenue.
We also subtract $\MargThreshold x^{(n+1)}$, the excess burn for including the fake bid $w$.
\begin{align}
    \notag
    \Delta \Rev &= \underbrace{t \, (w - v^{(t+1)}) \, (x^{(t)} - x^{(t+1)}) + \sum_{i = t+1}^n (i - 1) \, (v^{(i)} - v^{(i+1)}) \, (x^{(i)} - x^{(i+1)})}_{\mathsf{:= Core}} \label{eqn:CombinedRev}\\
    & \qquad - \underbrace{\Bigg[\sum_{i = t+1}^n v^{(i+1)} \, (x^{(i)} - x^{(i+1)}) + \MargThreshold x^{(n+1)}\Bigg]}_{\mathsf{:=Residue}}
\end{align}

We would like to compare $\Delta \Rev$ against the change in revenue $\Delta \Rev(\hat{t})$ when the miner fabricates a bid $w$ upon seeing $\hat{t}$ bids larger than $w$ for some $\hat{t} \geq t$.
We will start by explicitly computing $\Delta \Rev(\hat{t})$.

\begin{claim} \label{thm:NetRevSingleBid}
    Consider a position auction given by the allocation rule $\seq{x^{(i)}}{i \in \N}$ and marginal burn per allocation $\MargThreshold$.
    Then, the miner's change in revenue $\Delta \Rev(\hat{t})$ from fabricating $w$ on seeing $\hat{t}$ bids all larger than $w$ equals
    $$\hat{t} \, (w - \vv^{-1}(\MargThreshold)) \, (x^{(\hat{t})} - x^{(\hat{t} + 1})) - \MargThreshold \, x^{(\hat{t} + 1)}.$$
\end{claim}
\begin{proof}
Follows directly from \autoref{eqn:CombinedRev} setting $t = \hat{t} = n$.
\end{proof}

We are ready to prove \autoref{thm:PositionFakeBids}.
\autoref{thm:PositionFakeBids} reduces to proving 
$$\max_{\hat t \geq t} \Delta \Rev(\hat t) \geq \Delta \Rev.$$

We will start by lower bounding $\mathsf{Residue}$ by $\MargThreshold \, x^{(\hat{t} + 1)}$ for all $\hat{t}$.
\begin{align*}
    \notag
    \mathsf{Residue} &= \sum_{i = t+1}^n v^{(i+1)} \, (x^{(i)} - x^{(i+1)}) + \MargThreshold x^{(n+1)} \\
    &\geq \sum_{i = t+1}^n v^{(n+1)} \, (x^{(i)} - x^{(i+1)}) + \MargThreshold x^{(n+1)} \\
    &\geq \sum_{i = t+1}^n \MargThreshold \, (x^{(i)} - x^{(i+1)}) + \MargThreshold x^{(n+1)} \\
    &= \MargThreshold x^{(t+1)}.
\end{align*}
The third inequality follows since $v^{(n+1)} = \vv^{-1}(\MargThreshold) \geq \MargThreshold$.
Since $\hat{t} \geq t$ and $\seq{x^{(i)}}{i \in \N}$ is a decreasing sequence,
\begin{align}
    \mathsf{Residue} \geq \MargThreshold x^{(t+1)} \geq \MargThreshold x^{(\hat t+1)}. \label{eqn:Residue}
\end{align}

We will next prove that $\mathsf{Core}$ is at most $\max_{\hat{t} \geq t} \hat{t} \, (w - \vv^{-1}(\MargThreshold)) \, (x^{(\hat{t})} - x^{(\hat{t} + 1}))$.
Observe that
\begin{align*}
    \notag
    \mathsf{Core} &= t \, (w - v^{(t+1)}) \, (x^{(t)} - x^{(t+1)}) + \sum_{i = t+1}^{n} (i-1) \, (v^{(i)} - v^{(i +1)}) \, (x^{(i)} - x^{(i+1)}) \\
    &\leq t \, (w - v^{(t+1)}) \, (x^{(t)} - x^{(t+1)}) + \sum_{i = t+1}^{n} i \, (v^{(i)} - v^{(i +1)}) \, (x^{(i)} - x^{(i+1)}) \\
    &= (w - v^{(t+1)}) \times \Big(t \, (x^{(t)} - x^{(t+1)}) \Big) + \sum_{i = t+1}^{n} (v^{(i)} - v^{(i+1)}) \times \Big(i \, (x^{(i)} - x^{(i+1)}) \Big)
\end{align*}
Consider maximizing the right hand side.
Imagine we have a total weight $(w - v^{(n+1)})$, and want to distribute it amongst the terms in the sequence $\seq{i \, (x^{(i)} - x^{(i+1)})}{t \leq i \leq n}$ so as to maximize the weighted sum.
We choose $v^{(t)}, \dots, v^{(n)}$ so that the coefficient of $i \, (x^{(i)} - x^{(i+1)})$ in the weighted sum will be equal to $(v^{(i)} - v^{(i+1)})$ (or $w - v^{(t+1)}$ if $i =t$).
The right hand side is maximized by putting all of the $(w - v^{(n+1)})$ on $\max_{\hat{t} \geq t} \hat{t} \, (x^{(\hat{t})} - x^{(\hat{t} + 1}))$ and no weight on the rest, in which case, we get
\begin{align}
    \mathsf{Core} \leq \max_{\hat{t} \geq t} \hat{t} \, (w - \vv^{-1}(\MargThreshold)) \, (x^{(\hat{t})} - x^{(\hat{t} + 1})). \label{eqn:Core}
\end{align}
Subtracting \autoref{eqn:Residue} from \autoref{eqn:Core} concludes the proof of \autoref{thm:PositionFakeBids}.

To conclude, we have identified a value profile for which the miner can increase her revenue by fabricating exactly one bid smaller than all bids submitted by the users.
\end{proof}

We have proved that if there exists some value profile for which the miner will want to fabricate bids, then there exists a profile for which the miner will want to fabricate exactly one bid.
To prove on-chain miner simplicity, we will also have to make sure the miner does not censor any bids.
However, it is fairly straightforward to see that censoring a bid will only decrease the competition to get allocated and thus, only reduce the miner's revenue.

From \autoref{thm:PositionFakeBids} and \autoref{thm:NetRevSingleBid},
the mechanism is on-chain miner simple if 
$$t \, (x^{(t)} - x^{(t+1)}) \, (w - \vv^{-1}(\MargThreshold)) - \MargThreshold \, x^{(t+1)} \leq 0$$
for all $w \in \supp(\Distr)$ and $t \in \N$.
Indeed, the above is true exactly when
$$t \, (x^{(t)} - x^{(t+1)}) \, (\sup \Distr - \vv^{-1}(\MargThreshold)) - \MargThreshold \, x^{(t+1)} \leq 0,$$
concluding the ``if'' direction.

The ``only if'' direction is straightforward from \autoref{thm:NetRevSingleBid}.
If $t \, (x^{(t)} - x^{(t+1)}) \, (\sup \Distr - \vv^{-1}(\MargThreshold)) > \MargThreshold \, x^{(t+1)}$ for some $t \in \N$, consider a value profile with $t$ users with values approximately $\sup \Distr$.
The miner can increase her revenue by fabricating a bid $\hat{v}^{(t+1)}$ just below $\sup \Distr$.

\subsubsection{Proof of \autoref{thm:UnboundedPositionAuction}} \label{sec:ProofofUnboundedPositionAuction}

By an argument similar to \autoref{thm:OneBid} in the proof of \autoref{thm:BoundedPositionAuction}, if there exists a value profile for which the miner can increase her revenue by fabricating bids, then there exists a value profile for which the miner can increase her revenue by fabricating exactly one bid.

Similar to \autoref{thm:PositionFakeBids}, if the generalized position auction satisfies the conditions in \autoref{thm:UnboundedPositionAuction}, we will show that there must also exist a value profile $\vec{v}$ for which the miner can increase her revenue by fabricating a bid smaller than all values in $\vec{v}$.

\begin{claim}
    Suppose that a generalized position auction with an allocation rule $\seq{x^{(t)}}{t \in \N}$ satisfies
    $$t \, \Big(x^{(t)}(w) - x^{(t+1)}(w) \Big) \geq (t+1) \, \Big(x^{(t+1)}(w) - x^{(t+2)}(w) \Big)$$
    for all $t \in \N$.
    Then, if there exists a value profile $\vec{v}$ such that the miner can increase her revenue by fabricating a single bid, then there exists a value profile for which the miner can increase her revenue by fabricating a bid smaller than the values of all the users that have submitted a bid.
\end{claim}
\begin{proof}
For a value profile $\vec{v}$ with $t$ bids larger than $w$, we can compute the change in the payments charged to the users upon fabricating a bid $w$.
\begin{align*}
    \OnCPay(\vec{v}, w) - \OnCPay(\vec{v}) &= \int_{v^{(t+1)}}^w t \, \Big(x^{(t)}(z) - x^{(t+1)}(z)\Big) \,dz + \sum_{i = t+1}^n \int_{v^{(i+1)}}^{v^{(i)}} t \, \Big(x^{(i)}(z) - x^{(i+1)}(z)\Big) \,dz \\
    &\qquad \qquad - \sum_{i = t+1}^n v^{(i)} \, \Big(x^{(i)}(v^{(i)}) - x^{(i+1)}(v^{(i)})\Big).
\end{align*} 
Remember that we follow the convention $v^{(n+1)} = \vv^{-1}(0)$.
Similarly, the increase in burn from fabricating $w$ equals
\begin{align*}
    \OnCBurn(\vec{v}, w) - \OnCBurn(\vec{v}) &= \vv(w) \, x^{(t+1)}(w) - \int_{0}^{\vv(w)} x^{(t)}(z) \, d\vv(z) - \sum_{i = t+1}^n \vv(v^{(i)}) \, \Big(x^{(i)}(v^{(i)}) - x^{(i+1)}(v^{(i)})\Big) \\
    &\qquad \qquad + \sum_{i = t+1}^{n} \int_{\vv(v^{(i+1)})}^{\vv(v^{(i)})} \Big(x^{(i)}(z) - x^{(i+1)}(z)\Big) \, d\vv(z)
\end{align*}
The net change in the miner's revenue equals
\begin{align*}
    \Delta \Rev(\vec{v}, w) &:= \underbrace{\int_{v^{(t+1)}}^w t \, \Big(x^{(t)}(z) - x^{(t+1)}(z)\Big) \,dz + \sum_{i = t+1}^n \int_{v^{(i+1)}}^{v^{(i)}} t \, \Big(x^{(i)}(z) - x^{(i+1)}(z)\Big) \,dz}_{:= \mathsf{Core}} \\
    &- \Bigg[\vv(w) \, x^{(t+1)}(w) - \int_{0}^{\vv(w)} x^{(t)}(z) \, d\vv(z) + \sum_{i = t+1}^n \Big(v^{(i)} - \vv(v^{(i)}) \Big) \, \Big(x^{(i)}(v^{(i)}) - x^{(i+1)}(v^{(i)})\Big) \\
    & \underbrace{\qquad \qquad \qquad \qquad \qquad + \sum_{i = t+1}^{n} \int_{\vv(v^{(i+1)})}^{\vv(v^{(i)})} \Big(x^{(i)}(z) - x^{(i+1)}(z)\Big) \, d\vv(z)\Bigg]. \hspace{3cm}}_{:= \mathsf{Residue}}
\end{align*}
We will compare $\Delta \Rev(\vec{v}, w)$ against the increase in the miner's revenue from fabricating a bid $w$ for the value profile $\vec{v}|^t = (v^{(1)}, \dots, v^{(t)})$.
Remember that we chose $\vec{v}$ and $w$ so that exactly $t$ users had a value larger than $w$.
The increase in burn $\OnCBurn(\vec{v}|^t, w) - \OnCBurn(\vec{v}|^t)$ from fabricating $w$ equals
$$\vv(w) \, x^{(t+1)}(w) - \int_{0}^{\vv(w)} x^{(t)}(z) \, d\vv(z) \leq \mathsf{Residue}.$$
The increase in payments $\OnCPay(\vec{v}|^t, w) - \OnCPay(\vec{v}|^t)$ equals
\begin{align*}
    \int_{v^{(n+1)}}^w t \, \Big(x^{(t)}(z) - x^{(t+1)}(z)\Big) \,dz &\geq
    \int_{v^{(t+1)}}^w t \, \Big(x^{(t)}(z) - x^{(t+1)}(z)\Big) \,dz + \sum_{i = t+1}^n \int_{v^{(i+1)}}^{v^{(i)}} t \, \Big(x^{(i)}(z) - x^{(i+1)}(z)\Big) \,dz,
\end{align*}
where the inequality follows since 
$t \, \Big(x^{(t)}(w) - x^{(t+1)}(w) \Big) \geq \hat{t} \, \Big(x^{(\hat{t})}(w) - x^{(\hat{t})}(w) \Big)$ for all $\hat{t} \geq t$.

To summarize, for the value profile $\vec{v}$, if the miner can increase her revenue by fabricating a bid $w$, then, she can also increase her revenue by fabricating $w$ for the value profile $\vec{v}|^t$, in which case, $w$ will be the smallest bid submitted to the mechanism.
\end{proof}

\autoref{thm:UnboundedPositionAuction} follows immediately.
The mechanism is on-chain miner simple if the increase in the payments collected by the miner from fabricating a bid $w$ is at most the increase in burn, i.e,
$$t \times \int_{\vv^{-1}(0)}^w \Big(x^{(t)}(z) - x^{(t+1)}(z) \Big) \,dz \leq \vv(w) \, x^{(t+1)}(w) - \int_0^{\vv(w)} x^{(t+1)}(z) \,d\vv(z)$$
for all $t \in \N$ and $w \geq \vv^{-1}(0)$.

\subsubsection{Proof of \autoref{thm:GenPosFinite}} \label{sec:ProofofGenPosFinite}
    For cleanliness, we provide a proof sketch assuming the burn rule as given in \autoref{eqn:BurnGenPos} holds for all value profiles $\vec{v}$, and not just with probability $1$.
    However, it is fairly easy to extend our proof even when the burn identity is satisfied almost surely.

    If possible, consider a non-trivial simple-to-participate generalized position auction for a block with a capacity $\Omega < \infty$.
    Since the mechanism is non-trivial, for some $W \in \R$, $x^{(1)}(w) > 0$ for all $w \geq W$.
    In our proof, we will show that for $T \in \N$ much larger than $\frac{\Omega}{x^{(1)}(W)}$, there exists $w \geq W$ such that $x^{(T)}(w) \approx x^{(1)}(w) \geq x^{(T)}(W)$.
    Then, for a value profile $\vec{v} \in \supp(\Distr^T)$ for which all $T$ bids are larger than $w$, by monotonicity, the users receive a total allocation of at least 
    $$\sum_{t = 1}^T x^{(t)}(w) \geq T \times x^{(T)}(w) \approx T \times x^{(1)}(w) > \Omega,$$ violating the feasibility constraint of the block.

    To find $w$ such that $x^{(1)}(w) \approx x^{(T)}(w)$, we will use on-chain miner simplicity to show that for any small $\delta > 0$, there exists a large $w \geq W$ such that 
    \begin{align}
        \Big(x^{(t)}(w) - x^{(t+1)}(w)\Big) < \delta, \label{eqn:GenPosLimit}
    \end{align}
    for all $1 \leq t \leq T-1$.
    In other words, $x^{(T)}(w) = x^{(1)}(w) - \sum_{t = 1}^{T-1} \Big(x^{(t)}(w) - x^{(t+1)}(w)\Big) \approx x^{(1)}(w)$.

    Consider a value profile with $t$ users, all of them with a value greater than $w \geq W$.
    Suppose that the miner inserts a fake bid equal to $w$.
    From the calculations in \autoref{thm:UnboundedPositionAuction}, the above deviation is not profitable if and only if
    \begin{align}
        t \int_{v^{(t+1)}}^w \Big(x^{(t)}(z) - x^{(t+1)}(z)\Big) \,dz \leq \vv(w) \, x^{(t+1)}(w) - \int_{0}^{\vv(w)} x^{(t+1)}(z) \,d \vv(z). \label{eqn:OnMSGenPos} 
    \end{align}
    Remember that we denote $\vv^{-1}(0)$ by $v^{(t+1)}$.

    We will prove the claim in \autoref{eqn:GenPosLimit} by showing that the left hand side in \autoref{eqn:OnMSGenPos} is a sub-linear function in $w$, and thus, the derivative $t \, \big(x^{(t)}(z) - x^{(t+1)}(z)\big)$ has to be arbitrarily small for some $z$.

\begin{claim} \label{thm:SublinearBurn}
    For any constant $\eta > 0$, there exists $w_{\eta} > 0$ such that for all $w \geq w_{\eta}$,
    $$t \int_{v^{(t+1)}}^w \Big(x^{(t)}(z) - x^{(t+1)}(z)\Big) \,dz < w \times \eta.$$    
\end{claim}
\begin{proof}
Supposing that $\vv(w)$ is bounded, the claim follows immediately from \autoref{eqn:OnMSGenPos}.
$t \int_{v^{(t+1)}}^w \big(x^{(t)}(w) - x^{(t+1)}(z)\big) \,dz$ is bounded above by a constant, and is thus at most $w \times \eta$ for a sufficiently large $w$.

Now suppose that $\vv(w)$ is unbounded.
We will then prove
$$\vv(w) \, x^{(t+1)}(w) - \int_{0}^{\vv(w)} x^{(t+1)}(z) \,d \vv(z) = \int_{0}^{\vv(w)} \Big(x^{(t+1)}(w) - x^{(t+1)}(z) \Big) \,d \vv(z) < \vv(w) \times \eta,$$
and the claim is a direct consequence from \autoref{eqn:OnMSGenPos} and since $\vv(w) \leq w$.

Remember that $x^{(t+1)}(\cdot)$ is a monotone non-decreasing function and is bounded above by $1$, and must converge as $w \xrightarrow{} \infty$.
Let the limit point be $x^{(t+1)} \in [0, 1]$.
We will write $\vv(w) \times \eta$ as
\begin{align*}
        \vv(w) \times [x^{(t+1)} - (x^{(t+1)} - \eta)] = \int_0^{\vv(w)} x^{(t+1)} - (x^{(t+1)} - \eta) \,d\vv(z)
\end{align*}
By monotonicity, $x^{(t+1)} \geq x^{(t+1)}(z)$ for all $z$.
Further, for sufficiently large $\hat z$, $x^{(t+1)}(z) > x^{(t+1)} - \eta/2$ for all $z \geq \hat{z}$.
Thus,
\begin{align*}
    \int_0^{\vv(w)} x^{(t+1)} - (x^{(t+1)} - \eta) \,d\vv(z) &\geq \int_{\vv(\hat{z})}^{\vv(w)} \Big(x^{(t+1)}(w) - x^{(t+1)}(z) + \eta/2 \Big) \, dz \\
    &\geq \int_{\vv(\hat{z})}^{\vv(w)} \Big(x^{(t+1)}(w) - x^{(t+1)}(z) \Big) \, dz + \int_{0}^{\vv(\hat{z})} \Big(x^{(t+1)}(w) - x^{(t+1)}(z) \Big) \, dz 
\end{align*}
for a large enough $\vv(w)$.
Thus,
$$\vv(w) \, x^{(t+1)}(w) - \int_{0}^{\vv(w)} x^{(t+1)}(z) \,d \vv(z) < \vv(w) \times \eta.$$

\end{proof}

We are ready to prove the claim in \autoref{eqn:GenPosLimit}.

\begin{claim}
For any $\delta > 0$, there exists an arbitrarily large $w$ such that
$$\Big(x^{(t)}(w) - x^{(t+1)}(w)\Big) < \delta$$
for all $1 \leq t \leq T-1$.
\end{claim}
\begin{proof}
    Assume otherwise, that for some $\hat{W}$ and all $w \geq \hat{W}$, $\big(x^{(t)}(w) - x^{(t+1)}(w)\big) \geq \delta$ for some $t \in [1, T-1]$.
    Then, there must exist some $t$ for which $\big(x^{(t)}(w) - x^{(t+1)}(w)\big) \geq \delta$ in more than $\frac{1}{T}$ fraction of values in the interval $[\hat{W}, 10 \hat{W}]$.
    Therefore,
    \begin{align*}
        t \int_{v^{(t+1)}}^w \Big(x^{(t)}(z) - x^{(t+1)}(z)\Big) \,dz &\geq t \int_{\hat{W}}^{10 \hat{W}} \Big(x^{(t)}(z) - x^{(t+1)}(z)\Big) \,dz \\
        &\geq \hat{W} \times \frac{9t}{T}\delta.
    \end{align*}
    The above inequality holds for some $1 \leq t \leq T-1$ for any sufficiently large $\hat{W}$.
    In particular, for $\eta = \frac{9t}{T}\delta$, \autoref{thm:SublinearBurn} is violated for some $t$.
    Contradiction.
    Thus, there must exist some large $w$ for which
    $\big(x^{(t)}(w) - x^{(t+1)}(w)\big) < \delta$
    for all $1 \leq t \leq T-1$.
\end{proof}
We are ready to conclude the proof of \autoref{thm:GenPosFinite}.
For any small $\delta > 0$ and $T > \frac{\Omega}{x^{(1)}(W)}$, where $\Omega$ is the capacity of the block and $W$ is some value which receives a positive allocation,
there exists $w > W$ such that $\big(x^{(t)}(w) - x^{(t+1)}(w)\big) < \delta$ for all $t \in [1, T-1]$.
Thus,
$$x^{(T)}(w) = x^{(1)}(w) - \sum_{t = 1}^{T-1} \big(x^{(t)}(w) - x^{(t+1)}(w)\big) > x^{(1)}(w) - (T-1) \delta.$$
By monotonicity, we have
$$\sum_{t = 1}^T x^{(t)}(w) \geq T \times x^{(T)}(w) \geq T \times \Big(x^{(1)}(w) - (T-1) \, \delta \Big).$$
Setting $\delta = \frac{x^{(1)}(w)}{2(T-1)}$ and $T > \frac{2\Omega}{x^{(1)}(w)}$, we have
$$\sum_{t = 1}^T x^{(t)}(w) > \Omega,$$
violating the capacity constraint.
This is a contradiction; hence, the auction must be trivial.

%% file: Bib.bib
@inproceedings{GZ25,
  title={Truthful, Credible, and Optimal Auctions for Matroids via Blockchains and Commitments},
  author={Ganesh, Aadityan and Zhang, Qianfan},
  booktitle={Proceedings of the 26th ACM Conference on Economics and Computation},
  pages={923--943},
  year={2025}
}

@inproceedings{GafniY24Discrete,
  title={Discrete and Bayesian Transaction Fee Mechanisms},
  author={Gafni, Yotam and Yaish, Aviv},
  booktitle={The International Conference on Mathematical Research for Blockchain Economy},
  pages={145--171},
  year={2024},
  organization={Springer}
}

@inproceedings{ChungS23,
author = {Hao Chung and Elaine Shi},
title = {Foundations of Transaction Fee Mechanism Design},
booktitle = {Proceedings of the 2023 Annual ACM-SIAM Symposium on Discrete Algorithms (SODA)},
year={2023},
pages = {3856-3899},
}

@incollection{EyalS14,
	Author = {Eyal, Ittay and Sirer, Emin G{\"u}n},
	Booktitle = {Financial Cryptography and Data Security},
	Pages = {436--454},
	Publisher = {Springer},
	Title = {Majority is not enough: Bitcoin mining is vulnerable},
	Year = {2014}}

@article{DhangwatnotaiRY15,
	Author = {Peerapong Dhangwatnotai and Tim Roughgarden and Qiqi Yan},
	Bibsource = {dblp computer science bibliography, https://dblp.org},
	Doi = {10.1016/j.geb.2014.03.011},
	Journal = {Games and Economic Behavior},
	Pages = {318--333},
	Title = {Revenue maximization with a single sample},
	Url = {https://doi.org/10.1016/j.geb.2014.03.011},
	Volume = {91},
	Year = {2015}}

@article{BasuEOS19,
	Archiveprefix = {arXiv},
	Author = {Soumya Basu and David Easley and Maureen O'Hara and Emin G{\"{u}}n Sirer},
	Bibsource = {dblp computer science bibliography, https://dblp.org},
	Eprint = {1901.06830},
	Journal = {CoRR},
	Title = {Towards a Functional Fee Market for Cryptocurrencies},
	Url = {http://arxiv.org/abs/1901.06830},
	Volume = {abs/1901.06830},
	Year = {2019}}

@inproceedings{garay2013rational,
  author={Garay, Juan and Katz, Jonathan and Maurer, Ueli and Tackmann, Björn and Zikas, Vassilis},
  booktitle={FOCS}, 
  title={Rational Protocol Design: Cryptography against Incentive-Driven Adversaries}, 
  year={2013},
  volume={},
  number={},
  pages={648-657},
 }

@inproceedings{halpern2004rational,
    author = {Halpern, Joseph and Teague, Vanessa},
    title = {Rational secret sharing and multiparty computation: extended abstract},
    year = {2004},
    booktitle = {STOC},
    pages = {623–632},
}

@inproceedings{asharov2011game,
    author = {Gilad Asharov and Ran Canetti and Carmit Hazay},
    title = {Towards a Game Theoretic View of Secure Computation},
    booktitle = {Eurocrypt},
    year = {2011},
    pages = {426--445}
}

@inproceedings{groce2012rationalBA,
    author = {Adam Groce and Jonathan Katz and Aishwarya Thiruvengadam and Vassilis Zikas},
    title = {Byzantine Agreement with a Rational Adversary},
    booktitle = {ICALP},
    year = {2012},
    pages = {561--572}
}

@misc{gong2025collusion,
      author = {Tiantian Gong and Aniket Kate and Hemanta K. Maji and Hai H. Nguyen},
      title = {Disincentivize Collusion in Verifiable Secret Sharing},
      howpublished = {Cryptology {ePrint} Archive, Paper 2025/446},
      year = {2025},
      url = {https://eprint.iacr.org/2025/446}
}

@inproceedings{KGPBW25,
  title={Breaking Omert{\`a}: On Threshold Cryptography, Smart Collusion, and Whistleblowing},
  author={Kelkar, Mahimna and Ganesh, Aadityan and Partap, Aditi and Bonneau, Joseph and Weinberg, S Matthew},
  booktitle={Proceedings of the 2025 ACM SIGSAC Conference on Computer and Communications Security},
  pages={3505--3519},
  year={2025}
}

@inproceedings{LaviSZ19,
	Author = {Ron Lavi and Or Sattath and Aviv Zohar},
	Bibsource = {dblp computer science bibliography, https://dblp.org},
	Booktitle = {The World Wide Web Conference, {WWW} 2019, San Francisco, CA, USA, May 13-17, 2019},
	Crossrefignore = {DBLP:conf/www/2019},
	Doi = {10.1145/3308558.3313454},
	Pages = {2950--2956},
	Title = {Redesigning Bitcoin's fee market},
	Url = {https://doi.org/10.1145/3308558.3313454},
	Year = {2019}}

@article{Yao18,
	Archiveprefix = {arXiv},
	Author = {Andrew Chi{-}Chih Yao},
	Bibsource = {dblp computer science bibliography, https://dblp.org},
	Eprint = {1811.02351},
	Journal = {CoRR},
	Title = {An Incentive Analysis of some Bitcoin Fee Designs},
	Url = {http://arxiv.org/abs/1811.02351},
	Volume = {abs/1811.02351},
	Year = {2018}}

@article{ButerinRLP19,
	Archiveprefix = {arXiv},
	Author = {Vitalik Buterin and Dani{\"{e}}l Reijsbergen and Stefanos Leonardos and Georgios Piliouras},
	Bibsource = {dblp computer science bibliography, https://dblp.org},
	Eprint = {1903.04205},
	Journal = {CoRR},
	Title = {Incentives in Ethereum's Hybrid Casper Protocol},
	Url = {http://arxiv.org/abs/1903.04205},
	Volume = {abs/1903.04205},
	Year = {2019}}

@inproceedings{GaneshTW24,
  title={Revisiting the Primitives of Transaction Fee Mechanism Design},
  author={Ganesh, Aadityan and Thomas, Clayton and Weinberg, S Matthew},
  booktitle={Proceedings of the 25th ACM Conference on Economics and Computation},
  pages={703--703},
  year={2024}
}

@article{Roughgarden20,
  title={Transaction fee mechanism design for the Ethereum blockchain: An economic analysis of EIP-1559},
  author={Roughgarden, Tim},
  journal={arXiv preprint arXiv:2012.00854},
  year={2020}
}

@inproceedings{GaneshHSvM24,
  title={Fundamental Limits of Throughput and Availability: Applications to prophet inequalities and transaction fee mechanism design},
  author={Ganesh, Aadityan and Hartline, Jason D and Sinha, Atanu R and vonAllmen, Matthew},
  booktitle={Proceedings of the 25th ACM Conference on Economics and Computation},
  pages={108--135},
  year={2024}
}

@book{Rudin64,
  title={Principles of mathematical analysis},
  author={Rudin, Walter},
  volume={3},
  year={1964},
  publisher={McGraw-hill New York}
}

@article{Rochet85,
  title={The taxation principle and multi-time Hamilton-Jacobi equations},
  author={Rochet, Jean-Charles},
  journal={Journal of Mathematical Economics},
  volume={14},
  number={2},
  pages={113--128},
  year={1985},
  publisher={Elsevier}
}

@article{Myerson81,
	Author = {Roger B. Myerson},
	Journal = {Mathematics of Operations Research},
	Number = {1},
	Pages = {58-73},
	Title = {{Optimal Auction Design}},
	Volume = {6},
	Year = {1981}}

@article{HartlineBook,
  title={Mechanism design and approximation},
  author={Hartline, Jason D},
  journal={Book draft.},
  volume={122},
  number={1},
  year={2013}
}

@article{roughgarden2010algorithmic,
  title={Algorithmic game theory},
  author={Roughgarden, Tim},
  journal={Communications of the ACM},
  volume={53},
  number={7},
  pages={78--86},
  year={2010},
  publisher={ACM New York, NY, USA}
}

@article{Roughgarden21,
  title={Transaction fee mechanism design},
  author={Roughgarden, Tim},
  journal={ACM SIGecom Exchanges},
  volume={19},
  number={1},
  pages={52--55},
  year={2021},
  publisher={ACM New York, NY, USA},
  note={Full version at \url{https://arxiv.org/abs/2106.01340}}
}

@inproceedings{DrakopoulosLM23,
author = {Drakopoulos, Kimon and Lo, Irene and Mulvany, Justin},
title = {Blockchain Mediated Persuasion},
year = {2023},
booktitle = {Proceedings of the 24th ACM Conference on Economics and Computation},
location = {London, United Kingdom},
series = {EC '23}
}

@inproceedings{FerreiraW20,
  title={Credible, truthful, and two-round (optimal) auctions via cryptographic commitments},
  author={Ferreira, Matheus VX and Weinberg, S Matthew},
  booktitle={Proceedings of the 21st ACM Conference on Economics and Computation},
  pages={683--712},
  year={2020}
}

@article{EssaidiFW22,
  title={Credible, strategyproof, optimal, and bounded expected-round single-item auctions for all distributions},
  author={Essaidi, Meryem and Ferreira, Matheus VX and Weinberg, S Matthew},
  journal={arXiv preprint arXiv:2205.14758},
  year={2022}
}

@inproceedings{BabaioffDOZ12,
  title={On bitcoin and red balloons},
  author={Babaioff, Moshe and Dobzinski, Shahar and Oren, Sigal and Zohar, Aviv},
  booktitle={Proceedings of the 13th ACM conference on electronic commerce},
  pages={56--73},
  year={2012}
}

@article{BahraniGR24,
  title={Transaction fee mechanism design in a post-mev world},
  author={Bahrani, Maryam and Garimidi, Pranav and Roughgarden, Tim},
  journal={Cryptology ePrint Archive},
  year={2024}
}

@article{ADM24,
  title={Multidimensional Blockchain Fees are (Essentially) Optimal},
  author={Angeris, Guillermo and Diamandis, Theo and Moallemi, Ciamac},
  journal={arXiv preprint arXiv:2402.08661},
  year={2024}
}

@inproceedings{FMDPS21,
  title={Dynamic posted-price mechanisms for the blockchain transaction-fee market},
  author={Ferreira, Matheus VX and Moroz, Daniel J and Parkes, David C and Stern, Mitchell},
  booktitle={Proceedings of the 3rd ACM Conference on Advances in Financial Technologies},
  pages={86--99},
  year={2021}
}

@article{FerreiraGR24,
  title={Incentive-Compatible Collusion-Resistance via Posted Prices},
  author={Ferreira, Matheus VX and Gafni, Yotam and Resnick, Max},
  journal={arXiv preprint arXiv:2412.20853},
  year={2024}
}

@article{BabaioffN24,
  title={On the Optimality of EIP-1559 for Patient Bidders (Draft--Comments Welcome)},
  author={Babaioff, Moshe and Nisan, Noam},
  year={2024}
}

@article{ChungWS25,
  title={Foundations of Platform-Assisted Auctions},
  author={Chung, Hao and Wu, Ke and Shi, Elaine},
  journal={arXiv preprint arXiv:2501.03141},
  year={2025}
}

@article{GafniY22,
  title={Greedy Transaction Fee Mechanisms for (Non-) myopic Miners},
  author={Gafni, Yotam and Yaish, Aviv},
  journal={arXiv preprint arXiv:2210.07793},
  year={2022}
}

@article{GafniY24,
  title={Barriers to collusion-resistant transaction fee mechanisms},
  author={Gafni, Yotam and Yaish, Aviv},
  journal={EC 2024},
  note={arXiv preprint arXiv:2402.08564},
  year={2024}
}

@article{ChungRS24,
  title={Collusion-resilience in transaction fee mechanism design},
  author={Chung, Hao and Roughgarden, Tim and Shi, Elaine},
  journal={EC 2024},
  note={arXiv preprint arXiv:2402.09321},
  year={2024}
}

@article{ChenSZZ24,
  title={Bayesian mechanism design for blockchain transaction fee allocation},
  author={Chen, Xi and Simchi-Levi, David and Zhao, Zishuo and Zhou, Yuan},
  journal={arXiv preprint arXiv:2209.13099},
  year={2024}
}

@article{ShiCW23,
  title={What can cryptography do for decentralized mechanism design},
  author={Shi, Elaine and Chung, Hao and Wu, Ke},
  journal={Innovations in Theoretical Computer Science (ITCS)},
  year={2023}
}

@misc{ChitraFK23,
      author = {Tarun Chitra and Matheus V. X. Ferreira and Kshitij Kulkarni},
      title = {Credible, Optimal Auctions via Blockchains},
      howpublished = {Cryptology ePrint Archive, Paper 2023/114},
      year = {2023},
      url = {https://eprint.iacr.org/2023/114}
}

@article{WuSC24,
  title={Maximizing Miner Revenue in Transaction Fee Mechanism Design},
  author={Wu, Ke and Shi, Elaine and Chung, Hao},
  journal={ITCS},
  year={2024}
}


%% file: sample-bibliography.bib
@String{Computer = "{IEEE} Computer" }

@String{Springer = "Springer-Verlag" }
